%% file: mhl.tex
\documentclass[acmsmall,screen,dvipsnames]{acmart}
 \pdfoutput=1 

\setcopyright{rightsretained}
\acmPrice{}
\acmDOI{10.1145/3586045}
\acmYear{2023}
\copyrightyear{2023}
\acmSubmissionID{oopslaa23main-p89-p}
\acmJournal{PACMPL}
\acmVolume{7}
\acmNumber{OOPSLA1}
\acmArticle{93}
\acmMonth{4}
\received{2022-10-28}
\received[accepted]{2023-02-25}

\makeatletter
\def\@acmplainindent{0pt}
\def\@acmdefinitionindent{0pt}
\def\@proofindent{\noindent}
\makeatother

\usepackage{microtype}

\AtBeginDocument{%
  \providecommand\BibTeX{{%
    \normalfont B\kern-0.5em{\scshape i\kern-0.25em b}\kern-0.8em\TeX}}}

\usepackage{graphicx}				
					
\usepackage{amsfonts}
\usepackage{amsmath}
\usepackage{amsthm}
\usepackage{pifont}
\usepackage{cancel}
\usepackage[inline]{enumitem}
\usepackage{eqparbox}
\usepackage{hyperref}
\usepackage{listings}
\usepackage{mathpartir}
\usepackage{mathtools}
\usepackage{multicol}
\usepackage{multirow}
\usepackage{relsize}
\usepackage{scalerel}
\usepackage{stmaryrd}
\usepackage{thm-restate}
\usepackage{tikz}
\usepackage{wrapfig}
\usepackage{xfrac}
\usepackage{xspace}

\usepackage[all, cmtip]{xy}
\usepackage{xr}

\usetikzlibrary{decorations.pathreplacing,calligraphy}

\theoremstyle{remark}
\newtheorem{remark}{Remark}

\usepackage{cleveref}

\newcommand{\PP}{\mathbb{P}}

\newcommand{\de}[1]{\left\llbracket #1 \right\rrbracket}
\newcommand{\dem}[2]{\left\llbracket #1 \right\rrbracket^\dagger\!\!(#2)}
\newcommand{\sem}[1]{\llparenthesis #1 \rrparenthesis}

\newcommand{\samp}{
  \xleftarrow{\raisebox{-.35em}[0ex][0ex]{\tiny \ \$}}}
\newcommand{\skp}{\textsf{skip}}
\newcommand{\iftf}[3]{\mathsf{if}~#1~\mathsf{then}~#2~\mathsf{else}~#3}
\newcommand{\whl}[2]{\textsf{while}~#1~\textsf{do}~#2}

\newcommand{\bind}{\mathsf{bind}}
\newcommand{\unit}{\mathsf{unit}}

\newcommand{\monoid}{\diamond}

\newcommand{\hb}[1]{{\color{purple}\{#1\}}} 
\newcommand{\ob}[1]{{\color{purple}\langle #1\rangle}} 
\newcommand{\ib}[1]{{\color{purple}[#1]}} 
\newcommand{\lb}[1]{{\color{purple}\{\!| #1 |\!\}}} 
\newcommand{\triple}[3]{\ob{#1} ~#2~ \ob{ #3}}
\newcommand{\hoare}[3]{\hb{#1} ~#2~ \hb{#3}}
\newcommand{\lisbon}[3]{\lb{#1}~#2~\lb{#3}}
\newcommand{\inc}[3]{\ib{#1} ~#2~ \ib{ #3}}

\newcommand{\ok}{\mathsf{ok}}
\newcommand{\er}{\mathsf{er}}

\newcommand{\emp}{\mathsf{emp}}

\newcommand{\sep}{\ast}
\newcommand{\wand}{\;{\mathlarger{{-} \mathrel{\mkern -16mu} {-}\mathrel{\mkern -16mu}\ast}}\;}

\newcommand{\brackwidth}{.66em}

\newcommand{\mktripple}[2]{
\vDash
{\color{purple}\makebox[\brackwidth]{$#1$} P \makebox[\brackwidth]{$#2$}}
~C~
{\color{purple}\makebox[\brackwidth]{$#1$} Q \makebox[\brackwidth]{$#2$}}
}

\newcommand{\wpre}{\textsf{wp}}

\newcommand{\wpe}{\textsf{wpe}}

\DeclareMathAlphabet{\dsser}{U}{DSSerif}{m}{n}
\newcommand{\bb}[1]{\scalebox{1.15}{{$\dsser{#1}$}}}

\newcommand{\ident}{\varnothing}
\newcommand{\inj}{\dsser{i}}
\newcommand{\zero}{\bb{0}}
\newcommand{\com}{\fatslash\kern-.5em\fatslash}
\newcommand{\ie}{\emph{i.e.,} }
\newcommand{\eg}{\emph{e.g.,} }

\newcommand{\cites}[1]{\citeauthor{#1}'s~[\citeyear{#1}]}

\newcommand{\assume}[1]{\mathsf{assume}~#1}
\newcommand{\tru}{\mathsf{true}}
\newcommand{\fls}{\mathsf{false}}
\newcommand{\prop}{\mathsf{Prop}}

\newcommand{\vDashD}{\vDash^\downarrow\!}

\newcommand{\iffcases}[2]{
\begin{itemize}
\item[$(\Rightarrow)$] #1
\item[$(\Leftarrow)$] #2
\end{itemize}
}
\newcommand{\iffpf}[2]{
\begin{proof}\;
\iffcases{#1}{#2}
\end{proof}
}

\def\extended{0}

\ifx\extended\undefined
\newcommand{\Aref}[1]{\citet[\S\ref*{#1}]{appendix}}
\newcommand{\Acite}[1]{\cite[\S\ref*{#1}]{appendix}}
\newcommand{\Acref}[1]{\cite[\Cref*{#1}]{appendix}}
\else
\newcommand{\Aref}[1]{\Cref{#1}}
\newcommand{\Acite}[1]{(see \Cref{#1})}
\newcommand{\Acref}[1]{(\Cref{#1})}
\fi

\ifx\extended\undefined
\fi

\begin{document}

\ifx\extended\undefined
\title{Outcome Logic: A Unifying Foundation for Correctness and Incorrectness Reasoning}
\else
\title{Outcome Logic: A Unifying Foundation for Correctness and Incorrectness Reasoning (Full Version)}
\fi

\author{Noam Zilberstein}
\email{noamz@cs.cornell.edu}
\orcid{0000-0001-6388-063X}
\affiliation{%
  \institution{Cornell University}
  \country{USA}
}

\author{Derek Dreyer}
\email{dreyer@mpi-sws.org}
\orcid{0000-0002-3884-6867}
\affiliation{%
  \institution{MPI-SWS}
  \city{Saarland Informatics Campus}
  \country{Germany}
}

\author{Alexandra Silva}
\email{alexandra.silva@cornell.edu}
\orcid{0000-0001-5014-9784}
\affiliation{%
  \institution{Cornell University}
  \country{USA}
}


\begin{abstract}
Program logics for bug-finding (such as the recently introduced Incorrectness Logic) have framed correctness and incorrectness as dual concepts requiring different logical foundations.  In this paper, we argue that a single unified theory can be used for both correctness and incorrectness reasoning.
We present Outcome Logic (OL), a novel generalization of Hoare Logic that is both \emph{monadic} (to capture computational effects) and \emph{monoidal} (to reason about outcomes and reachability). 
OL expresses true positive bugs, while retaining \emph{correctness} reasoning abilities as well. 
To formalize the applicability of OL to both correctness and incorrectness, we prove that any false OL specification can be disproven in OL itself.
We also use our framework to reason about new types of incorrectness in nondeterministic and probabilistic programs. Given these advances, we advocate for OL as a new foundational theory of correctness and incorrectness.
\end{abstract}



\maketitle

\begin{flushright}
\emph{``Program correctness and incorrectness are two sides of the same coin.''} -- \citet{il}
\end{flushright}

%

\section{Introduction}
\label{sec:intro}


Developing formal methods to prove program \emph{correctness}---the \emph{absence} of bugs---has been a holy grail in program logic and static analysis research for many decades.
However, seeing as many static analyses deployed in practice are \emph{bug-finding} tools, \citet{il} recently advocated for the development of formal methods for proving program {\em incorrectness}; we need expressive, efficient, and compositional ways to reliably identify the \emph{presence} of bugs as well.

The aforementioned paper of \citet{il} proposed Incorrectness Logic (IL) as a logical foundation for reasoning about program incorrectness.
IL is inspired by---and in a precise technical sense \emph{dual to}---Hoare Logic. Like Hoare Logic, IL specifications are \emph{compositional}, given in terms of preconditions $P$ and postconditions $Q$. Hoare Triples $\hoare PCQ$ stipulate that the result of running the program $C$ on any state satisfying $P$ will be a state that satisfies $Q$. Incorrectness Triples $\inc PCQ$ \emph{go in reverse}---all states satisfying $Q$ must be reachable from some state satisfying $P$.
\[
\begin{array}{llllllcl}
\textsf{\textbf{Hoare Logic:}} & \mktripple{\{}{\}} & \text{iff} & \forall {\sigma\vDash P}. & \forall \tau. & \tau\in\de C(\sigma) & \Rightarrow & {\tau\vDash Q} \\
\textsf{\textbf{Incorrectness Logic:}} & \mktripple{[}{]} & \text{iff} & \forall {\tau\vDash Q}. & \exists\sigma. & \tau \in\de C(\sigma) & \text{and} & {\sigma\vDash P}
\end{array}
\]
Practically speaking, IL differs from Hoare Logic in two key ways. First, whereas Hoare Logic has no false negatives (\ie \emph{all} executions of a verified program behave correctly), IL has \textit{no false positives}: any bug found using IL is in fact reachable by \emph{some} execution of the program.
Second, whereas Hoare Logic is over-approximate, IL is \textit{under-approximate}: to prove that a program is incorrect, one only needs to specify (in the postcondition) a \emph{subset} of the possible outcomes, which helps to ensure the efficiency of large-scale analyses.
Subsequent work has focused on extending IL to account for a variety of program errors (\eg memory errors, memory leaks, data races, and deadlocks) and on using the resulting Incorrectness Separation Logics (ISLs) to explain and inform the development of bug-catching static analyses~\cite{isl,cisl,realbugs}.

Despite these exciting advances, we argue that the foundations of incorrectness reasoning are still far from settled---and worthy of reconsideration. IL achieves \textit{true positives} (reachability of end-states) and \emph{under-approximation} through the same mechanism: quantification over all states that satisfy the postcondition. However, this conflation of concepts leads to several problems: 

\medskip
\noindent\textbf{\em Expressivity.}
The semantics of IL only encompasses under-approximate types of incorrectness, which does not fully account for all bugs that may be encountered in real programs. For example, as we will see in \Cref{sec:ofls}, IL can be used to show the \emph{reachability} of bad states, but it cannot prove \emph{unreachability} of good states. 

\medskip
\noindent\textbf{\em Generality.}
IL is not amenable to probabilistic execution models and therefore is not a good fit for reasoning about incorrectness in randomized programs (\Cref{sec:prob-incorrectness}).

\medskip
\noindent\textbf{\em Error Reporting.}
IL cannot easily describe what conditions are \emph{sufficient} to trigger a bug (\Cref{sec:manifest}), meaning that analyses based on IL must implement extra algorithmic checks to determine whether a bug is worth reporting~\cite{realbugs}.

\medskip
\noindent Our key insight is that reachability and under-approximation are separate concepts that can (and should) be handled independently. But once reachability is separated from under-approximation, the resulting program logic no longer applies only to bug-finding.
 In this paper, we show how the full spectrum of correctness and incorrectness reasoning can be achieved with a unified foundation: a generalization of ``good old'' Hoare Logic that we call \textbf{\em Outcome Logic (OL)}.
In addition to consolidating the foundations of incorrectness with traditional correctness reasoning, OL overcomes all the aforementioned drawbacks of IL.

In OL, assertions are no longer predicates over program states, but rather predicates on an \emph{outcome monoid}, whose elements can be, for instance, \emph{sets of program states} or \emph{probability distributions on program states}.
The monoidal structure enables us to model a new \emph{outcome conjunction}, $P \oplus Q$, asserting that the predicates $P$ and $Q$ each hold in \emph{reachable} executions (or hold in subdistributions on program executions). We can also under-approximate by joining a predicate with $\top$, the trivial outcome: $P\oplus\top$ states that $P$ only partially covers the program outcomes.
OL offers several advantages as a unifying foundation for correctness and incorrectness:

\smallskip
\noindent\textbf{\emph{Generality.}} OL unifies program analysis across two dimensions. First, since any untrue OL spec can be disproven in OL (\Cref{thm:falsification}),  correctness and incorrectness reasoning are possible in a single program logic. Second, OL uses a monadic semantics which allows it to be instantiated for different evaluation models such as nondeterminism, erroneous termination, and probabilistic choice, thereby unifying correctness and incorrectness reasoning across execution models. 

\smallskip
\noindent\textbf{\emph{Beyond Reachability.}}
Until now, the study of incorrectness has revolved primarily around \textit{reachability} of crash states.
We prove that OL handles a broader characterization of incorrectness than IL in nondeterministic programs (\Cref{thm:ndfalse}), as well as probabilistic incorrectness (\Cref{thm:probfalse}).

\smallskip
\noindent\textbf{\emph{Manifest Errors.}}
In order to improve fix rates in automated bug finding tools, \citet{realbugs} only report bugs that occur regardless of context.
These bugs---called \emph{manifest errors}---are not straightforward to characterize using Incorrectness Logic: an auxiliary algorithm is needed to check whether some bug is truly a manifest error. In contrast, manifest errors are trivial to characterize in OL---\cites{realbugs} original definition can be expressed as an OL triple (\Cref{lem:manifest}).

\smallskip
\noindent
The contributions of the paper are as follows: 
\begin{itemize}[leftmargin=0.5cm]
\item We provide an overview of the semantics of IL and explain what is needed in order to characterize broader classes of errors (\Cref{sec:overview}). We show how reasoning about outcomes can account for reachability of end-states and enable under-approximation (when desired).
\item We define Outcome Logic formally (\Cref{sec:prelim} and \Cref{sec:logic}), parametric on a monad and an assertion logic. We define syntax and semantics of the logic, using Bunched Implications (BI) formulae for pre- and postconditions, and provide inference rules to reason about validity. 
\item We show that OL is suitable for both correctness and incorrectness reasoning by proving that false OL triples can be disproven within OL (\Cref{sec:falsification}). As a corollary, OL can disprove Hoare triples, which was one motivation for IL  (\Cref{thm:mhlil}). We go further and show three kinds of incorrectness that can be captured in OL, only one of which is expressible in IL (\Cref{sec:ndfalse}). 
\item We exemplify how OL can be instantiated to find memory errors (\Cref{sec:incorrectness}) and probabilistic bugs (\Cref{sec:prob}). We argue that the latter use case is not feasible in IL (\Cref{sec:prob-incorrectness}).
\end{itemize}
Finally, we conclude in \Cref{sec:related} and \Cref{sec:conclusion} by discussing related work and next steps.

\section{Overview: A Landscape of Triples}
\label{sec:overview}

The study of incorrectness has made apparent the need for new program logics that guarantee true positives and support under-approximate reasoning, since standard Hoare Logic---which does not enjoy those properties---is incapable of proving the presence of bugs. Concretely, in a valid Hoare Triple, denoted $\vDash\hoare PCQ$, running the program $C$ in any state satisfying the precondition $P$ will result in a state satisfying the postcondition $Q$ (the formal definition is given in Figure~\ref{fig:triples}). Suppose we wanted to use such a triple to prove that the program $x:=\textsf{malloc}()\fatsemi [x]\leftarrow1$ has a bug (\textsf{malloc} may nondeterministically return null, causing the program to crash with a segmentation fault when the subsequent command attempts to store the value 1 at the location pointed to by $x$). We might be tempted to specify the triple as follows:
\begin{equation}\label{eq:example}
\hoare{\tru}
{x:=\textsf{malloc}()\fatsemi [x]\leftarrow1}
{(\ok: x\mapsto 1) \lor (\er: x=\textsf{null})}
\end{equation}
Here, the assertion $(\ok:p)$ means that the program terminated successfully in a state satisfying $p$ and $(\er:q)$ means that it crashed in a state satisfying $q$. However, this is not quite right. According to the semantics of Hoare Logic, every possible end state must be covered by the postcondition, hence the need to use a disjunction to indicate that two outcomes are possible. But since we do not know that every state described by the postcondition is \emph{reachable}, it is possible that every program trace ends up satisfying the first disjunct ($\ok:x\mapsto 1$) and the error state is never reached. 

Incorrectness Logic offers a solution to this problem. In a valid Incorrectness Triple, $\vDash\inc PCQ$, every state satisfying $Q$ is reachable by running $C$ in some state satisfying $P$. So, simply switching the triple type in the above example \emph{does} give us a witness that the error is possible.
\begin{equation}\label{ex:il1}
\inc{\tru}
{x:=\textsf{malloc}()\fatsemi [x]\leftarrow1}
{(\ok: x\mapsto 1) \lor (\er: x=\textsf{null})}
\end{equation}
Though the conclusion remains a disjunction, the semantics of the incorrectness triple (\Cref{fig:triples}) ensures that \emph{every} state in the disjunction is reachable. Moreover, we can under-approximate by \emph{dropping disjuncts} from the postcondition and use the simpler specification: 
\begin{equation}\label{eq:inc-example}
\inc{\tru}
{x:=\textsf{malloc}()\fatsemi [x]\leftarrow1}
{\er: x=\textsf{null}}
\end{equation}
This more parsimonious specification still witnesses the error while also helping to ensure efficiency of large-scale automated analyses, which must keep descriptions at each program point small.

The duality between Hoare Logic and Incorrectness Logic appears sensible. Hoare Logic has \emph{no false negatives}---a program is only correct if we account for all the possible outcomes. Incorrectness Logic has \emph{no false positives}---an error is only worth reporting if it is truly reachable. However, we argue in this paper that incorrectness reasoning and Hoare Logic are \emph{not} in fact at odds: an approach to incorrectness that is more similar to Hoare Logic is not only possible but, in fact, advantageous for several reasons, including the ability to express when an error will be \emph{manifest} and the ability to reason about additional varieties of incorrectness. 

\subsection{Unifying Correctness and Incorrectness}
\label{sec:mhl}

\begin{figure}
{\small
\begin{tabular}{m{12em}|lllllcr}
Triple Name & Syntax && \multicolumn{3}{l}{Semantics} \\
\hline\hline
{\sf Hoare Logic} &
$\mktripple{\{}{\}}$
&iff&
$\forall \sigma\vDash P.$ & $\forall \tau.$ & $\tau\in\de C(\sigma)$ & $\Rightarrow$ & $\tau\vDash Q$
\\\medskip
{\sf Incorrectness Logic (IL) / \newline
Reverse Hoare Logic (RHL)} &
$\mktripple{[}{]}$
&iff&
$\forall \tau\vDash Q.$ & $\exists\sigma.$ & $\tau \in\de C(\sigma)$ & and & $\sigma\vDash P$
\\\medskip
{\sf Outcome Logic (OL) }&
$\mktripple{\langle}{\rangle}$
&iff&
$\forall m. $ & \multicolumn{4}{r}{$m\vDash P \qquad \implies \qquad \dem Cm\vDash Q$}
\end{tabular}}

\caption{Semantics of triples where $P$ and $Q$ are logical formulae, $C$ is a program,  $\Sigma$ is the set of all program states, $\sigma,\tau \in \Sigma$, and $\de{C}\colon\Sigma\to\bb{2}^\Sigma$ is the reachable states function. In the last line of the table, $M$ is a monad, $m\in M\Sigma$ and $\de{C}^\dagger\colon M\Sigma\to M\Sigma$ is the monadic lifting of $\de\cdot \colon \Sigma \to M\Sigma$.}
\label{fig:triples}
\end{figure}

Our first insight is that the inability to prove the existence of bugs is not inherent in the semantics of Hoare Logic. Rather, it is the result of an assertion logic that is not expressive enough to reason about reachability. Triple (\ref{eq:example}) shows how the usual logical disjunction is inadequate in reaching this goal.  To remedy this, we use a logic with extra algebraic structure on outcomes, reminiscent of the use of a resource logic in separation logic~\cite{localreasoning,sl}. In this case, resources are program outcomes rather than heap locations. Program outcomes do not necessarily need to be the usual traces in a (non-)deterministic execution model, but can also arise from programs with alternative execution models such as probabilistic computation. To model different types of computations in a uniform way, we use an execution model parametric on a monad. We call this new logic Outcome Logic (OL), with triples denoted by $\triple PCQ$ (defined formally in \Cref{fig:triples}). 
Let us schematically point out the generalizations in these new triples:
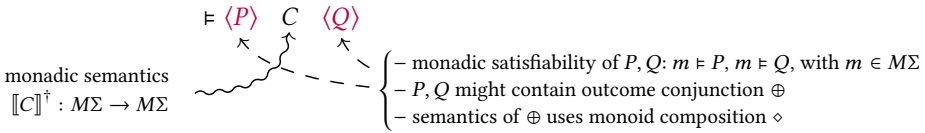
\[
\xymatrix@R=.05cm@C=0cm{
&\vDash \ob  P\  &C&\ \ob Q &\\
{\footnotesize
\begin{array}{c}
 \text{monadic semantics }\\
  \de{C}^\dagger   \colon M\Sigma\to M\Sigma
  \end{array}}
 \ar@{~>}@(r,d)[rru]&&&& 
 \ar@{-->}@(l,d)[lllu] \ar@{-->}@(l,d)[lu]
\left\{ {\footnotesize\begin{array}{@{-\ }l}
\text{monadic satisfiability of $P,Q$: $m\vDash P$, $m\vDash Q$, with $m\in M\Sigma$} \\  \text{$P,Q$ might contain outcome conjunction $\oplus$}\\
 \text{semantics of $\oplus$ uses monoid composition $\monoid$}   \end{array}}\right.\\
 }
\]
OL triples follow the spirit of Hoare Logic---first quantifying over elements satisfying the precondition and then stipulating that the result of running the program on such an element must satisfy the postcondition. The difference is that in OL triples, the pre- and postconditions are satisfied by a monoidal collection of outcomes $m\in M\Sigma$ rather than individual program states $\sigma\in\Sigma$.
This allows us to introduce a new connective in the logic---the \emph{outcome conjunction} $\oplus$---which models program outcomes as resources. Consider the postcondition in triple (\ref{ex:il1}) if we replace $\vee$ by $\oplus$:
\[
(\ok: x\mapsto 1) \vee (\er: x=\textsf{null}) \quad \text{vs.} \quad
(\ok: x\mapsto 1) \oplus (\er: x=\textsf{null})
\]
A program state satisfies the first formula just by satisfying one of the disjuncts, whereas the second one requires a collection of states that can be split to witness satisfaction of both. 
This ability to split outcomes emerges as a requirement that $M\Sigma$ is a (partial commutative) monoid. 
Given two outcomes $m_1,m_2 \in M\Sigma$, there is an operation $\monoid$ that enables us to combine them $m_1\monoid m_2 \in M\Sigma$. The satisfiability of $\oplus$ is then defined using $\monoid$ to split the monoidal state:
\[
m\vDash P\oplus Q
\qquad \text{ iff }\qquad
\exists{m_1,m_2 \in M\Sigma}.\quad m=m_1\monoid m_2 \quad\text{and}\quad m_1 \vDash P \quad\text{and}\quad m_2\vDash Q
\]
Consider instantiating the above to the powerset monad that associates a set $A$ with the set of its subsets $\bb{2}^A$. Given a semantic function $\de{C}\colon \Sigma \to \bb{2}^\Sigma$ that maps individual start states $\sigma$ to the set of final states reachable by executing $C$, we can give a monadic semantics $\dem{C}{S} = \bigcup_{\sigma \in S} \de{C}(\sigma)$ where $S$ is a \emph{set} of start states.\footnote{The $\de{-}^\dagger$ function is formally the monadic (or Kleisli) extension of $\de{-}$; we will define this formally in \Cref{sec:prelim}.} The monoid composition $\monoid$ on $\bb{2}^A$ is given by set union, which is used compositionally to define satisfiability of $\oplus$ as follows: $S\vDash P\oplus Q$ iff $S_1\vDash P$ and $S_2\vDash Q$ such that $S=S_1\cup S_2$. Given some satisfaction relation for individual program states ${\vDash_\Sigma}\subseteq \Sigma\times\prop$, we then define satisfaction of atomic assertions as follows:
\[
S\vDash P
\qquad\text{iff}\qquad
S\neq\emptyset \quad\text{and}\quad \forall \sigma\in S.~\sigma\vDash_\Sigma P 
\]
The extra restriction $S\neq\emptyset$ witnesses that $P$ is reachable (and not vacuously satisfied). Putting this all together, we instantiate the generic OL triples (\Cref{fig:triples}) to the powerset monad:
\[
\vDash\triple PCQ
\qquad\text{iff}\qquad
 \forall S\in\bb{2}^\Sigma. \quad
 S\vDash P 
 \quad\Rightarrow\quad
  \dem CS\vDash Q 
\]
Now, we can revisit the example in triple (\ref{ex:il1}) in OL using $\oplus$ instead of $\vee$:
\begin{equation}
\triple{\ok:\tru}
{x:=\textsf{malloc}()\fatsemi [x]\leftarrow1}
{(\ok: x\mapsto 1) \oplus (\er: x=\textsf{null})}
\end{equation}
This specification \emph{does} witness the bug---for any start state there is at least one end state that satisfies each of the outcomes. However, we are still recording extra, non-erroneous outcomes, which is problematic for a large scale analysis algorithm. Following the example in triple (\ref{eq:inc-example}),  we would like to specify the bug above in a way that mentions only the relevant outcome in the postcondition. We can achieve this by simply weakening the postcondition. According to the semantics above, the following implications hold:
\[
S \vDash P\oplus Q \quad \Rightarrow \quad S \vDash P\oplus\top \qquad \text{ and } \qquad S \vDash P\oplus Q \quad\Rightarrow \quad S \vDash \top\oplus Q 
\]
So in a sense, we can \emph{drop outcomes} by converting them to $\top$. For notational convenience, we define the following under-approximate triple:
\[
\vDash^\downarrow\triple PCQ
\qquad\text{iff}\qquad
\vDash\triple PC{Q\oplus\top}
\]
Using this shorthand, the following simpler specification is also valid:
\begin{equation}\label{eq:monadic-example}
\vDash^\downarrow\triple{\ok:\tru}
{x:=\textsf{malloc}()\fatsemi [x]\leftarrow1}
{\er: x=\textsf{null}}
\end{equation}
This example demonstrates that OL is suitable for reasoning about crash errors, just like IL. However our goal is not simply to cover the same use cases as IL, but rather to go further. Next, we will show in \Cref{sec:ofls} that there are bugs expressible in OL that cannot be expressed in IL. In \Cref{sec:omanifest} we will also explain why the semantics of OL are a better fit for characterizing an important class of bugs known as manifest errors.

\subsection{A Broader Characterization of Correctness and Incorrectness}\label{sec:ofls}

In the semantics of Incorrectness Logic, the notions of reachability and under-approximation are conflated: both are a consequence of the fact that IL quantifies over the states that satisfy the postcondition. However, reachability and under-approximation are separate concepts and OL allows us to reason about each independently. Reachability is expressed with the outcome conjunction $\oplus$ and under-approximation is achieved by dropping outcomes. Separating reachability and under-approximation is useful for \emph{both} correctness \emph{and} incorrectness reasoning.

To see this, we will first investigate \emph{correctness} properties that rely on reachability. Before the introduction of Incorrectness Logic by \citet{il}, \citet{reversehoare} devised a semantically equivalent logic, which they called Reverse Hoare Logic. The goal of this work was to prove \emph{correctness} specifications that involved multiple possible end states, all of which must be reachable. As we saw in Example~\ref{eq:example}, Hoare Logic cannot express such specifications. So, \citet{reversehoare} proposed the Reverse Hoare Triple, which---like Incorrectness Triples---guarantees that every state described by the postcondition is reachable.

The motivating example for Reverse Hoare Logic was a nondeterministic shuffle function. Consider the following specification, where $\Pi(a)$ is the set of permutations of $a$:
\[
\inc{\tru}{b := \textsf{shuffle}(a)}{b \in \Pi(a)}
\]
This specification states that \emph{every} permutation of the list is a possible output of \textsf{shuffle}; however, it is not a complete correctness specification. It does not rule out the possibility that the output is not a permutation of the input ($b\notin\Pi(a)$). The semantics of Reverse Hoare Logic is motivated by \emph{reachability}, but---like Incorrectness Logic---it achieves reachability in a manner that is inextricably linked to under-approximation, which is undesirable for correctness reasoning.

 \citet{reversehoare} note this, stating that a complete specification for $\textsf{shuffle}$ would require both Hoare Logic \emph{and} Reverse Hoare Logic, but also that it would be worthwhile to study logics that can ``express both the reachability of good states and the non-reachability of bad states''~\cite[\S8]{reversehoare}.
OL does just that---the full correctness of the shuffle program can be captured using a single OL triple that guarantees reachability \emph{without} under-approximating:
\begin{equation}\label{ex:shuffle}
\triple{\tru}{b:=\textsf{shuffle}(a)}{\smashoperator{\bigoplus_{\pi\in\Pi(a)}} (b=\pi)}
\end{equation}
The OL specification above states \emph{not only} that all the permutations are reachable, \emph{but also} that they are the only possible outcomes. So, OL allows us to express a correctness property in a single triple that otherwise would have required \emph{both} a Hoare Triple \emph{and} a Reverse Hoare Triple.

We now turn to consider incorrectness reasoning.
Given that the above OL triple is a complete correctness specification, we are interested to know what it would mean for $\mathsf{shuffle}$ to be incorrect. In other words, what would it take to \emph{disprove} the specification of $\mathsf{shuffle}$? There are two ways that the triple could be false: either one particular permutation $\pi\in\Pi(a)$ is not reachable or the output $b$ is (sometimes) not a permutation of $a$. Both bugs can be expressed as OL triples:
\[
\exists\pi\in\Pi(a).\ \triple{\tru}{b := \textsf{shuffle}(a)}{b\neq\pi}
\qquad\qquad
\triple{\tru}{b := \textsf{shuffle}(a)}{(b\notin\Pi(a))\oplus\top}
\]
These triples both denote \emph{true bugs} since the validity of either triple implies that specification (\ref{ex:shuffle}) is false. In fact, these are the \emph{only} ways that specification (\ref{ex:shuffle}) can be false. This follows from a more general result called Falsification, which we prove in \Cref{thm:ndfalse}:
\[
\not\vDash\triple PC{\bigoplus_{i=1}^n Q_i}
\qquad\text{iff}\qquad
\exists P'\Rightarrow P.\quad
\exists i. ~\vDash\triple {P'}C{\lnot Q_i}
\quad\text{or}\quad
\vDash\triple {P'}C{(\bigwedge_{i=1}^n\lnot Q_i)\oplus\top}
\]
Intuitively, a nondeterministic program is incorrect iff either one of the desired outcomes never occurs or some undesirable outcome sometimes occurs.\footnote{In general, there is also a third option: the program diverges (has no outcomes). See \Cref{thm:ndfalse}.}
Incorrectness Logic can only characterize the latter type of incorrectness, whereas OL accounts for both and is thus strictly more expressive in the nondeterministic setting. An analogous result holds for probabilistic programs (\Cref{sec:probfalse}), whereas IL is not suitable for reasoning about probabilistic incorrectness at all (\Cref{sec:prob-incorrectness}).

\subsection{Semantic Characterizations of Bugs}\label{sec:omanifest}

In addition to enabling us to witness a larger class of incorrectness than IL (unreachable states and probabilistic incorrectness), OL also provides a more \emph{intuitive} way to reason about the type of bugs that IL was designed for: reachability of unsafe states.

Recalling the crash error in \Cref{sec:mhl}, both IL triples and OL triples soundly characterize the bug, as they both witness a trace that reaches the crash. The Incorrectness Triple (\ref{eq:inc-example}) states that any failing execution where $x$ is null is reachable from some starting state. In other words, true is a \textit{necessary} condition to reach a segmentation fault. However, true is trivially a necessary condition, so this triple does not tell us much about what will trigger the bug in practice. By contrast, the OL triple (\ref{eq:monadic-example}) states that true is a \textit{sufficient} condition, which gives us more information---the bug can \emph{always} occur no matter what the starting state is.

The latter semantics has a close correspondence to a class of bugs, known as \emph{manifest errors}~\cite{realbugs}, which occur regardless of how the enclosing procedure is used and are of particular interest in automated bug-finding tools.
\citet[Def.\ 3.2]{realbugs} give a formal characterization of manifest errors, but it is not a natural fit for Incorrectness Logic: determining whether an IL triple is a manifest error requires an auxiliary algorithmic check.
Though \citet{realbugs} note that there are connections between manifest errors and under-approximate variants of Hoare Logic, we go further in proving that their original definition of a manifest error is semantically equivalent to an OL triple of the form $\vDash^\downarrow\triple{\ok:\tru}{C}{\er: q\sep\tru}$ (\Cref{lem:manifest}). Manifest errors are therefore trivial to characterize in OL by a simple syntactic inspection.
This suggests that OL is semantically closer to the way in which programmers naturally characterize bugs.

In addition to being an intuitive foundation for incorrectness, OL unifies program analysis across two dimensions. First, it unifies correctness and incorrectness reasoning within a single program logic, and second, it does so across execution models (\eg nondeterministic and probabilistic). In the remainder of the paper, we will formalize the ideas that have been exemplified thus far. We formalize the OL model in \Cref{sec:prelim} and \Cref{sec:logic}, prove the applicability of OL to nondeterministic and probabilistic correctness and incorrectness in \Cref{sec:falsification}, and show how OL can be used in nondeterministic and probabilistic domains in \Cref{sec:incorrectness} and \Cref{sec:prob},  respectively. Given these advantages, we argue that OL offers a promising alternative foundation for incorrectness reasoning.

\section{A Modular Programming Language}
\label{sec:prelim}

We start by defining a programming language, inspired by Dijkstra's guarded command language~\cite{gcl}, see Figure~\ref{fig:comlang}. The syntax includes  $\zero$, which represents divergence, $\bb{1}$, acting as skip, sequential composition $C_1\fatsemi C_2$, choice $C_1+C_2$, iteration $C^\star$, and parametrizable atomic commands $c$. At first sight this looks like a standard imperative language (with nondeterministic choice). However, we will interpret the syntax in a semantic model that is parametric on a monad and a partial commutative monoid. The former enables a generic semantics of sequential composition, whereas the latter provides a generic interpretation of choice.

Before we define the semantic model we need to recall the definition of a monad and partial commutative monoid. We assume familiarity with basic category theory (categories, functors, natural transformations), see \citet{pierce} for an introduction.
\begin{definition}[Monad]
A monad is a triple  {\normalfont$\langle M, \bind, \unit\rangle$} in which $M$ is a functor on a category $\mathcal C$, $\normalfont\unit \colon \textsf{Id}\Rightarrow M$ is a natural transformation, and $\normalfont\bind \colon  MA \times (A\to MB)\to MB$ satisfies:
\begin{enumerate}
\normalfont
\item $\bind(m, \unit) = m$
\item $\bind(\unit(x), f) = f(x)$
\item $\bind(\bind(m, f), g) = \bind(m, \lambda x. \bind(f(x),g))$
\end{enumerate} 
\end{definition}
\noindent Typical examples of monads include powerset, error, and distribution monads (defined in \Cref{sec:falsification} and \Cref{sec:incorrectness}). 
Given a function $f\colon A \to MB$, its monadic extension $f^\dagger \colon MA \to MB$ is defined as $f^\dagger(m) = \textsf{bind}(m, f)$. 

\begin{definition}[PCM]
A partial commutative monoid (PCM) is a triple  $\langle X, \monoid, \ident \rangle$ consisting of a set $X$ and a partial binary operation $\monoid\colon X\to X\rightharpoonup X$ that is associative, commutative, and has unit $\ident$. 
\end{definition}
A typical example of a PCM, used in probabilistic reasoning, is $\langle [0,1], +, 0\rangle$ ($+$ is partial, it is undefined when the addition is out-of-bounds).
We are now ready to define the execution model we need to provide semantics to our language. 

\begin{definition}[Execution Model]
An Execution Model is a structure $\normalfont\langle M, \textsf{bind}, \textsf{unit}, \monoid, \ident\rangle$ such that $\normalfont\langle M, \textsf{bind}, \textsf{unit}\rangle$ is a monad in the category of sets, and for any set $A$, $\langle MA, \monoid, \ident\rangle$ is a PCM that preserves the monad \textsf{bind}: $\normalfont\textsf{bind}(m_1\monoid m_2, k) = \textsf{bind}(m_1, k) \monoid \textsf{bind}(m_2,k)$  and $\normalfont\bind(\ident, k) = \ident$.
\end{definition}

\begin{figure}
{\small
\[\de{C}\colon\Sigma\rightharpoonup M\Sigma\]
\begin{align*}
C\ ::=&~ \zero & \de{\zero}(\sigma) &= \ident\\
\mid&~ \bb{1} & \de{\bb{1}}(\sigma) &= \textsf{unit}(\sigma)\\
\mid&~ C_1 \fatsemi C_2 & \de{C_1\fatsemi C_2}(\sigma) &= \textsf{bind}(\de{C_1}(\sigma), \de{C_2})\\
\mid&~ C_1 + C_2 & \de{C_1+ C_2}(\sigma) &= \de{C_1}(\sigma) \monoid \de{C_2}(\sigma)\\
\mid&~ C^\star & \llbracket{C^\star}\rrbracket(\sigma) &= \textsf{lfp}(\lambda f.\lambda \sigma. f^\dagger(\de{C}(\sigma))) \monoid \unit(\sigma))(\sigma) \\
\mid&~ c & \de{c}(\sigma) &= \de{c}_\textsf{atom}(\sigma)
\end{align*}}
\caption{Syntax and Semantics of the Command Language parameterized by an execution model $\langle M, \bind, \unit, \monoid, \ident\rangle$ and a language of atomic commands with semantics $\de{c}_\textsf{atom} \colon \Sigma\to M\Sigma$}
\label{fig:comlang}
\end{figure}

In \Cref{fig:comlang} we present the semantics of the language. The monad operations are used to provide semantics to $\bb{1}$ and sequential composition $\fatsemi$ whereas the monoid operation is used in the semantics of choice and iteration. Note that in general the semantics of the language is partial since $\monoid$ is partial, which is necessary in order to express a probabilistic semantics, since two probability distributions can only be combined if their cumulative probability mass is at most 1. For the languages we will work with in this paper, there are simple syntactic checks to ensure totality of the semantics. In the probabilistic case, this involves ensuring that all uses of $+$ and $\star$ are guarded. We show that the semantics is total for the execution models of interest in \Aref{app:totality}.

\begin{example}[State and Guarded Commands]
\label{ex:gcl}
The base language introduced in the previous section is parametric over a set of program states $\Sigma$. In this example, we describe a specific type of program state, the semantics of commands over those states, and a mechanism to define the typical control flow operators (\textsf{if} and \textsf{while}). First, we assume some syntax of program expressions $e\in\textsf{Exp}$ which includes variables $x\in\textsf{Var}$ as well as the typical Boolean and arithmetic operators. Atomic commands come from the following syntax.
\[c ::= \textsf{assume }e \mid x:=e \qquad (x\in\textsf{Var}, e\in \textsf{Exp}) \]
The command $\textsf{assume}~e$ does nothing if $e$ is true and eliminates the current outcome if not; $x:=e$ is variable assignment. 
A program stack is a mapping from variables to values $\mathcal{S} = \{ s:\textsf{Var}\to\textsf{Val}\}$ where program values $\textsf{Val} = \mathbb{Z} \cup \mathbb{B}$ are integers ($\mathbb Z$) or Booleans ($\mathbb{B} = \{\textsf{true}, \textsf{false}\}$).
Expressions are evaluated to values given a stack using $\de{e}_\textsf{Exp}\colon\mathcal{S}\to\textsf{Val}$. 
The semantics of atomic commands $\de{c}\colon \mathcal{S}\to M\mathcal{S}$, parametric on an execution model, is defined below.
\[
{\de{\textsf{assume}~e}}(s) = \left\{\begin{array}{ll}
\textsf{unit}(s) & \text{if}~\de{e}_\textsf{Exp}\!(s) = \textsf{true} \\
\ident & \text{if}~\de{e}_\textsf{Exp}\!(s) = \textsf{false}
\end{array}\right.
\hspace{2em}
{\de{x:=e}}(s) = \textsf{unit}(s[x\mapsto\de{e}_\textsf{Exp}\!(s)])
\]
While a language instantiated with the atomic commands described above is still nondeterministic, we can use \textsf{assume} to define the usual (deterministic) control flow operators as syntactic sugar.
\begin{align*}
\iftf{e}{C_1}{C_2} ~&= (\textsf{assume}~e\fatsemi C_1)+(\textsf{assume}~\lnot e\fatsemi C_2) & \skp ~&= \bb{1} & C^0 ~&= \bb{1}\\
\whl{e}{C} ~&= (\textsf{assume}~e\fatsemi C)^\star\fatsemi\textsf{assume}~\lnot e &
\textsf{for}~N~\textsf{do}~C ~&= C^N
& C^{k+1} ~&= C\fatsemi C^k
\end{align*}
In fact, when paired with a nondeterministic evaluation model, this language is equivalent to \cites{gcl} Guarded Command Language (\textsf{GCL}) by a straightforward syntactic translation.
\end{example}

\section{Outcome Logic}
\label{sec:logic}

In this section, we formally define Outcome Logic (OL). We first define the logic of outcome assertions which will act as the basis for writing pre- and postconditions in OL. Next, we give the semantics of OL triples, which is parametric on an execution model, atomic command semantics, and an assertion logic. Finally, we give proof rules that are sound for all OL instances.

\subsection{A Logic for Monoidal Assertions: Modeling the Outcome Conjunction}
\label{sec:outcomes}

We now give a formal account of the outcome assertion logic that was briefly described in \Cref{sec:mhl}. The outcome assertion logic is an instance of the Logic of Bunched Implications (BI)~\cite{bi}, a substructural logic that is used to reason about resource usage. Separation logic~\cite{sl} and its extensions~\cite{csl} are the most well-known applications of BI. In our case, the relevant resources are \emph{program outcomes} rather than heap locations. 

We use the formulation of BI due to~\citet{simon}. While \citet{simon} gives a thorough account of the BI proof theory, we are mainly interested in the semantics for the purposes of this paper. The syntax and semantics are given in Figure~\ref{fig:outcomes} with logical negation $\lnot\varphi$ being defined as $\varphi\Rightarrow\bot$. The semantics is parametric on a BI frame $\langle X, \monoid, \preccurlyeq, \varnothing\rangle$---where $\langle X, \monoid, \varnothing\rangle$ is a PCM and $\mathord\preccurlyeq\subseteq X\times X$ is a preorder---and
a satisfaction relation for basic assertions  ${\vDash_\mathsf{atom}}\subseteq X\times\prop$.

The two non-standard additions are the \emph{outcome conjunction} $\oplus$, a connective to join outcomes, and $\top^\oplus$, an assertion to specify that there are no outcomes. These intended meanings are reflected in the semantics: $\top^\oplus$ is only satisfied by the monoid unit $\ident$, whereas $\varphi\oplus\psi$ is satisfied by $m$ iff $m$ can be partitioned into $m_1 \monoid m_2$ to satisfy each outcome formula separately.
We will focus on \emph{classical} interpretations of BI where the preorder $\preccurlyeq$ is equality.\footnote{Intuitionistic interpretations of BI with non-trivial preorders can be used as an alternative way to encode under-approximate program logics. This idea is explored in \Aref{app:intuitionistic}.}


\begin{definition}[Outcome Assertion Logic]\label{def:outcome}
Given an execution model $\langle M, \bind, \unit, \monoid, \ident\rangle$ and a satisfaction relation for atomic assertions ${\vDash_\textsf{atom}}\subseteq M\Sigma\times\prop$, an Outcome Assertion Logic is an instance of BI based on the BI frame $\langle M\Sigma, \monoid, =, \ident\rangle$. Informally, we refer to BI assertions $\varphi,\psi$ as outcome assertions and the atomic assertions $P,Q\in\prop$ as individual outcomes.
\end{definition}

\begin{remark}[Notation for Assertions]
For the remainder of the paper, lowercase Greek metavariables $\varphi,\psi$ refer to (syntactic) outcome assertions (\Cref{def:outcome}), uppercase Latin metavariables $P$, $Q$ refer to atomic assertions (individual outcomes), and lowercase Latin metavariables $p$, $q$ refer to assertions on individual program states.
\end{remark}

\begin{example}[Outcomes]
We mentioned one example of a PCM in \Cref{sec:overview}: $X$ can be sets of program states and the monoid operation $\monoid$ is set union. Another example is probability (sub)distributions over a set and $\monoid$ is $+$. This monoid operation is partial; adding two subdistributions is only possible if the mass associated with a point (and the entire distribution) remains in $[0,1]$.
\end{example}

\begin{figure}
\[
\small
\begin{array}{rl@{\qquad\qquad}lll}
\varphi ::=& \top \hspace{6em}\;& m\vDash\top &\multicolumn{2}{l}{\text{always}}\\
\mid&\bot & m\vDash\bot&\multicolumn{2}{l}{\text{never}}\\
\mid&\top^\oplus & m\vDash\top^\oplus & \text{iff}& m=\varnothing \\
\mid& \varphi\land\psi & m\vDash\varphi\land \psi &\text{iff} & m\vDash\varphi ~\text{and}~m\vDash\psi\\
\mid& \varphi\oplus\psi & m\vDash\varphi\oplus\psi & \text{iff} & \exists m_1,m_2. ~ m_1\monoid m_2 \preccurlyeq m ~\text{and}~ m_1\vDash\varphi ~\text{and}~ m_2\vDash\psi\\
\mid& \varphi\Rightarrow\psi & m\vDash \varphi\Rightarrow\psi & \text{iff} & \forall m'. ~\text{if}~m \preccurlyeq m' ~\text{and}~m'\vDash\varphi ~\text{then}~ m'\vDash\psi \\
\mid& P & m\vDash P & \text{iff} & P\in\prop ~\text{and}~ m \vDash_\textsf{atom} P
\end{array}
\]
\caption{Syntax and semantics of BI given a BI frame $\langle X, \monoid,\preccurlyeq,\varnothing\rangle$ and satisfaction relation ${\vDash_\textsf{atom}}\subseteq X\times\prop$}
\label{fig:outcomes}
\end{figure}

As discussed in \Cref{sec:overview}, under-approximation and the ability to drop outcomes is an important part of incorrectness reasoning as it allows large scale analyses to only track pertinent information. We use the following shorthand to express under-approximate outcome assertions.

\begin{definition}[Under-Approximate Outcome Assertions]\label{def:uxoutcome}
Given an outcome assertion logic with satisfaction relation ${\vDash}\subseteq M\Sigma\times\prop$, we define an under-approximate version ${\vDash^\downarrow}\subseteq M\Sigma\times\prop$ as $m \vDashD \varphi$ iff $~m \vDash \varphi\oplus\top$.
\end{definition}

Intuitively, $\varphi\oplus\top$ corresponds to under-approximation since it states that $\varphi$ only covers a subset of the outcomes (with the rest being unconstrained, since they are covered by $\top$). Defining under-approximation in this way allows us to reason about correctness and incorrectness within a single program logic. It also enables us to drop outcomes simply by weakening; it is always possible to weaken an outcome to $\top$, so $m\vDash P\oplus Q$ implies that $m\vDash P\oplus\top$. Equivalently, $m\vDash^\downarrow P\oplus Q$ implies that $m\vDash^\downarrow P$.
These facts are proven in \Aref{app:ux}. A similar formulation would be possible using an intuitionistic interpretation of BI (where, roughly speaking, we take the preorder to be $m_1\preccurlyeq m_2$ iff $\exists m. ~m_1\monoid m = m_2$). We prove this correspondence in \Aref{app:intuitionistic}.

\subsection{Outcome Triples}
\label{sec:mhldef}

We now have all the ingredients needed to define the validity of the program logic.

\begin{definition}[Outcome Triples]\label{def:mht} The parameters needed to instantiate OL are:
\begin{enumerate}
\item An execution model: $\langle M, \bind, \unit, \monoid, \ident\rangle$
\item A set of program states $\Sigma$ and semantics of atomic commands: $\de{c}_\mathsf{atom}\colon\Sigma\to M\Sigma$
\item A syntax of atomic assertions $\prop$ and satisfaction relation: ${\vDash_\mathsf{atom}}\subseteq M\Sigma\times\prop$
\end{enumerate}
Now, let $\de{-} \colon \Sigma\to M\Sigma$ be the semantics of the language in \Cref{fig:comlang} with parameters (1) and (2) and $\vDash$ be the outcome assertion satisfaction relation (\Cref{def:outcome}) with parameters (1) and (3). For any program $C$ (\Cref{fig:comlang}), and outcome assertions $\varphi$ and $\psi$:
\[
\vDash\triple{\varphi}C\psi
\qquad\text{iff}\qquad
 \forall m\in M\Sigma. \quad m\vDash \varphi \quad\implies\quad \dem{C}{m}\vDash \psi
 \]
\end{definition}
OL is a generalization of Hoare Logic---the triples first quantify over elements satisfying the precondition and then stipulate that the result of running the program on those elements satisfies the postcondition. The difference is that now the pre- and postconditions are expressed as outcome assertions and thus satisfied by a monoidal collection $m\in M\Sigma$, which can account for execution models based on nondeterminism and probability distributions.

Using outcome assertions for pre- and postconditions adds significant expressive power. We already saw in \Cref{sec:overview} how Outcome Logic allows us to reason about reachability and under-approximation. We can also encode other useful concepts such as partial correctness---the postcondition holds \emph{if} the program terminates---by taking a disjunction with $\top^\oplus$ to express that the program may diverge\footnote{
Disjunctions are defined $\varphi\vee\psi$ iff $\lnot(\lnot\varphi\land\lnot\psi)$, a standard encoding in classical logic.
}.
For convenience, we define the following notation where the left triple encodes under-approximation and the right triple encodes partial correctness.\
 \[
 \vDashD\triple\varphi{C}\psi
 \quad\text{iff}\quad
 \vDash\triple{\varphi}C{\psi\oplus\top}
 \qquad\qquad
 \vDash_\mathsf{pc}\!\!\triple\varphi{C}\psi
 \quad\text{iff}\quad
 \vDash\triple\varphi{C}{\psi\lor\top^\oplus}
 \]
 In fact, the right triple corresponds exactly to standard Hoare Logic (\Cref{fig:triples}) if we instantiate OL using the powerset semantics (\Cref{def:ndeval}) and limit the pre- and postconditions to be atomic assertions. This result is stated below and proven in \Aref{app:triple-proofs}.
 \begin{restatable}[Subsumption of Hoare Triples]{theorem}{subhoare}
\label{thm:hl}
$\vDash\hoare PCQ$ \quad iff \quad
$\vDash_\mathsf{pc}\!\triple PC{Q}$
\end{restatable}
While capturing many logics in one framework is interesting and demonstrates the versatility of Outcome Triples, our primary goal is to investigate the roles that these program logics can play for expressing correctness and incorrectness properties. We justify OL as a theoretical basis for correctness and incorrectness reasoning in \Cref{sec:falsification} and give examples for how OL can be applied to nondeterministic and probabilistic programs in \Cref{sec:incorrectness} and \Cref{sec:prob}.

\subsection{Proof Systems}
\label{sec:rules}

\begin{figure}
\footnotesize
\begin{flushleft}\fbox{Generic Rules}\end{flushleft}
\[
\inferrule{\;}{\triple{\varphi}{\zero}{\top^\oplus}}{\textsc{Zero}}
\hspace{1.5em}
\inferrule{\;}{\triple{\varphi}{\bb{1}}{\varphi}}{\textsc{One}}
\hspace{1.5em}
\inferrule{\triple \varphi{C_1}\psi \\ \triple \psi{C_2}\vartheta}
{\triple \varphi{C_1\fatsemi C_2}\vartheta}
{\textsc{Seq}}
\hspace{1.5em}
\inferrule{\forall i\in\mathbb{N}.~\triple{\varphi_i}C{\varphi_{i+1}}}
{\triple{\varphi_0}{\textsf{for}~N~\textsf{do}~C}{\varphi_N}}
{\textsc{For}}
\]
\vspace{1em}
\[
\inferrule{
\triple{\varphi_1}C{\psi_1}
\hspace{1em}
\triple{\varphi_2}C{\psi_2}
}
{\triple{\varphi_1\oplus\varphi_2}C{\psi_1\oplus\psi_2}}
{\textsc{Split}}
\qquad\qquad
\inferrule{
\varphi'\Rightarrow\varphi \hspace{1em}
\triple\varphi{C}\psi \hspace{1em}
\psi\Rightarrow\psi'}
{\triple{\varphi'}C{\psi'}}
{\textsc{Consequence}}
\]
\[
\inferrule{\;}{\triple{\top^\oplus}C{\top^\oplus}}{\textsc{Empty}}
\qquad\qquad
\inferrule{\;}{\triple{\varphi}C\top}{\textsc{True}}
\qquad\qquad
\inferrule{\;}{\triple{\bot}C\varphi}{\textsc{False}}
\]
\vspace{1em}
\begin{flushleft}
\fbox{Nondeterministic Rules}
\end{flushleft}
\[
\inferrule{\triple \varphi{C_1}{\psi_1} \\ \triple \varphi{C_2}{\psi_2}}
{\triple \varphi{C_1+C_2}{\psi_1\oplus \psi_2}}
{\textsc{Plus}}
\hspace{4em}
\inferrule
{\triple{\varphi}{\bb{1} + C\fatsemi C^\star}{\psi}}
{\triple{\varphi}{C^\star}\psi}
{\textsc{Induction}}
\]
\vspace{1em}
\begin{flushleft}\fbox{Expression-Based Rules}\end{flushleft}
\[
\inferrule{\;}{\triple{P[e/x]}{x:=e}{P}}{\textsc{Assign}}
\hspace{4em}
\inferrule{
P_1\vDash e
\\
P_2\vDash \lnot e
}
{\triple{P_1\oplus P_2}{\assume e}{P_1}}
{\textsc{Assume}}
\]
\vspace{1em}
\[
\inferrule{
P_1\vDash e \\ \triple{P_1}{C_1}{Q_1}
\\
P_2\vDash \lnot e \\ \triple{P_2}{C_2}{Q_2}
}
{\triple{P_1\oplus P_2}{\iftf{e}{C_1}{C_2}}{Q_1\oplus Q_2}}
{\textsc{If (Multi-Outcome)}}
\]

\caption{Inference rules that are valid for a variety of OL instantiations. The metavariables $\varphi$, $\psi$ refer to arbitrary outcome assertions and $P$, $Q$ refer to atomic (single-outcome) assertions.
}
\label{fig:baserules}
\end{figure}

Now that we have formalized the \emph{validity} of Outcome triples (denoted $\vDash\triple \varphi C\psi$), we can construct proof systems for this family of logics.  We write $\vdash\triple \varphi C\psi$ to mean that the triple $\triple \varphi C\psi$ is \emph{derivable} from a set of inference rules. Each set of inference rules that we define throughout the paper will be \emph{sound} with respect to a certain OL instance.

\medskip
\noindent\textbf{\em Global rules.} Some generic rules that are valid for any OL instance are shown at the top of \Cref{fig:baserules}. Most of the rules including \textsc{Zero}, \textsc{One}, and \textsc{Seq} are standard. The Rule of \textsc{Consequence} allows the strengthening and weakening of pre- and postconditions respectively using any semantically valid BI implication. The \textsc{Split} rule allows us to analyze the program $C$ with two different pre/postcondition pairs and join the results using an outcome conjunction.
\medskip

\noindent\textbf{\em Rules for nondeterministic programs.} In the middle of \Cref{fig:baserules} we see two rules that are only valid in nondeterministic languages where the semantics is based on the powerset monad. The \textsc{Plus} rule characterizes nondeterministic choice by joining the outcomes from analyzing each branch using an outcome conjunction. Repeated uses of the \textsc{Induction} rule allow us to unroll an iterated command for a finite number of iterations.
\medskip

\noindent\textbf{\em Rules for guarded programs.}  Finally, at the bottom of \Cref{fig:baserules} is a collection of rules for expression-based languages that have the syntax introduced in \Cref{ex:gcl}. We write $P\vDash e$ to mean that $P$ \emph{entails} $e$. Formally, if $P\vDash e$ and $Q\vDash\lnot e$ and $m\vDash P\oplus Q$, then $\dem{\assume e}{m}\vDash P$. Substitutions $P[e/x]$ must be defined for basic assertions and satisfy $m\vDash P[e/x]$ implies $\dem{x:=e}{m}\vDash P$.

The \textsc{Assign} rule uses \emph{weakest-precondition} style backwards substitution. \textsc{Assume} uses expression entailment to annihilate the outcome where the guard is false.
Similarly, \textsc{If (Multi-Outcome)} uses entailment to map entire outcomes to the true or false branches of an if statement, respectively.
\medskip

\noindent All the rules in \Cref{fig:baserules} are sound (see \Aref{app:soundness} for details of the proof). 
\begin{theorem}[Soundness of Proof System]
If \;$\vdash\triple PCQ$\; then\; $\vDash\triple PCQ$
\end{theorem}

\noindent Note that it is not possible to have generic loop-invariant based iteration rules that are valid for \emph{all} instances of Outcome Logic. This is because loop invariants assume a \emph{partial} correctness specification; they do not guarantee termination. Outcome Logic---in some instantiations---guarantees reachability of end states and therefore must witness a \emph{terminating} program execution. 
This is in line with the Backwards Variant rule from Incorrectness Logic~\cite[Fig.2]{il}, the While rule from Reverse Hoare Logic~\cite[Fig.2]{reversehoare}, and Loop Variants from Total Hoare Logic~\cite{variant}. Such a rule for \textsf{GCL} is available in \Aref{sec:extra-rules}.

\section{Modeling Correctness and Incorrectness via Outcomes}
\label{sec:falsification}

Incorrectness Logic was motivated in large part by its ability to {\em disprove} correctness specifications (\ie Hoare Triples) \cite[Thm 4.1]{ilalgebra}. 
In this section, we prove that OL can not only disprove Hoare Triples (\Cref{thm:mhlil}), but it can also express strictly \emph{more} types of incorrectness than IL can. \Cref{thm:ndfalse} shows three classes of bugs in nondeterministic programs that can be characterized in OL, only one of which is expressible in IL. \Cref{sec:probfalse} shows that OL can express probabilistic incorrectness too, whereas IL cannot.

Our first result is stated in terms of \emph{semantic} triples in which the pre- and postconditions are \emph{semantic} assertions (which we denote with uppercase Greek metavariables $\Phi,\Psi\in\bb{2}^{M\Sigma}$) rather than the \emph{syntactic} assertions $\varphi,\psi\in\prop$ we have seen thus far. The advantage of this approach is that we can show the power of the OL model without worrying about the expressiveness of the syntactic assertion language (as a point of reference, the formal development of Incorrectness Logic is purely semantic~\cite{il, ilalgebra, realbugs}, as was the metatheory for separation logic~\cite{semanticsep, yang}).

The following Falsification theorem states that any false OL triple can be disproven within OL. Since we already know that OL subsumes Hoare Logic (\Cref{thm:hl}), it follows that any correctness property that is expressible in Hoare Logic can be disproven using OL. We use $\vDash_S\triple\Phi C\Psi$ to denote a valid semantic OL triple, that is: if $m\in\Phi$, then $\dem Cm\in\Psi$.
The assertion $\mathsf{sat}(\Phi)$ means that $\Phi$ is satisfiable, in other words $\Phi\neq\emptyset$.

\begin{restatable}[Semantic Falsification]{theorem}{falsification}\label{thm:falsification}
For any OL instance and any program $C$ and semantic assertions $\Phi$, $\Psi$:
\[
\not\vDash_S\triple\Phi C\Psi
\qquad\text{iff}\qquad
\exists \Phi'.\quad\text{such that}\quad\Phi'\Rightarrow\Phi,  \quad\mathsf{sat}(\Phi'), \quad\text{and}\quad \vDash_S\triple{\Phi'}C{\lnot\Psi}
\]
\end{restatable}
\begin{proof}
We provide a proof sketch here. If $\not\vDash_S\triple\Phi C\Psi$, then there must be an $m\in\Phi$ such that $\dem Cm\notin\Psi$. Choosing $\Phi' = \{m\}$ gives us $\vDash_S\triple{\Phi'}C{\lnot \Psi}$. For the reverse direction, we know from $\mathsf{sat}(\Phi')$ that there is an $m\in\Phi'$ and from $\Phi'\Rightarrow\Phi$, we know that $m\in\Phi$ and from $\vDash_S\triple{\Phi}C{\lnot\Psi}$, we know that $\dem Cm\notin\Psi$, so $\not\vDash_S\triple\Phi C\Psi$.
\end{proof}

The full proof of this theorem and formulation of semantic triples are given in \Aref{sec:semtriples}. While this result shows the power of the OL \emph{model}, we also seek to answer whether the outcome assertion syntax given in \Cref{def:outcome} can express the pre- and postconditions needed to disprove other triples. We answer this question in the affirmative, although the forward direction of the result has to be proven separately for nondeterministic and probabilistic models. While the semantic proof above applies to \emph{any} OL instance, the syntactic versions that we present in \Cref{sec:ndfalse} and \Cref{sec:probfalse} rely on additional properties of the specific OL instance. Despite the added complexity, we deem this worthwhile since syntactic descriptions give us a characterizations that can be used in the design of automated bug-finding tools.
 
The reverse direction of \Cref{thm:falsification} corresponds to \cites{il} Principle of Denial, though the original Principle of Denial used two triple types (IL and Hoare) and now we only need to use one (OL).  We can prove a syntactic version of The Principle of Denial for OL, which can be thought of as a generalization of the \emph{true positives} property, since it tells us when an OL triple (denoting a bug) disproves another OL triple (denoting correctness).

\begin{restatable}[Principle of Denial]{theorem}{truepos}\label{thm:denial}
For any OL instance and any program $C$ and syntactic assertions $\varphi$, $\varphi'$, and $\psi$:
\[
\text{If}\quad \varphi'\Rightarrow\varphi,  \quad\mathsf{sat}(\varphi'), \quad\text{and}\quad \vDash\triple{\varphi'}C{\lnot\psi}
\quad\text{then}\quad
\not\vDash\triple\varphi C\psi
\]
\end{restatable}
This theorem is a consequence of \Cref{thm:falsification}, together with a result stating how to translate syntactic triples to equivalent semantic ones \Acref{lem:strip}. 

Proving a syntactic version of the forward direction of \Cref{thm:falsification} is more complicated---it requires us to witness the existence of a syntactic assertion corresponding to $\Phi'$. The way in which this assertion is constructed depends on several properties of the OL instance. One additional requirement is that the program $C$ must terminate after finitely many steps, otherwise the precondition may not be finitely expressible. This is a common issue when generating preconditions and as a result many developments choose to work with semantic assertions rather than syntactic ones~\cite{kaminski}. The IL falsification results are also only given semantically~\cite{il,ilalgebra}, which avoids infinitary assertions in loop cases.

In the following sections, we will investigate falsification in both nondeterministic and probabilistic OL instances. In doing so, we will provide more specific falsification theorems which both deal with syntactic assertions and more precisely characterize the ways in which particular programs can be incorrect.
While we have just seen that we can obtain a falsification witness for correctness specifications $\triple\varphi C\psi$ by negating the postcondition, proving a triple with postcondition $\lnot\psi$ may not be convenient. For example, if $\psi$ is a sequence of outcomes $Q_1\oplus\cdots\oplus Q_n$, then it is not immediately clear what $\lnot\psi$ expresses. We therefore provide more intuitive assertions for canonical types of incorrectness encountered in programs.

\subsection{Falsification in Nondeterministic Programs}
\label{sec:ndfalse}

In this section, we explore falsification for nondeterministic programs. The first step is to formally define a nondeterministic instance of OL by defining an evaluation model and BI frame.

\begin{definition}[Nondeterministic Evaluation Model]\label{def:ndeval}
A nondeterministic evaluation model based on program states $\sigma\in\Sigma$ is $\langle \bb{2}^\Sigma, \bind, \unit, \cup, \emptyset\rangle$ where $\langle \bb{2}^{(-)}, \bind, \unit\rangle$ is the powerset monad:
\[
\bind(S, k) \triangleq \bigcup_{x\in S} k(x)
\qquad\qquad
\unit(x) \triangleq \{x\}
\]
\end{definition}

\begin{definition}[Nondeterministic Outcome Assertions]\label{def:ndoc}
Given some satisfaction relation on program states ${\vDash_\Sigma}\subseteq \Sigma\times\prop$, we create an instance of the outcome assertion logic (\Cref{def:outcome}) with the BI frame $\langle \bb{2}^\Sigma, \cup, =, \emptyset\rangle$ such that atomic assertions come from $\prop$ and are satisfied as follows:
\[
S\vDash P
\qquad\text{iff}\qquad
S\neq\emptyset \quad\text{and}\quad \forall\sigma\in S.~\sigma\vDash_\Sigma P
\]
\end{definition}

We impose one additional requirement, that the atomic assertions $Q\in\prop$ can be logically negated\footnote{Crucially, $\overline{Q}$ is not the same as $\lnot Q$ (where $\lnot$ is from BI) since $S\vDash \lnot Q$ iff $S=\emptyset$ or $\exists \sigma\in S. \sigma\not\vDash_\Sigma Q$ whereas $S\vDash\overline{Q}$ iff $S\neq\emptyset$ and $\forall \sigma\in S.~\sigma\not\vDash_\Sigma Q$.}, which we will denote $\overline{Q}$. Now, we return to the question of how to falsify a sequence of nondeterministic outcomes $Q_1\oplus \cdots\oplus Q_n$. \Cref{lem:ndfalse} shows that there are exactly three ways that this assertion can be false. 

\begin{restatable}[Falsifying Assertions]{lemma}{ndfalse}
\label{lem:ndfalse}
For any $S \in \bb{2}^\Sigma$ and atomic assertions $Q_1, \ldots, Q_n$,
\[
S\not\vDash Q_1\oplus\cdots\oplus Q_n
\qquad\text{iff}\qquad
\exists i.~S\vDash\overline{Q}_i
\quad\text{or}\quad
S\vDash (\overline{Q}_1\land\cdots\land \overline{Q}_n)\oplus\top
\quad\text{or}\quad
S\vDash\top^\oplus
\]
\end{restatable}
If we take $Q_1\oplus\cdots\oplus Q_n$ to represent a desirable set of program outcomes, then \Cref{lem:ndfalse} tells us that said program can be wrong in exactly three ways. Either there is some desirable outcome ($Q_i$) that the program never reaches, there is some undesirable outcome ($\overline{Q}_1\land\cdots\land\overline{Q}_n$) that the program sometimes reaches, or there is an input that causes it to diverge ($\top^\oplus$). 
 Now, following from this result, we can state what it means to falsify a nondeterministic specification:

\begin{restatable}[Nondeterministic Falsification]{theorem}{ndfalsethm}
\label{thm:ndfalse}
For any OL instance based on the nondeterministic evaluation model (\Cref{def:ndeval}) and outcome assertions (\Cref{def:ndoc}), $\not\vDash\triple{\varphi}C{\bigoplus_{i=1}^n Q_i}$ iff:
\[
\exists \varphi'\Rightarrow\varphi.\; \mathsf{sat}(\varphi') \;\text{ and }\;
\exists i.~\vDash\triple{\varphi'}C{\overline{Q}_i}
\;\text{or}\;
\vDashD\triple{\varphi'}C{\bigwedge_{i=1}^n \overline{Q}_i}
\;\text{or}\;
\vDash\triple{\varphi'}C{\top^\oplus}
\]
\end{restatable}

The type of bugs expressible in Incorrectness Logic are a special case of \Cref{thm:ndfalse}. Since IL is under-approximate, it can only express the second kind of bug (reachability of a bad outcome), not the first (non-reachability of a good outcome), or last (divergence). IL was motivated by its ability to disprove Hoare Triples---since Hoare Triples are a special case of OL (\Cref{thm:hl}), \Cref{thm:ndfalse} suggests that OL can disprove Hoare Triples as well. We make this correspondence explicit in the following Corollary where, compared to \Cref{thm:ndfalse}, the first two cases collapse since there is only a single outcome and the divergence case no longer represents a bug since the Hoare Triple is a partial correctness specification. 


\begin{restatable}[Hoare Logic Falsification]{corollary}{mhlil}\label{thm:mhlil}
\[
\not\vDash\hoare PCQ
\qquad\text{iff}\qquad
\exists\varphi\Rightarrow P.\quad
\mathsf{sat}(\varphi)
\quad\text{and}\quad
\vDashD\triple\varphi{C}{\overline Q}
\]
\end{restatable}

So, although we do not show that OL \emph{semantically} subsumes Incorrectness Logic, it does have the ability to express the same bugs as IL.
OL can also disprove more complex correctness properties, such as that of the $\mathsf{shuffle}$ function that we saw in \Cref{sec:ofls}. As we will now see, another OL instance is capable of disproving probabilistic properties too.



\subsection{Falsification in Probabilistic Programs}
\label{sec:probfalse}

Before we can define falsification in a probabilistic setting, we must establish some preliminary definitions.
Probabilistic programs use an execution model based on probability (sub)distributions. A (sub)distribution $\mu\in\mathcal{D}X$ over a set $X$ is a function mapping elements $x\in X$ to probabilities in $[0, 1]\subset\mathbb{R}$. The support of a distribution is the set of elements having nonzero probability $\textsf{supp}(\mu) = \{ x \mid \mu(x) > 0\}$ and the mass of a distribution is $|\mu| = \sum_{x\in\textsf{supp}(\mu)}\mu(x)$. A valid distribution must have mass at most 1. The empty distribution $\varnothing$ maps everything to probability 0 and distributions can be summed pointwise $\mu_1 + \mu_2 = \lambda x. \mu_1(x) + \mu_2(x)$ if $|\mu_1| + |\mu_2| \le 1$. For any countable set $X$, $\langle\mathcal DX, +, \varnothing\rangle$ is a PCM.
In addition, distributions can be weighted by scalars $p\cdot \mu = \lambda x. p\cdot\mu(x)$ if $p\cdot |\mu| \le 1$ (this is always defined if $p\le 1)$.
The Dirac distribution $\delta_x$ assigns probability 1 to $x$ and 0 to everything else. We complete the definition of a probabilistic execution model:
\begin{definition}[Probabilistic Evaluation Model]
\label{def:probeval}
A probabilistic evaluation model based on program states $\Sigma$ is defined as $\langle \mathcal{D}\Sigma, \bind, \unit, +, \varnothing\rangle$ where $\langle \mathcal{D}, \bind, \unit\rangle$ is the  \citet{giry} monad:
\[
\bind(\mu, k) = \smashoperator{\sum_{x\in\mathsf{supp}(\mu)}} \mu(x)\cdot k(x)
\qquad\qquad
\unit(x) = \delta_x
\]
\end{definition}

We can make our imperative language probabilistic by adding a command $x\samp\eta$ for sampling from finitely supported probability distributions $\eta\in\mathcal{D}\mathsf{Val}$ over program values. This command is intended to be added to an existing language such as \textsf{GCL} (\Cref{ex:gcl}) or $\textsf{mGCL}$ (\Cref{sec:sep}). The program semantics and atomic assertions are based on distributions over program states $\mu\in\mathcal{D}\Sigma$. The semantics for the sampling command is defined in terms of variable assignment. This allows us to abstract over the type of program states. 
\[
c ::= x\samp\eta
\hspace{6em}
\de{x\samp\eta}(\sigma) = \textsf{bind}(\eta, \lambda v. \de{x := v}(\sigma))
\]

\begin{definition}[Probabilistic Outcome Assertions]\label{def:proboc}
Given some satisfaction relation on program states ${\vDash_\Sigma}\subseteq \Sigma\times\prop$, we instantiate the outcome assertion logic (\Cref{def:outcome}) with the BI frame $\langle \mathcal D\Sigma, +, =, \varnothing\rangle$ such that atomic assertions have the form $\PP[A] = p$ where $p\in [0,1]$, $A\in\prop$, and:
\[
\mu\vDash \PP[A] = p
\qquad\text{iff}\qquad
|\mu| = p \quad\text{and}\quad \forall\sigma\in\mathsf{supp}(\mu).~\sigma\vDash_\Sigma A
\]
\end{definition}

Intuitively, the assertion $\PP[A] = p$ states that the outcome $A$ occurs with probability $p$. As a shorthand for under-approximate assertions, we also define $\PP[A] \ge p$ to be $(\PP[A]= p)\oplus\top$ (see \Cref{lem:problb} for a semantic justification).

We will now investigate falsification of probabilistic assertions of the form $\bigoplus_{i=1}^n(\PP[A_i] = p_i)$. In general, any such sequence can be falsified by specifying the precise probabilities of all combinations of the outcomes $A_i$. In the special case where $n=2$, $\mu\not\vDash (\PP[A] =p \oplus \PP[B] = q)$ iff:
\[
\mu\quad\vDash\quad
 \PP[A\land B] = p_1
 \quad\oplus\quad
 \PP [A\land\lnot{B}] = p_2
 \quad\oplus\quad
 \PP[\lnot{A}\land B] = p_3
 \quad\oplus\quad
 \PP[\lnot A\land\lnot B] = p_4
\]
Such that $p_4>0$ or $p_2 > p$ or $p_3>q$ or $p_1+p_2+p_3 \neq p+q$. The more general version of this result shows that $2^n$ outcomes are needed to disprove an assertion with $n$ outcomes \Acref{lem:probastfls,lem:probfls}, which is infeasible for large $n$. However, there are several special cases that require many fewer outcomes. For example, if all the $A_i$s are pairwise disjoint, then falsification can be achieved with just $n+1$ outcomes. Below, we use $\vec q$ to denote a vector of probabilities $q_1,\ldots,q_n$.

\begin{restatable}[Disjoint Falsification]{theorem}{probfalse}
\label{thm:probfalse}
First, let $A_0 = \bigwedge_{i=1}^n\lnot A_i$. If all the events are disjoint (for all $i\neq j$, $A_i \land A_j$ iff $\fls$), then:
\[
\not\vDash\triple{\varphi}C{\bigoplus_{i=1}^n (\PP[A_i] = p_i)}
\qquad\text{iff}\qquad
\exists \vec{q}, \varphi' \Rightarrow\varphi.\quad
\vDash\triple{\varphi'} C{\bigoplus_{i=0}^n (\PP[A_i] = q_i)}
\]
Such that $\mathsf{sat}(\varphi')$ and $q_0 \neq 0$ or for some $i$ $q_i\neq p_i$.

\end{restatable}

Many specifications fall into this disjointness case since the primary way in which proofs split into multiple probabilistic outcomes is via sampling, which always splits the postcondition into disjoint outcomes with the sampled variable $x$ taking on a unique value.

The correctness of some probabilistic programs is specified using lower bounds. For example, we may want to specify that some good outcome occurs \emph{with high probability}. These assertions can also be falsified using a lower bound.

\begin{restatable}[Principle of Denial for Lower Bounds]{theorem}{lbfalse}
\label{thm:lbfalse}
\[
\text{If}\quad
\exists \varphi' \Rightarrow\varphi.\quad
\mathsf{sat}(\varphi'),\quad
\vDash\triple{\varphi'}C{\PP[\lnot A] \ge q}
\quad\text{then}\quad
\not\vDash\triple\varphi C{\PP[A] \ge p}
\quad(\text{where}~q > 1-p)
\]
\end{restatable}

Note that this implication only goes one way, since the original specification $\PP[A]\ge p$ could be satisfied by a sub-distribution $\mu$ where $|\mu|<1$ and therefore $(\PP_\mu[A] \not\ge p) \centernot\Rightarrow (\PP_\mu[\lnot A] > 1-p)$. There are many more special cases for probabilistic falsification, but the relevant cases for the purposes of this paper fall into the categories discussed above.

\section{Outcome Logic for Memory Errors}
\label{sec:incorrectness}

In this section we specialize OL to prove the existence of memory errors in nondeterministic programs. The program logic is constructed in four layers. First, at its core, there is an assertion logic for describing heaps in the style of separation logic (\Cref{sec:heaps}). On top of that, we build an assertion logic with the capability of describing error states and multiple outcomes (\Cref{sec:er}). Then, we define the execution model using a monad combining both errors and nondeterminism (\Cref{sec:em}). Finally, we provide proof rules for this multi-layered logic (\Cref{sec:sep}).

We use this logic in \Cref{sec:pushback} to reason about memory errors in the style of Incorrectness Separation Logic~\cite{isl}. We also discuss why the semantics of Outcome Logic is a good fit for this type of bug finding by examining manifest errors in more depth (\Cref{sec:manifest}).



\subsection{Heap Assertions}
\label{sec:heaps}

First, we create a syntax of logical assertions to describe the heap in the style of Separation Logic~\cite{sl}. In order to describe why a program crashed, we need \emph{negative heap assertions} in addition to the standard points-to predicates. These assertions, denoted $x\not\mapsto$, state that the pointer $x$ is invalidated~\cite{isl}. The syntax for the heap assertion logic is below.
\begin{equation}\label{eq:sl}
p\in\textsf{SL} ::= \textbf{emp}\mid \exists x.p\mid p\land q\mid p\lor q\mid p\Rightarrow q \mid p \sep q\mid p \wand q \mid e\mid e_1\mapsto e_2\mid e\mapsto-\mid e\not\mapsto
\end{equation}
In this syntax, $e\in\textsf{Exp}$ is an expression which includes $\textsf{true}$ and $\textsf{false}$.
We add logical negation $\lnot p$ as shorthand for $p\Rightarrow\textsf{false}$. These assertions are satisfied by a stack and heap pair $(s,h)\in\mathcal{S}\times\mathcal{H}$. Stacks are defined as before (\Cref{ex:gcl}) and heaps are partial functions from positive natural numbers (addresses) to program values or bottom $\mathcal{H} = \{h\mid h\colon\mathbb{N}^+\rightharpoonup\textsf{Val}+\{\bot\}\}$\footnote{Note that $\ell\notin\mathsf{dom}(h)$ indicates that we have no information about the pointer $\ell$ whereas $h(\ell) = \bot$ indicates that $\ell$ is deallocated. This is why $h$ is both partial \emph{and} includes $\bot$ in the co-domain.}.
The constant $\textsf{null}$ is equal to 0, so it is not a valid address and therefore $\textsf{null}\notin\textsf{dom}(h)$ for any heap $h$. The semantics of $\mathsf{SL}$ is defined in \Aref{app:sl-assert} and is similar to that of \citet{isl}. 

\subsection{Reasoning about Errors}\label{sec:er}

While most formulations of Hoare Logic focus only on safe states, descriptions of error states are a fundamental part of Incorrectness Logic~\cite{il}. Reasoning about errors is built into the semantics of incorrectness triples and the underlying programming languages. In the style of Incorrectness Logic, we use $(\ok:p)$ and $(\er:p)$ to indicate whether or not the program terminated successfully. In our formulation, however, these are regular assertions rather than part of the triples themselves. This makes our assertion logic more expressive because we can describe programs that have multiple outcomes---some of which are successful and some erroneous---in a single triple. The semantics of programs that may crash is also encoded as a monadic effect.

\begin{definition}[Assertion logic with errors]\label{def:erassert}
Given a set of error states $E$, a set of program states $\Sigma$, and the relations ${\vDash_E} \subseteq E\times\prop_E$ and  ${\vDash_\Sigma}\subseteq \Sigma\times\prop_\Sigma$, we construct a new assertion logic with semantics ${\vDash}\subseteq (E+\Sigma)\times (\prop_\Sigma \times\prop_E)$ defined below:
\[
\inj_L(e) \vDash (p, q)
\quad\text{iff}\quad
e\vDash_E q
\qquad\qquad
\inj_R(\sigma)\vDash (p,q)
\quad\text{iff}\quad
\sigma\vDash_\Sigma p
\]
\end{definition}
In the above, $\inj_L\colon E\to E+\Sigma$ and $\inj_R\colon \Sigma\to E+\Sigma$ are the left and right injections, respectively. We also add syntactic sugar $(\ok:p) \triangleq (p, \fls)$ and $(\er:q) \triangleq (\fls, q)$, so in general the assertion $(p,q)$ can be thought of as $(\ok:p)\vee(\er:q)$.
Additional logical operations ($\lnot$, $\land$, and $\lor$) are defined in \Aref{sec:erlogic}.
We now combine errors with separation logic as follows:

\begin{definition}[Separation Logic with Errors]\label{heaps}
We define an assertion logic as follows:
\begin{itemize}
\item The syntax of basic assertions $\prop$ is given in \Cref{def:erassert} with $\prop_E= \prop_\Sigma=\textsf{SL}$, the heap assertion logic (\ref{eq:sl}). So, $\prop$ has the syntax $(\ok:p)$ and $(\er:q)$ where $p,q\in\mathsf{SL}$.
\item $\Sigma$, the set of program states, is given by $\mathcal S \times \mathcal H$.
\item The satisfaction relation is also given in \Cref{def:erassert} with $E = \Sigma$, so ${\vDash}\subseteq (\Sigma+\Sigma)\times\prop$.
\end{itemize}
\end{definition}

\subsection{Execution Model}\label{sec:em}

We will now create an execution model supports both nondeterminism and errors by combining the powerset monad (\Cref{def:ndeval})  with an error monad. We begin by defining the error monad, which is based on taking a coproduct with a set $E$ of errors. In order to use errors in conjunction with another effect (\ie nondeterminism), we define a monad transformer~\cite{transformers}. This is valid since the error monad composes with all other monads~\cite{composingmonads}. 

\begin{definition}[Execution model with errors]
Given some execution model $\langle M, \textsf{bind}_M, \textsf{unit}_M, \monoid, \ident\rangle$, we define a new execution model $\langle M(E+ -), \textsf{bind}_\er, \textsf{unit}_\er, \monoid, \ident\rangle$ such that:
\[
\textsf{bind}_\er(m, k) = \textsf{bind}_M\left(m, \lambda x.\left\{\begin{array}{ll}
k(y) & \text{if}~ x = \inj_R(y)\\
\textsf{unit}_M(x) & \text{if}~x = \inj_L(y)
\end{array}\right.\right)
\hspace{3em}
\textsf{unit}_\er(x) = \textsf{unit}_M(\inj_R(x))
\]
\end{definition}
\noindent Note the monoid definition ($\monoid$ and $\ident$) remains the same as the original execution model. For example, if the outer monad is powerset, we still use set union and empty set in the same way---errors only exist within a single outcome.
\begin{example}[Execution model for nondeterminism and errors]\label{ex:em}
We are particularly interested in the above definition when $M$ is the powerset monad, \ie $\bb{2}^{(-)}$. This results in an execution model 
$\langle \bb{2}^{E+ -}, \textsf{bind}, \textsf{unit}, \monoid, \ident\rangle$ where the operations are derived as follows:
\[\small
 \textsf{bind}(S, k) = \{\inj_L(x)\mid  \inj_L(x) \in S\} \cup \smashoperator{\bigcup_{ \inj_R(x)\in S} }k(x)
 \hspace{3em}
 \textsf{unit}(x) = \{\inj_R(x)\}
 \]
\end{example}

\noindent We now turn to defining atomic commands for manipulating the heap in a language called the Guarded Command Language with Memory (\textsf{mGCL}). The syntax for \textsf{mGCL} is given below and the semantics is in \Aref{app:sl-prog}. Note that \textsf{mGCL} commands are deterministic and can therefore be interpreted in both nondeterministic and probabilistic evaluation models.
\[
c\in\textsf{mGCL} ::= \textsf{assume}~e\mid x:=e \mid x := \textsf{alloc}() \mid\textsf{free}(e)\mid x\leftarrow [e] \mid [e_1] \leftarrow e_2 \mid\textsf{error}()
\]
Assume and assignment are the same as in \textsf{GCL} (\Cref{ex:gcl}). The usual heap operations for allocation ($\textsf{alloc}$), deallocation ($\textsf{free}$), loads ($x\leftarrow[e]$), and stores ($[e_1]\leftarrow e_2$) are also included along with an \textsf{error} command that immediately fails.  We also define $x:=\mathsf{malloc}()$ as syntactic sugar for $(x:=\textsf{alloc}()) + (x:=\textsf{null})$, which is valid in nondeterministic evaluation models.

\begin{definition}[Outcome-Based Separation Logic]\label{def:msl}
We instantiate OL (\Cref{def:mht}) with: 
\begin{enumerate}
\item The execution model is from \Cref{ex:em} with $E=\mathcal S \times \mathcal H$.
\item The language of atomic commands is \textsf{mGCL}. 
\item The assertion logic is the one given in \Cref{def:ndoc} using \Cref{heaps} for basic assertions.
\end{enumerate}
\end{definition}
\noindent Note that although the execution model has been augmented with errors, the nondeterministic falsification result (\Cref{thm:ndfalse}) still holds for Outcome-Based Separation Logic.

\subsection{Proof Rules for Memory Errors}
\label{sec:sep}

Now that we have defined the semantics of OL triples that can express properties about memory and errors, let us turn to the proof theory. In this section, we will define proof rules for Outcome-Based Separation Logic (\Cref{def:msl}), which we will use in subsequent sections to prove that programs crash due to memory errors.
We will define these proof rules in a way that is \emph{generic} with respect to the execution model, leveraging the fact that the atomic \textsf{mGCL} commands are deterministic, and thus can be given specifications that hold good under multiple different execution models (\eg nondeterminism or probabilistic computation).

Concretely, let us observe that the semantics of \textsf{mGCL} is based on the composition of two monads: an outer monad $M$ (\eg powerset), and the error monad $E+-$.
Since the atomic commands of \textsf{mGCL} are deterministic, however, their semantics is agnostic to the choice of the outer monad $M$, and can be specified axiomatically without needing to talk explicitly about (multiple) outcomes.
Hence, we define a new type of triple that is capable of making assertions about errors (using $\er$ and $\ok$), but says nothing about outcomes (using $\oplus$):
%
\begin{definition}[Liftable Triples]\label{def:pure}
Consider an OL instance based on the composition of two monads $M = M_1 \circ M_2$, and so $\de{C}\colon \Sigma\to (M_1\circ M_2)\Sigma$ (note that any monad can be decomposed in this way, by taking $M_2 = \mathsf{Id}$).
One such example is the execution model from \Cref{ex:em} where $M_1 = \bb{2}^{(-)}$ and $M_2 = E+-$.
The validity of an OL triple that is \emph{liftable} into the monad $M_1$ is defined as follows:
\[
\vDash_{M_1}\!\!\triple pCq \quad\text{iff}\quad \forall \sigma\in M_2\Sigma.\quad \sigma\vDash p
~
\Rightarrow
~
\exists\tau\in M_2\Sigma.~
 \dem{C}{\unit_{M_1}(\sigma)} = \unit_{M_1}(\tau)
~\text{and}~
\tau\vDash q
\]
Intuitively, this triple says that $C$ is deterministic; if we run it on any individual state that satisfies $p$, then the result will be an individual state satisfying $q$. In the case of \Cref{ex:em}, this means that $p$ and $q$ describe elements of $E+\Sigma$; they can describe error states (using $\er$ and $\ok$), but cannot use $\oplus$.
Similarly, we write $\vdash_{M_1}\!\!\triple pCq$ to denote a liftable derivation, which is sound with respect to the above semantics and can be lifted into the monad $M_1$.

\end{definition}

\begin{figure}
{\footnotesize
\vspace{1em}
\begin{flushleft}
\fbox{Separation Logic Small Axioms}
\end{flushleft}
\[
\inferrule{\;}
{\vdash_M\!\triple{\ok:p}{\textsf{error}()}{\textsf{er}:p}}
{\textsc{Error}}
\hspace{4em}
\inferrule{\;}
{\vdash_M\!\triple{\ok:x=v\land\emp}{x := \textsf{alloc}()}{\ok:x\mapsto -}}
{\textsc{Alloc}}
\]
\[
\inferrule{\;}
{\vdash_M\!\triple{\ok:e\mapsto -}{\textsf{free}(e)}{\ok:e\not\mapsto}}
{\textsc{Free Ok}}
\hspace{4em}
\inferrule{\;}
{\vdash_M\!\triple{\ok:e\not\mapsto}{\textsf{free}(e)}{\er:e\not\mapsto}}
{\textsc{Free Er}}
\]

\[
\inferrule{\;}
{\vdash_M\!\triple{\ok:e_1\mapsto -}{[e_1]\leftarrow e_2}{\ok:e_1\mapsto e_2}}
{\textsc{Store Ok}}
\hspace{4em}
\inferrule{\;}
{\vdash_M\!\triple{\ok:e_1\not\mapsto}{[e_1]\leftarrow e_2}{\textsf{er}:e_1\not\mapsto}}
{\textsc{Store Er}}
\]

\[
\inferrule{\;}
{\vdash_M\!\triple{\ok:x=m\land e\mapsto n}{x\leftarrow [e]}{\ok:x=n\land e[m/x]\mapsto n}}
{\textsc{Load Ok}}
\quad
\inferrule{\;}
{\vdash_M\!\triple{\ok: e\not\mapsto}{x\leftarrow[e]}{\textsf{er}:e\not\mapsto}}
{\textsc{Load Er}}
\]


\[
\inferrule
{\vdash_M\!\triple{\epsilon:p}c{\epsilon':q} \\ \textsf{fv}(r)\cap\textsf{mod}(c) = \emptyset}
{\vdash_M\!\triple{\epsilon:p\sep r}c{\epsilon': q\sep r}}
{\textsc{Frame}}
\]
\begin{flushleft}
\fbox{Monadic Rules}
\end{flushleft}
\[
\inferrule
{\vdash_{\bb{2}^{(-)}}\!\!\triple pCq }
{\triple pCq}
{\textsc{Nondeterministic Lifting}}
\qquad
\inferrule
{\;}
{\vdash_M\!\!\triple{\textsf{er}:p}C{\textsf{er}:p}}
{\textsc{Error Propagation}}
\]
}
\caption{Proof rules for Outcome-Based Separation Logic, the OL instantiation of \Cref{def:msl}. The first group is inspired by \cites{localreasoning} small axioms with additional rules added for unsafe states. The second group deal with the monadic execution model.}
\label{fig:sepproof}
\end{figure}

\Cref{fig:sepproof} contains the proof rules for Outcome-Based Separation Logic (\Cref{def:msl}). The first group of rules is very close to the standard separation logic proof system originally due to \citet{localreasoning}, with the addition of rules to reason about unsafe states inspired by \citet{isl}. These rules are liftable into any monad $M$ (since errors compose with all other monads). The \textsc{Lifting} proof rule states that if some triple is liftable into the powerset monad (where $p$ and $q$ are satisfied by individual states), then we can obtain a new triple where $p$ and $q$ are satisfied by sets of states (as in \Cref{def:ndoc}). All of the small axioms above can be lifted in this way.

In order to use the proof rules for conditionals and assignment from \Cref{fig:baserules}, we also define expression entailment and substitution.
Both operations are only defined for $\ok$ assertions.
\[
(\ok:p)\vDash e \quad\text{iff}\quad p\Rightarrow e
\qquad \qquad
(\ok:p)[e/x] \triangleq \ok:p[e/x]
\]
This means that, for example, the \textsc{Assign} rule only allows us to prove $\triple{\ok:p[e/x]}{x:=e}{\ok:p}$. If an error has occurred, we instead use the \textsc{Error Propagation} rule to propagate the error forward through the proof (\ie $\triple{\er:p}{x:=e}{\er:p}$), since the program will never recover from the crash.

\subsection{Proof of a Bug}
\label{sec:pushback}

We now demonstrate that the OL proof system shown in \Cref{fig:sepproof} is effective for bug-finding. The program in Figure~\ref{fig:pushback} has a possible use-after-free error. This program first appeared as a motivating example for ISL~\cite{isl}. It models a common error in C++ when using the \texttt{std::vector} library. A call to \texttt{push\_back} can reallocate the vector's underlying memory buffer, in which case pointers to that buffer become invalid.

As in \citet{isl}, we model the vector as a single heap location, and the \texttt{push\_back} function nondeterministically chooses to either reallocate the buffer or do nothing. A subsequent memory access may then fail, as seen in the main function. Since our language does not have procedures, we model these as macros and prove the existence of the bug with all the code inlined. The proof mostly makes use of standard separation logic proof rules and is quite similar to the ISL version~\cite{isl} especially in the use of negative heap assertion after the call to free. Under-approximation is achieved using the rule of consequence to drop one of the outcomes.


\begin{figure}[t]
\[
\arraycolsep=2pt
\begin{array}{l|l}
\begin{array}{l}
\textsf{main}():\\
\hspace{1em}x \leftarrow [v]\fatsemi\\
\hspace{1em}\texttt{push\_back}(v)\fatsemi\\
\hspace{1em}{[x]}\leftarrow 1\\\\\\
\textsf{push\_back}(v):\\
\hspace{1em}\begin{array}{lll}
\left(\begin{array}{l}
y \leftarrow [v]\fatsemi\\
\textsf{free}(y)\fatsemi\\
y := \textsf{alloc}()\fatsemi\\
{[v]} \leftarrow y
\end{array}
\right)
&+&
\skp
\end{array}
\end{array}
\hspace{2em}\;
&
\hspace{2em}
\footnotesize
\begin{array}{l}
{\color{ForestGreen}\langle \ok:v \mapsto a \sep a \mapsto - \rangle}\\
x \leftarrow [v]\fatsemi\\
{\color{ForestGreen}\langle \ok:v \mapsto x \sep x \mapsto - \rangle}\\
\begin{array}{lll}
\left(
\begin{array}{l}
{\color{ForestGreen}\langle \ok:v \mapsto x \sep x \mapsto - \rangle}\\
y \leftarrow [v]\fatsemi\\
{\color{ForestGreen}\langle \ok:v \mapsto x \sep x \mapsto - \land y=x\rangle}\\
\textsf{free}(y)\fatsemi\\
{\color{ForestGreen}\langle \ok:v \mapsto x \sep x \not\mapsto \land y=x\rangle}\\
y := \textsf{alloc}()\fatsemi\\
{\color{ForestGreen}\langle \ok:v \mapsto x \sep x \not\mapsto \sep y \mapsto-\rangle}\\
{[v]} \leftarrow y\\
{\color{ForestGreen}\langle \ok:v \mapsto y \sep x \not\mapsto \sep y \mapsto-\rangle}\\
\end{array}
\right)
&+&
\begin{array}{l}
{\color{ForestGreen}\langle \ok:v \mapsto x \sep x \mapsto - \rangle}\\
\skp\\
{\color{ForestGreen}\langle \ok:v \mapsto x \sep x \mapsto - \rangle}\\
\end{array}
\end{array}\fatsemi\\
{\color{ForestGreen}\langle (\ok:v \mapsto y \sep x \not\mapsto \sep y \mapsto-) \oplus (\ok:v \mapsto x \sep x \mapsto-)\rangle }\\
{[x]}\leftarrow 1\\
{\color{red}\langle (\er:v \mapsto y \sep x \not\mapsto \sep y \mapsto-) \oplus (\ok:v \mapsto x \sep x \mapsto 1) \rangle}\\
{\color{red} \implies\langle (\er: {x \not\mapsto} \sep \tru ) \oplus \top \rangle}
\end{array}
\end{array}
\]
\caption{Program with a possible use-after-free error (left) and proof sketch (right)}
\label{fig:pushback}
\end{figure}

Correctness for this program would be given by the postcondition $(\ok: v\mapsto x\sep x\mapsto 1)$. As \Cref{thm:ndfalse} showed, we can disprove it by showing that an undesirable outcome will sometimes occur. In this case, that undesirable outcome is $(\er:{x\not\mapsto} \sep \tru)$. Clearly, $(\er:{x\not\mapsto}\sep\tru)\oplus\top$ implies $\lnot (\ok: v\mapsto x\sep x\mapsto 1)$, so the specification in \Cref{fig:pushback} disproves the correctness specification.

\subsection{Manifest Errors}
\label{sec:manifest}

\citet{realbugs} showed empirically that the fix rates of bug-finding tools can be improved by reporting only those bugs that occur regardless of context. These errors are known as \textit{manifest errors}, as demonstrated in the examples below.
$$
\triple{\ok:x\not\mapsto}{[x]\leftarrow 1}{\textsf{er}: x\not\mapsto}
\qquad
\triple{\ok:\tru}{x:=\textsf{malloc}()\fatsemi [x]\leftarrow 1}{(\textsf{er}: x=\textsf{null})\oplus\top}
$$
The left program has a \emph{latent} error since it is only triggered if the pointer is already deallocated, therefore it would not be reported. The right program has a \emph{manifest} error since it is possible no matter the context in which the program is invoked. \citet[Def. 3.2]{realbugs} give the following definition for manifest errors:
\[
\vDash\inc {\ok:p}C{\er:q} \quad \text{is a manifest error}
\qquad\text{iff}\qquad
\forall \sigma.\quad \exists\tau\in\de{C}(\sigma).\quad \tau\vDash (\er:q\sep\tru)
\]
First note that the precondition $p$ does not appear in the formal definition. This indicates that IL preconditions do not meaningfully describe the conditions \emph{sufficient} to reach an end state. In addition, the universal quantification over the precondition resembles Hoare Logic more closely than Incorrectness Logic (which quantifies over the postcondition).
As stated in the following lemma, the formal definition of a manifest error can be expressed as an OL triple.

\begin{restatable}[Manifest Error Characterization]{lemma}{manifest}\label{lem:manifest}
\[
\vDash\inc{p}C{\er:q} \quad\text{is a manifest error}
\qquad\qquad\text{iff}\qquad\qquad
\vDashD\triple{\ok:\tru}C{\er:q\sep\tru}
\]
\end{restatable}

Following from this result, determining whether a program has a manifest error is equivalent to proving an OL triple of the form above.
Characterizing a manifest error using IL is much harder. \citet{realbugs} provide an algorithm to do so, which involves several satisfiability checks (which are \textsf{NP}-hard). The difficulty in characterizing manifest errors suggests that under-approximation in IL is too powerful. To see this, we compare the standard \textsc{If} rule from OL to \textsc{One-Sided If}---a hallmark of IL which allows the analysis to only consider one branch of an if statement.
\[\small
\inferrule{
\triple{\ok:p\land e}{C_1}{\epsilon:q}
\\
\triple{\ok:p\land\lnot e}{C_2}{\epsilon:q}
}
{\triple{\ok:p}{\iftf e{C_1}{C_2}}{\epsilon:q}}
{\textsc{If}}
\qquad\quad
\inferrule{\inc{p\land e}{C_1}{\epsilon:q}}
{\inc{p}{\iftf{e}{C_1}{C_2}}{\epsilon:q}}
{\textsc{One-Sided If}}
\]
\textsc{One-Sided If} generates imprecise preconditions since the precondition of the premise ($p\land e$) is stronger than the precondition of the conclusion ($p$). OL, on the other hand, requires the precondition to be precise enough to force the execution down a specific logical path, otherwise both paths must be considered as seen in the \textsc{If} rule. As such, OL enables under-approximation in \emph{just the right ways}; only outcomes that result from nondeterministic choice can be dropped.

\cites{realbugs} discussion of manifest errors suggests that \emph{sufficient} preconditions are important; we need to know what happens when we run the program on \emph{any} state satisfying the precondition. Interestingly, there is no analogous motivation for covering the whole postcondition (as IL does). Reachability is important, but we only have to reach \emph{some} error state, not \emph{all} of them. In fact, as we will see in our exploration of probabilistic programming, covering the entire post is often infeasible.

\section{Probabilistic Incorrectness}
\label{sec:prob}

Randomization is a powerful tool that is seeing increased adoption in mainstream software development as it is essential for machine learning and security applications. The study of probabilistic programming has a rich history~\cite{psem, ppdl}, but there is little prior work on proving that probabilistic programs are \emph{incorrect}. In \Cref{sec:probfalse}, we gave a theoretical result showing that probabilistic specifications in OL can be disproven. In this section, we provide a proof system for probabilistic OL and use it to prove incorrectness in a particular example program.

We work with the probabilistic OL instance using the evaluation model from \Cref{def:probeval} and the outcome assertions in \Cref{def:proboc}. The basic commands are assignment and $\mathsf{assume}$ from \textsf{GCL} (\Cref{ex:gcl}) with probabilistic sampling added ($x\samp\eta$). There are only two proof rules for the probabilistic language, given in \Cref{fig:probproof}. The \textsc{Lifting} rule allows us to lift a derivation (e.g. for variable assignment) into a probabilistic setting. This is sound, since every state in the support must transition from $A$ to $B$, thus $\PP[A]$ before running $C$ is equal to $\PP[B]$ after. The \textsc{Sample} rule splits the postcondition into a separate outcome for each value in the support of $\eta$.

The rules for conditional branching in \Cref{fig:baserules} can be used in probabilistic proofs by defining expression entailment $(\PP[A]=p)\vDash e$ iff $A\vDash e$. \textsc{Assign} can also be used; substitution propagates inside the probabilistic assertion $(\PP[A]=p)[e/x] = (\PP[A[e/x]] = p)$.
Note that the conditional rules require us to know the probability that the guard is true or false upfront. This is standard for probabilistic Hoare Logics~\cite{ellora, hartog}.

\begin{figure}
\[\footnotesize
\inferrule
{\vdash_{\mathcal D}\!\!\triple ACB}
{\triple{\PP[A] = p}C{\PP[B] = p}}
{\textsc{Lifting}}
\hspace{4em}
\inferrule
{\forall v\in\textsf{supp}(\eta).~\triple{A}{x:=v}{B_v}}
{\triple{\PP[A] = p}{x \samp\eta}{\smashoperator{\bigoplus_{v\in\textsf{supp}(\eta)}}(\PP[B_v] = p\cdot\eta(v))}}
{\textsc{Sample}}
\]\vspace{-1em}
\caption{Probabilistic proof rules.
}
\label{fig:probproof}
\end{figure}

Absent are rules for while loops. Looping rules in probabilistic languages are complex since invariants cannot be used when probabilities change across iterations. Such proof rules are certainly expressible in our model, but are out of scope for this paper. For examples of how this is done, see \citet{ellora, hartog}.

\subsection{Error Bounds for Machine Learning}

Randomization is often used in approximation algorithms where computing the exact solution to a problem is difficult. In these applications, some amount of error is acceptable as long as it is likely to be small. One such application is supervised learning algorithms, which produce a \textit{hypothesis} from a set of labelled examples. The examples are members of some set $X$ and are drawn randomly from some probability distribution $\eta\in\mathcal{D}X$. The hypothesis is a function $h : X \rightarrow \mathbb B$ which guesses whether new data points are positive or negative examples.


Consider the simple learning problem in which we want to learn a point $t\in [0,1]\subset\mathbb R$ on the unit interval. Since we require distributions used in programs to be finite, we can approximate $[0,1]$ as $\{ k\cdot\Delta\mid 0\le k \le \frac1\Delta\}$ for some finite \emph{step size} $\Delta$. Anything in the interval $[0, t]$ is considered a positive example, and anything greater than $t$ is a negative example. This concept is illustrated at the top of Figure~\ref{fig:interval} and the program below---expressed in a probabilistic extension of \textsf{GCL}---learns this concept by repeatedly sampling examples and refining the hypothesis $h$ after each round. The resulting hypothesis is always equal to the largest positive example that the algorithm has seen. Therefore it will always classify negative examples correctly and only make mistakes on positive examples between $h$ and $t$.

\begin{figure}
\begin{minipage}{.3\textwidth}
\begin{tikzpicture}[scale=.5, thick]
    \draw (-1,0)-- (7,0); 
        \draw (0,0.1) -- (0,-0.1) node[below] {0};
        \draw (6,0.1) -- (6,-0.1) node[below] {1};

    {\color{NavyBlue}
    \draw (0,.25) -- (4,.25);
    \fill (0,.25) circle (0.25);
    \fill (4,.25) circle (0.65) node {$\color{white}t$};}
    {\color{ForestGreen}
    \draw (0,.5) -- (2,.5);
    \fill (2,.5) circle (0.65) node {$\color{white}h$};
    \fill (0,.5) circle (0.25);}
    
        \draw [ultra thick, decorate,
    decoration = {calligraphic brace}] (2,1.25) --  (4,1.25) ;
    \draw [dotted] (4, .9) -- (4, 1.25);
    \fill (3, 2) node {$\textsf{er}(h)$};
    
        \end{tikzpicture}
        \vspace{2em}
\[
\begin{array}{l}
h := -1\fatsemi\\
\textsf{for}~N~\textsf{do}\\
\hspace{1em}x \samp [0, 1]\fatsemi\\
\hspace{1em}\textsf{if}~\mathcal{L}(x) \land x > h~\textsf{then}\\
\hspace{2em}h := x\\
\hspace{1em}\textsf{else}\\
\hspace{2em}\skp
\end{array}
\]
\end{minipage}\hfill\vline
\begin{minipage}{.68\textwidth}
{\scriptsize\[
\begin{array}{l}
{\color{purple}\langle \PP[\tru] = 1 \rangle \implies\langle \PP[\er(-1) > \varepsilon] \ge (1-\varepsilon)^{0} \rangle}\\
\hspace{.5em}h := -1~\fatsemi\\
{\color{purple}\langle \PP[\er(h) > \varepsilon] \ge (1-\varepsilon)^{0} \rangle}\\
\hspace{.5em}\textsf{for}~N~\textsf{do}\\
\hspace{1em}{\color{purple}\langle \PP[\er(h) > \varepsilon] \ge (1-\varepsilon)^{i} \rangle}\\
\hspace{1.5em}x \samp [0, 1]\fatsemi\\
\hspace{1em}{\color{purple}\langle \bigoplus_{q\in[0,1]}\PP[\er(h) > \varepsilon \land x = q] \ge \Delta\cdot(1-\varepsilon)^{i} \rangle}\\
\hspace{1.5em}\textsf{if}~\mathcal{L}(x) \land x > h~\textsf{then}\\
\hspace{2em}{\color{purple}\langle \bigoplus_{q\in(h,t]}\PP[\er(h) > \varepsilon \land x = q] \ge \Delta\cdot(1-\varepsilon)^{i} \rangle}\\
\hspace{2em}{\color{purple}\implies\langle \bigoplus_{x\in(h, t - \varepsilon]}\PP[\er(x) > \varepsilon ] \ge \Delta\cdot(1-\varepsilon)^{i} \rangle}\\
\hspace{2em}{\color{purple}\implies\langle \PP[\er(x) > \varepsilon ] \ge (t-\varepsilon - h)\cdot(1-\varepsilon)^{i} \rangle }\\
\hspace{2.5em}h := x\\
\hspace{2em}{\color{purple}\langle \PP[\er(h) > \varepsilon ] \ge (t-\varepsilon - h)\cdot(1-\varepsilon)^{i} \rangle }\\
\hspace{1.5em}\textsf{else}\\
\hspace{2em}{\color{purple}\langle \bigoplus_{q\in[0,h]\cup(t,1]}\PP[\er(h) > \varepsilon \land x = q] \ge \Delta\cdot(1-\varepsilon)^{i} \rangle}\\
\hspace{2.5em}\skp\fatsemi\\
\hspace{2em}{\color{purple}\langle \bigoplus_{q\in[0,h]\cup(t,1]}\PP[\er(h) > \varepsilon \land x = q] \ge \Delta\cdot(1-\varepsilon)^{i} \rangle}\\
\hspace{2em}{\color{purple}\implies\langle \PP[\er(h) > \varepsilon ] \ge (h + 1-t) \cdot(1-\varepsilon)^{i} \rangle}\\
\hspace{1em}{\color{purple}\langle \PP[\er(h) > \varepsilon ] \ge (t-\varepsilon-h)\cdot(1-\varepsilon)^{i} \oplus \PP[\er(h) > \varepsilon ] \ge (h + 1-t) \cdot(1-\varepsilon)^{i} \rangle}\\
\hspace{1em}{\color{purple}\implies\langle \PP[\er(h) > \varepsilon ] \ge (1-\varepsilon)^{i+1} \rangle}\\
\hspace{0em}{\color{purple}\langle \PP[\er(h) > \varepsilon ] \ge (1-\varepsilon)^N \rangle}\\
\end{array}
\]}
\end{minipage}
\caption{The interval learning problem: a diagram of the learning problem (top left), a program implementing interval learning (bottom left), and a proof sketch (right).}
\label{fig:interval}
\end{figure}

The \textit{labelling} oracle $\mathcal{L}(x) = x \le t$ gives the true label of any point on the interval. Let $\textsf{er}(h) = t-h$ be the error of the hypothesis (the total probability mass between $h$ and $t$). The goal is to determine the probability that $h$ has error greater than $\varepsilon$ after $N$ iterations. Practically speaking, this simulates training the model on a dataset of size $N$. Intuitively, the error will be less than $\varepsilon$ if the algorithm ever samples an example in the interval $[t-\varepsilon, t]$. The chance of getting a hit in this range increases greatly with the number of examples seen. While this problem may seem contrived, it is a 1-dimensional version of the \textit{Rectangle Learning Problem} which is known to have practical applications and the proof ideas are extensible to other learnable concepts~\cite{colt}.


To prove that this program is \emph{correct}, we want to say that the resulting hypothesis has small error with high probability. Choosing an error bound $\varepsilon$ and a confidence parameter $\delta$, we say that the program is correct if at the end $\PP[\er(h) \le \varepsilon] \ge 1-\delta$. Now, we can look to \Cref{thm:lbfalse} to determine how to \emph{disprove} the correctness specification. We need to show that the probability of the opposite happening ($\er(h) > \varepsilon$) is higher than $\delta$. Based on the derivation in \Cref{fig:interval}, we conclude that the program is incorrect if $(1-\varepsilon)^N > \delta$. Suppose we had a dataset of size $N=100$ and desired at most 1\% error ($\varepsilon = 0.01$) with 90\% likelihood ($\delta = 0.1$). Then the postcondition tells us that the error is higher than 1\% with probability at least 37\%. Clearly $37\% > \delta$, so the program is incorrect; we need a larger dataset in order to get a better result.

\subsection{Probabilistic Incorrectness Logic}
\label{sec:prob-incorrectness}

It is natural to ask whether a similar result could be achieved using a probabilistic variant of Incorrectness Logic. However, such a program logic is cumbersome and produces poor characterizations of errors. To show this, we begin by examining the semantics of a probabilistic IL triple.
\[
\vDash\inc PCQ
\qquad
\text{iff}
\qquad
\forall \mu\vDash Q. \quad \exists\mu'. \quad\mu \sqsubseteq \dem{C}{\mu'} \quad\text{and}\quad \mu'\vDash P
\]
This definition differs from standard Incorrectness Logic in two ways. First, assertions are satisfied by \emph{distributions} over program states $\mu\in\mathcal{D}\Sigma$ rather than individual program states $\sigma\in\Sigma$. This is necessary in order to make the assertion logic quantitative. Second, under-approximation is achieved using the sub-distribution relation $\sqsubseteq$ instead of set inclusion\footnote{
This order is defined pointwise: $\mu_1 \sqsubseteq\mu_2$ iff $\forall x. \mu_1(x)\le\mu_2(x)$.
}. As is typical with Incorrectness Logic, this definition stipulates that \textit{any} subdistribution satisfying the postcondition must be reachable by an execution of the program. While in non-probabilistic cases it can already be hard to fully characterize a valid end-state, even more information is needed in the probabilistic case.

To demonstrate this, consider the interval learning program from Figure~\ref{fig:interval}. The postcondition of this triple is $\PP[\er(h) >\varepsilon] \ge (1-\varepsilon)^N$, which is not a valid postcondition for an incorrectness triple because it does not adequately describe the final distribution. That is, there are many distributions satisfying this assertion that could not result from running the program. In one such distribution, $h=-1$ with probability 1. So, lower bounds are not suitable for use in Incorrectness Logic because a distribution can be invented where the probability is arbitrarily large, rendering it unreachable. But changing the inequality to an equality to obtain $\PP[\er(h) >\varepsilon] = (1-\varepsilon)^N$ does not solve the problem. This assertion can be satisfied by a distribution where $h=-1$ with probability $(1-\varepsilon)^N$, which is also unreachable. In order for an assertion to properly characterize the output distribution, it has to specify all the possible values of $h$. Such an assertion is given below:
\[
\bigoplus_{x=0}^\varepsilon\left(
\PP[\er(h)=x] = (1-x)^N - (1-(x+\Delta))^N
\right)
\]
The original assertion was easy to understand; we immediately knew the probability of having a large error. By contrast, the added information needed for IL actually obscures the result. It is not useful to know the probability of each value of $h$, we only care about bounding the probability that $\er(h)>\varepsilon$.  In general, Probabilistic Incorrectness Logic requires us to specify the \emph{entire joint distribution} over all the program variables which is certainly undesirable and often infeasible.

Many techniques in probabilistic program analysis summarize the output distribution in alternative ways. This includes using expected values~\cite{wpe,kaminski} and probabilistic independence~\cite{psl}. If those techniques are used to express correctness, it makes sense that similar ideas would be desirable for incorrectness. However, techniques that \emph{summarize} a distribution are incompatible with Incorrectness Logic since they do not specify the output distribution in a sufficient level of detail. Based on these findings, we conclude that developing probabilistic variants of Incorrectness Logic is not a promising research direction. In fact, the differences between correctness and incorrectness are often quite blurred in probabilistic examples. Since some amount of error is typically expected, it is not possible to reason about correctness \emph{without} reasoning about incorrectness. It is therefore sensible that a unified theory captures both.

\section{Related Work}
\label{sec:related}

\noindent\textit{\textbf{Incorrectness reasoning and program analysis}}.
In motivating Incorrectness Logic (IL), \citet{il} posed the twin challenges of \emph{sound} and \emph{scalable} incorrectness reasoning: program logics for incorrectness must guarantee \emph{true positive} bugs, while also supporting \emph{under-approximation} in order to scale to large codebases.
Outcome Logic (OL) takes inspiration from those challenges, but offers a solution that is closer to traditional Hoare Logic~\cite{hoarelogic} and, as such, is also compatible with correctness reasoning.

Outcome Logic was also inspired in part by \emph{Lisbon triples}, which were first described in a published article by \citet[\S{5}]{ilalgebra} under the name \emph{backwards under-approximate triples}.\footnote{Though \citet[\S{3.2}]{realbugs} also mention backwards under-approximate triples, their potential has gone largely unexplored.}
The semantics of Lisbon triples is based on \cites{wpp} calculus of \emph{possible correctness}:
for any initial state satisfying the precondition, there exists \emph{some} trace of execution leading to a final state satisfying the postcondition.
As such, Lisbon triples describe true positives (behaviors that are witnessed by an actual trace, assuming the pre is satisfiable).
As recounted by \citet[\S{7}]{il}, Lisbon triples predate Incorrectness Logic; Derek Dreyer and Ralf Jung suggested them as a foundation for incorrectness reasoning during a discussion with Peter O'Hearn and Jules Villard that took place at POPL'19 in Lisbon (hence the name ``Lisbon Triples'').

Shortly thereafter, O'Hearn developed the semantics of IL triples.
His major motivation for developing IL (instead of further exploring Lisbon triples) was the goal of finding a logical foundation for \emph{scalable} bug-catching static analysis tools (such as Pulse-X~\cite{realbugs}), and one key to scalability is the ability to discard program paths (aka ``drop disjuncts'') during analysis.
More concretely, the analysis accumulates a disjunction of assertions which symbolically represents the set of possible states at each program point. If this set gets too large, then it is important to be able to drop some of the disjuncts in order to save memory and computation time.
Thanks to its reverse rule of consequence---which supports \emph{strengthening} of the postcondition---IL provides a sound logical foundation for dropping disjuncts, whereas Lisbon triples do not.



One can see OL as a generalization of Lisbon triples which supports discarding of program paths in a different way than IL does: namely, via the \emph{outcome conjunction} connective, which enables reasoning about multiple executions at the same time.\footnote{In \Aref{app:triple-proofs}, we show that Lisbon triples are in fact a special case of OL.}
Specifically, ``disjuncts'' arise in a program analysis when the program makes a \emph{choice} to branch based on either a logic condition (\eg an if statement or while loop) or a computational effect (\eg nondeterminism or randomization). 
In IL, both types of choice are encoded by standard disjunction.
In OL, on the other hand, we distinguish these two forms of choice by using disjunction ($\vee$) for the former and outcome conjunction ($\oplus$) for the latter.
This leads to a different approach for supporting discarding of program paths, but one which we believe can serve as an alternative logical foundation for practical static analyses.

Let us first consider the case of choices arising from computational effects.
Incorrectness Logic includes a \textsc{Choice} rule that allows analyses to drop one branch of a nondeterministic choice. An analogous derived rule is also sound in OL \Acite{app:uxrules}; both are shown below.
\[
\inferrule{\inc{P}{C_1}{Q}}{\inc{P}{C_1+C_2}{Q}}{\textsc{Choice (IL)}}
\qquad\qquad
\inferrule{\triple{P}{C_1}{Q}}{\triple{P}{C_1+C_2}{Q\oplus\top}}{\textsc{Under-Approx (OL)}}
\]
Given that nondeterministic variants of OL provide reachability guarantees, it may appear surprising that a conclusion about $C_1+C_2$ can be made without showing that $C_2$ terminates. However, the assertion $\top$ encompasses all outcomes (including nontermination), so this inference is valid.
Note that there are also symmetric versions of these rules where the $C_2$ branch is instead taken.

Let us now consider the case of choices arising from logical conditions, where
the differences between OL and IL are more pronounced. Consider the following program, which will only fail in the case that $b$ is true.
\[
\triple{\ok: x\mapsto -}
{\iftf b{\mathsf{free}(x)}\skp \fatsemi [x] \leftarrow 1}
{(\ok:x\mapsto 1) \vee (\er:x\not\mapsto)}
\]
The semantics of OL does not permit us to simply drop one of the disjuncts in the postcondition. If we want to only explore the program path in which the error occurs, then we need to push information about the logical condition $b$ backwards into the precondition.
\[
\triple{\ok: {x\mapsto-} \land b}
{ \iftf b{\mathsf{free}(x)}\skp \fatsemi [x] \leftarrow 1}
{\er:x\not\mapsto}
\]
This is in contrast to Incorrectness Logic, in which we \emph{can} drop disjuncts, \emph{but in return} we need to ensure that every state described by the postcondition is reachable. More precisely, $(\er:x\not\mapsto)$ is not a strong enough IL postcondition for the aforementioned program because it includes the unreachable state in which $x\not\mapsto$, but $b$ is false. In IL, one must therefore specify the bug as follows:
\[
\inc{x\mapsto-}
{\iftf b{\mathsf{free}(x)}\skp \fatsemi [x] \leftarrow 1}
{\er:{x\not\mapsto} \land b}
\]
So, in either case we must record the same amount of information about the logical condition $b$. The difference is whether this information appears in the pre- or postcondition. As we discussed in \Cref{sec:manifest}, there are advantages to having a more precise precondition (as OL does): it enables us to easily determine how to trigger a bug and characterize manifest errors. Conversely, the precise postconditions required by IL make it difficult to design abstract domains, suggesting that IL is not compatible with popular analysis techniques like abstract interpretation~\cite{ascari}.

Furthermore, in order to generate more useful bug reports and error traces for the user, practical static analysis tools like Pulse-X~\cite{realbugs} \emph{do} in any case push logical conditions backwards to the pre-condition using a technique called bi-abduction~\cite{biab09,biabduction}. This suggests that while the theories of OL and IL differ substantially, it may be possible to build practical static analysis tools atop OL in a similar manner to IL-based tools like Pulse-X. We plan to investigate this further in future work.



\smallskip
\noindent\textit{\textbf{Unifying correctness and incorrectness}}.
Parallel efforts have been made to unify correctness and incorrectness reasoning within a single program logic. \citet{Bruni2021ALF,brunijacm} introduced Local Completeness Logic (LCL), which is based on Incorrectness Logic, but with limits on the rule of consequence such that an over-approximation of the reachable states can always be recovered from the postcondition. Similarly, Exact Separation Logic (ESL)~\cite{exactsl} combines the semantics of IL and Hoare Logic in triples that exactly describe the reachable states.

Both of these logics are capable of proving correctness properties as well as finding true bugs. But they achieve this by compromising the ability to use the rule of consequence, which is crucial to scalable analysis algorithms. Analyses based on Hoare Logic use consequences to abstract the postcondition, reducing the information overhead and aiding in finding loop invariants. Analyses based on IL use consequences to drop disjuncts and consider fewer program paths. Since neither type of consequence is valid in LCL and ESL, it remains unclear whether those theories can feasibly serve as the foundation of practical tools. By contrast, Outcome Logic enjoys the full power of the (forward) rule of consequence and can also drop nondeterministic paths.

There has also been work to connect the theories of correctness and incorrectness algebraically using Kleene Algebra with Tests (KAT) \cite{kat}, an equational theory for reasoning about program equivalence. \citet{ilalgebra,inckleene} showed that both Hoare Logic and IL can be embedded in variants of KAT and used this insight to formalize connections between the two types of specifications. While this provides an algebraic theory powerful enough to capture Hoare Logic and IL, this connection does not go as deep as the unification offered by OL and does not provide a clear path to shared analyses for both program verification and bug finding.

Since our paper was conditionally accepted to OOPSLA, a closely related paper has appeared on arXiv, which presents a program logic, called Hyper Hoare Logic, for proving and disproving program hyper-properties (properties relating multiple program traces)~\cite{hyperhoare}. It achieves this using the same underlying semantics as Outcome Logic instantiated to the powerset monad. Their work shows the applicability of the OL model beyond the usage scenarios that we envisioned in this paper.

\smallskip
\noindent\textit{\textbf{Separation logic and Iris}}.
While both separation logic~\cite{localreasoning,sl} and Outcome Logic employ Bunched Implications~\cite{bi} as a fundamental part of their metatheories, the way in which BI is used in each case is substantially different.

In separation logic and its extensions such as Iris~\cite{iris1,iris}, the value of the BI resource monoid is neatly demonstrated by the \textsc{Frame} Rule, which enables local reasoning by adding assertions about unused resources to the pre- and postconditions of some smaller proof derivation. In this way, framing allows us to talk about the same program execution with additional (unused) resources. By contrast, the outcome conjunction deals with assertions about \emph{different} program executions.

The \textsc{Frame} Rule is in general unsound with respect to the outcome conjunction.
To demonstrate this, we use the same counterexample that \citet{sl} used
to demonstrate that the
Rule of Constancy is unsound in separation logic:
\[
\inferrule{
\triple{x\mapsto -}{[x]\leftarrow 4}{x\mapsto 4}
}{
\triple{x\mapsto - \oplus y\mapsto 3}{[x] \leftarrow 4}{x\mapsto 4 \oplus y\mapsto 3}
}{\textsc{Frame}}
\]
It is easy to see that this is an invalid inference. The outcome conjunction
does not preclude that $x$ and $y$ are aliased, in which case it must be that
$y\mapsto 4$ in the postcondition. Instead, we have the \textsc{Split} rule (\Cref{fig:baserules}), which allows us to analyze a program separately for each outcome in the precondition and then compose the resulting outcomes in the postcondition.

This example shows that, although both
separation logic and OL use BI, the two logics are modeling two very different aspects of
the program (resource usage vs.\ program outcomes, respectively), and the resulting program logics are therefore different.

OL and separation logic are not mutually exclusive. In \Cref{sec:incorrectness}, we saw how separation logic can be embedded in OL. In addition, we believe that combining OL with Iris is a very interesting direction for future research: Iris offers advanced mechanisms to reason modularly about concurrency, and OL offers a way to extend Hoare Logic to be amenable to both correctness and incorrectness reasoning. Combining the two would result in a program logic capable of proving the existence of bugs in concurrent programs (while a concurrent version of Incorrectness Logic already exists \cite{cisl}, it is not built atop Iris and does not support the full capabilities offered by Iris).

In a concurrent version of Outcome Logic, outcomes would model possible interleavings of concurrent branches. In an assertion of the form $P\oplus\top$, the predicate $P$ could describe an undesirable outcome that occurs in \emph{some} of those interleavings (\ie a bug), which is not currently possible to express in Iris.

\smallskip
\noindent\textit{\textbf{Probabilistic and quantitative program analysis}}.
Probabilistic variants of Hoare Logic \cite{ellora,hartog,rand2015vphl,polaris} were a major source of inspiration for the design of Outcome Logic. Whereas pre- and postconditions of standard Hoare Logic describe individual program states, probabilistic variants of Hoare Logic use assertions that describe \emph{distributions} over program states. These logics also include connectives similar to the outcome conjunction, but specialized to probability distributions. In Outcome Logic, we generalize from probability distributions to support a wider variety of PCMs. 


Starting with the seminal work of \citet{psem,ppdl}, expected values have been a favorite choice for probabilistic program analysis. \citet{wpe}'s weakest-pre-expectation ($\wpe$) calculus computes expected values of program expressions with an approach similar to \cites{Dijkstra76} Weakest Precondition calculus. Many extensions to $\wpe$ have arisen, including to handle nondeterminism, runtimes~\cite{kaminski}, and Separation Logic~\cite{qsl}. This line of work has not intersected with Incorrectness Logic since the semantics of weakest-pre is incompatible with IL, although \citet{qsl} hinted at the nuanced interaction between correctness and incorrectness in quantitative settings with their ``faulty garbage collector'' example. We hope that our new perspective---using Hoare Logic for incorrectness---will encourage the use of $\wpe$ calculi for bug-finding.

\citet{qsp} developed a Quantitative Strongest Post (QSP) calculus and noted its connections to IL, which was originally characterized by \citet{il} in terms of \cites{Dijkstra76} strongest-post. QSP is an interesting foundation for studying the Galois Connections between types of quantitative program specifications, although the goals are somewhat orthogonal to our own in that we sought to \emph{unify} correctness and incorrectness rather than explore dualities.

\section{Conclusion}
\label{sec:conclusion}

Formal methods for incorrectness remain a young field. The foundational work of \citet{il} has already led to several program logics for proving the existence of bugs such as memory errors, memory leaks, data races, and deadlocks~\cite{isl,cisl,realbugs}. However, as with any new field there are growing pains---manifest errors and probabilistic programs are an awkward fit in the original formulation of IL.
This has inspired us to pursue a new theory incorporating \cites{il} core tenets of incorrectness---true positives and under-approximation---while also accounting for more evaluation models and different types of incorrectness. Outcome Logic achieves just that, with the added benefit of unifying the theories of correctness and incorrectness in a single program logic. Our Falsification Theorem (\Cref{thm:falsification}) shows that any OL triple can be disproven within the logic. So, any bug invalidating a correctness specification can be expressed.  OL also offers a cleaner characterization of manifest errors, suggesting it may be semantically closer to the way that programmers reason about bugs. 

In this paper, we introduced OL as a theoretical basis for incorrectness reasoning, but in the future we plan to further explore its practical potential as well. Incorrectness Logic has been shown to scale well as an underlying theory for bug-finding in large part due to its ability to \emph{drop disjuncts}~\cite{isl,realbugs}; analysis algorithms accumulate a disjunction of possible outcomes as they move forward through a program, and due to the semantics of IL, these disjuncts can be soundly pruned to keep the search space small. Hoare Logics (including OL) cannot drop disjuncts. However, as we saw in \Cref{sec:overview} and \Cref{sec:logic}, OL \emph{can} drop \emph{outcomes}, which we believe is sufficient to make the algorithm scale to large codebases (although this remains to be demonstrated). Furthermore, since OL triples can be used both for correctness and incorrectness reasoning, we plan to develop a bi-abductive~\cite{biab09,biabduction} algorithm to infer procedure summaries that can be used by \emph{both} correctness verification and bug-finding analyses.


When \citet{il} remarked that ``program correctness and incorrectness are two sides of the same coin,'' he was expressing that just as programmers spend significant mental energy debugging (reasoning about \emph{incorrectness}), we in the formal methods community must invent sound reasoning principles for incorrectness. We take this idea one step further, suggesting that program correctness and incorrectness are two \emph{usages} of the same \emph{program logic}. We hope that this unifying perspective will continue to invigorate the field of incorrectness reasoning and invite the reuse of tools and techniques that have already been successfully deployed for correctness reasoning.

\section*{Acknowledgments}

We thank Peter O'Hearn, Josh Berdine, Azalea Raad, Jules Villard, Quang Loc Le, and Julien Vanegue for their helpful feedback.
This work has been supported in part by the Defense Advanced Research Projects Agency under Contract HR001120C0107.

\bibliographystyle{ACM-Reference-Format}
\bibliography{mhl}

\ifx\extended\undefined\else
\input{appendix}
\fi

\end{document}
\endinput

%% file: appendix.tex
\newpage

\appendix

\section{Totality of Language Semantics}
\label{app:totality}

As mentioned in \Cref{sec:prelim}, the semantics of the language in \Cref{fig:comlang} can be made total in all the execution models that we use (nondeterministic and probabilistic), despite depending on the partial monoid operator ($\monoid$). In this section we discuss restrictions that must be placed on probabilistic languages in order to make the semantics total and also establish the existence of the least fixed point used in the semantics of $C^\star$.

Regardless of the execution model, proving the fixed point existence requires us to prove that the semantic map $\de{-}^\dagger$ is continuous with respect to some partial order. We remark that a preorder can be generically defined in terms of the monoid operation $m_1 \sqsubseteq m_2$ iff there exists $m$ such that $m_1 \monoid m = m_2$. In both the nondeterministic and probabilistic case, this relation is also anti-symmetric, therefore it is a partial order. In fact, in the case of the powerset monad, $\sqsubseteq$ is equivalent to $\subseteq$.

We also introduce the notion of syntactic validity for a program $C$. For example, the use of expressions must be well-typed. That is, if $\assume{e}$ appears in the program, then $e$ must be boolean valued, \ie $\forall \sigma. \de{e}(\sigma)\in\mathbb{B}$.

\subsection{Nondeterministic Languages}

Since the monoid operation for nondeterminsitic languages is set union (a total function), we can allow unrestricted access to $C_1+C_2$ and $C^\star$. Therefore, to ensure totality, we must only prove that the least fixed point exists.

\begin{lemma}[Fixed point existence]\label{lem:lfp}
For any semantics of atomic commands, the function $F(f)(\sigma) = f^\dagger(\de{C}(\sigma))\monoid \unit(\sigma)$ has a least fixed point when specialized to the powerset monad.
\end{lemma}
\begin{proof}
We first note that in the lemma statement $f\colon \Sigma\to \bb{2}^\Sigma$ and $\de{C}\colon \Sigma\to \bb{2}^\Sigma$. We also define the point-wise partial order $f_1\sqsubseteq f_2$ iff $\forall x. f_1(x)\subseteq f_2(x)$. Clearly, the function $\lambda x.\emptyset$ is the bottom of this order. This also means that for any non-empty chain $f_1\sqsubseteq f_2 \sqsubseteq \ldots$ it must be that $\bigsqcup_i f_i = \lambda x.\bigcup_i f_i(x)$.
 We now show that $F$ is Scott continuous:
{\small
\begin{align*}
F\left(\bigsqcup_{i}f_i\right) ~&= \lambda \sigma.\left(\bigsqcup_{i}f_i\right)^\dagger(\de{C}(\sigma))\cup \{\sigma\} \\
&= \lambda \sigma. \left(\bigcup_{\tau\in\de{C}(\sigma)} \left(\bigsqcup_{i}f_i\right)(\tau)\right) \cup \{\sigma\} \\
&= \lambda \sigma. \bigcup_{\tau\in\de{C}(\sigma)} \left(\bigcup_{i}f_i(\tau)\cup \{\sigma\} \right)\\
&= \lambda \sigma.\bigcup_{i}\left(f^\dagger_i(\de{C}(\sigma))\cup \{\sigma\}\right) \\
&= \lambda \sigma.\bigcup_{i} F(f_i)(\sigma)\\
&= \bigsqcup_{i} F(f_i)
\end{align*}}
Therefore, by the Kleene Fixed Point Theorem, $\textsf{lfp}(F) = \bigsqcup_{n\in\mathbb{N}} F^n(\lambda x.\emptyset)$.
\end{proof}

\subsection{Probabilistic Languages}

In probabilistic languages, we can ensure totality using simple syntactic checks. That is, we syntactically limit programs to not use $C_1+C_2$ and $C^\star$, but rather the guarded versions as shown in \Cref{ex:gcl}. In addition, we establish that $\bind$ is total. Since $\bind$ is implemented as a sum, we must ensure that the cumulative probability mass of the summands does not exceed 1. This is easy to see:
\[\small
\left|\bind_\mathcal{D}(\mu, f)\right| = \left|\sum_{\sigma\in\textsf{supp}(\mu)} \mu(\sigma) \cdot f(\sigma)\right| = \smashoperator{\sum_{\sigma\in\textsf{supp}(\sigma)}} \mu(\sigma)\cdot |f(\sigma)|
\le \smashoperator{\sum_{\sigma\in\textsf{supp}(\sigma)}} \mu(\sigma) = |\mu|
\]
Since $f : \Sigma \to \mathcal{D}\Sigma$, then for any $\sigma\in\Sigma$, $|f(\sigma)| \le 1$.
Therefore, we have shown that $|\bind(\mu, f)| \le |\mu|$ ($\bind$ is contractive) and since it cannot add probability mass it must be total.

\begin{lemma}[Totality of Probabilistic Language Semantics]
The function $\de{C}\colon \Sigma\to\mathcal{D}(\Sigma)$ is total subject to the syntactic restrictions on $C$ described above.
\end{lemma}
\begin{proof}
The proof is by induction on $C$. All of the cases except if statements and while loops trivially follow from the definition of $\de{-}$.
\begin{itemize}
\item\textsc{If}. First note that $\de{\iftf e{C_1}{C_2}}(\sigma) = \de{(\assume{e}\fatsemi C_1) + (\assume{\lnot e}\fatsemi C_2)}(\sigma) = \de{\assume{e}\fatsemi C_1}(\sigma) + \de{\assume{\lnot e}\fatsemi C_2}(\sigma)$. Now, we do case analysis on the value of $\de{e}_\textsf{Exp}(\sigma)$. If $\de{e}_\textsf{Exp}(\sigma)=\tru$, then $\de{\assume e}(\sigma) = \delta_\sigma$ and $ \de{\assume{\lnot e}}(\sigma) = \varnothing$. Therefore, we know that $\de{\assume{e}\fatsemi C_1}(\sigma) = \de{C_1}(\sigma)$ and $\de{\assume{\lnot e}\fatsemi C_2}(\sigma)=\varnothing$. By the induction hypothesis $\de{C_1}(\sigma)$ is defined and so $\de{C_1}(\sigma)+\varnothing$ must also be defined. The case where $\de{e}_\textsf{Exp}(\sigma) = \fls$ is symmetrical.
\item\textsc{While}. We begin by proposing an alternate semantics for (guarded) while loops:
{\small
\[
\de{\whl eC}(\sigma) = \textsf{lfp}(F)(\sigma)
\quad\text{where}\quad
F(f)(\sigma) = f^\dagger(\de{\assume{e}\fatsemi C}(\sigma))\monoid \de{\assume{\lnot e}}(\sigma)
\]}
In this semantics, we push the $\assume{\lnot e}$ in to the fixed point computation which allows $\monoid$ to be defined. In the nondeterminism case where $\monoid$ is total, this semantics is equivalent to the one defined in \Cref{fig:comlang}. Now, note that when using the partial order described at the beginning of this section, the supremum of two distributions (if it exists) is $\mu_1\sqcup\mu_2 = \lambda\sigma.\max(\mu_1(\sigma), \mu_2(\sigma))$. We can therefore see that addition distributes over the supremum:
{\small
\begin{align*}
(\mu_1 \sqcup \mu_2) + \mu ~&= (\lambda\sigma.\max(\mu_1(\sigma), \mu_2(\sigma))) + \mu\\
&= \lambda\sigma.\max(\mu_1(\sigma), \mu_2(\sigma))) + \mu(\sigma) \\
&= \lambda\sigma.\max((\mu_1+\mu)(\sigma), (\mu_2+\mu)(\sigma)))\\
&= (\mu_1+\mu) \sqcup (\mu_2+\mu)
\end{align*}}
We now proceed to prove that $F$ is Scott continuous.
We use the same point-wise order that we saw in \Cref{lem:lfp}, $f_1 \sqsubseteq f_2$ iff $\forall x.f_1(x) \sqsubseteq f_2(x)$.
{\small
\begin{align*}
F\left(\bigsqcup_{i}f_i\right) ~&= \lambda \sigma. \left(\bigsqcup_{i}f_i\right)^\dagger (\de{\assume e\fatsemi C}(\sigma)) + \de{\assume{\lnot e}}(\sigma) \\
&= \lambda\sigma.\smashoperator{\sum_{\tau\in\de{\assume e\fatsemi C}(\sigma)}} (\textstyle{\bigsqcup_i} f_i(\tau)) + \de{\assume{\lnot e}}(\sigma) \\
&= \lambda\sigma.\bigsqcup_i (f_i^\dagger(\de{\assume e\fatsemi C}(\sigma)) + \de{\assume{\lnot e}}(\sigma))\\
\intertext{Note that this sum is always defined since one of $\de{\assume{e}}(\sigma)$ or $\de{\assume{\lnot e}}(\sigma)$ must be $\varnothing$.}
&= \lambda \sigma.\bigsqcup_{i} F(f_i)(\sigma)\\
&= \bigsqcup_{i} F(f_i)
\end{align*}}
Therefore, by the Kleene Fixed Point Theorem, $\textsf{lfp}(F) = \bigsqcup_{n\in\mathbb{N}} F^n(\lambda x.\varnothing)$.
\end{itemize}
\end{proof}

\section{Under-Approximation}\label{app:ux}

In \Cref{def:uxoutcome}, we defined under-approximate outcome assertions $m\vDashD\varphi$ to be syntactic sugar for $m\vDash\varphi\oplus\top$. In order to motivate this choice, we prove the following results, which show that this definition of under-approximation corresponds to \emph{dropping outcomes}.  

\begin{lemma}[Dropping Outcomes]
In any BI frame, the following implications hold: $\varphi\oplus\psi \Rightarrow \varphi\oplus\top$ and $\varphi\oplus\psi \Rightarrow \top\oplus\psi$.
\end{lemma}
\begin{proof}
Suppose that $m\vDash\varphi\oplus\psi$. Then there exists $m_1$ and $m_2$ such that $m \succcurlyeq m_1\monoid m_2$ and $m_1\vDash\varphi$ and $m_2\vDash\psi$. Clearly, also $m_2\vDash\top$, so $m\vDash\varphi\oplus\top$. The second implication is symmetric.
\end{proof}

\begin{lemma}[Dropping Outcomes (Under-Approximate)]
In any BI frame, if $m\vDashD\varphi\oplus\psi$, then $m\vDash\varphi$.
\end{lemma}
\begin{proof}
Since $m\vDashD\varphi\oplus\psi$, then $m\vDash\varphi\oplus\psi\oplus\top$. This means that $m_1\vDash\varphi$, $m_2\vDash\psi$ and $m_3\vDash\top$ such that $m_1\monoid m_2\monoid m_3 \preccurlyeq m$. Clearly, $m_2\vDash\top$ as well. Recombining these, we get $m\vDash\varphi\oplus\top\oplus\top$ which is equivalent to $m\vDash \varphi\oplus\top$, or just $m\vDashD\varphi$.
\end{proof}

%
%

\subsection{Alternative Formulation using Intuitionistic BI}\label{app:intuitionistic}

In \Cref{sec:outcomes} we defined a single variant of the outcome logic using classical BI with under-approximate assertions as syntactic sugar. A different development is possible using an intuitionistic interpretation of BI with a preorder defined in terms of the monoid composition:
\[
m_1 \preccurlyeq m_2
\qquad\text{iff}\qquad
\left\{\begin{array}{ll}
m_2 = \varnothing & \text{if}\quad m_1 = \varnothing \\
\exists m.\; m_1 \monoid m = m_2 & \text{if} \quad m_1 \neq\varnothing
\end{array}\right.
\]
Note that the first case ensures that $\varnothing$ is only related to itself, which is necessary to ensure that $m\vDash\top^\oplus$ iff $m=\varnothing$. Atomic assertions in intuititionistic BI interpretations must respect the persistence property: if $m\vDash P$ and $m'\succcurlyeq m$, then $m'\vDash P$ (this is also referred to as monotonicity in Kripke semantics). We will now show that the under-approximate satisfaction relation $\vDashD$ is valid as an intuitionistic satisfaction relation for atomic propositions.

\begin{lemma}[Under-Approximate Satisfaction is Persistent]
For any $m, m'\in M\Sigma$ and atomic assertion $P$, if $m\vDashD P$ and $m'\succcurlyeq m$, then $m'\vDashD P$.
\end{lemma}
\begin{proof}
If $m = \varnothing$, then $m'$ is also $\varnothing$ and so clearly $m'\vDashD P$. Now suppose that $m\neq\varnothing$. Since $m\vDashD P$, then $m\vDash P\oplus \top$. Since $m'\succcurlyeq m$ and $m\neq\varnothing$, there is some $m''$ such that $m\monoid m'' = m'$. Clearly, $m''\vDash\top$, so $m'\vDash (P\oplus\top)\oplus\top$. This means that $m'\vDash P\oplus\top$, or in other words $m'\vDashD P$.
\end{proof}

If we combine the under-approximate satisfaction relation with the basic assertions for nondeterministic and probabilistic evaluation models (\Cref{def:ndoc,def:proboc}), we get a sensible semantics. As the following two lemmas show, under-approximation in the nondeterministic case corresponds to existential quantification and in the probabilistic case it corresponds to lower bounds.

\begin{lemma}
\label{lem:uxsat}
In the powerset interpretation of BI, $S\vDashD P$ iff $\exists \sigma\in S.~\sigma\vDash_\Sigma P$.
\end{lemma}
\iffpf{
Suppose that $S\vDashD P$, so $S\vDash P\oplus\top$, or in other words $S_1\vDash P$ and $S_2\vDash \top$ such that $S=S_1\cup S_2$. Further, this means that $S_1 \neq\emptyset$ and $\forall \sigma\in S_1.~\sigma\vDash_\Sigma P$. Since we know $S_1$ is nonempty, then there exists $\sigma\in S_1.~\sigma\vDash_\Sigma P$ and since $S_1\subseteq S$, then $\sigma\in S$ as well, so $\exists \sigma\in S.~\sigma\vDash_\Sigma P$.
}{
Suppose that $\exists\sigma\in S.~\sigma\vDash_\Sigma P$. Now, let $T = \{\sigma\}$, so clearly $T\vDash P$ and $S\vDash\top$ and $T\cup S = S$. Therefore, $S\vDash P\oplus\top$ and so $S\vDashD P$.
}

\begin{lemma}
\label{lem:problb}
In the distribution interpretation of BI, $\mu\vDashD (\PP[A] = p)$ iff $\PP_\mu[A] \ge p$, where:
\[
\PP_\mu[A] \triangleq \sum \{\mu(\sigma) \mid \sigma\in\mathsf{supp}(\mu), \sigma\vDash_\Sigma A \}
\]
\end{lemma}
\iffpf{
Assume that $\mu\vDashD (\PP[A] =p)$, so $\mu\vDash (\PP[A]=p)\oplus\top$. Therefore $\mu_1\vDash (\PP[A] = p)$ and $\mu_2\vDash\top$ such that $\mu_1+\mu_2 = \mu$. This tells us that $\PP_{\mu_1}[A] = p$. When we add $\mu_2$ to $\mu_1$ to get $\mu$, the probability of $A$ can only increase, so $\PP_\mu[A] \ge p$.
}{
Assume that $\PP_\mu[A] \ge p$. That means there must be a sub-distribution $\mu_1$ of $\mu$ such that $|\mu_1| = p$ and $\forall\sigma\in\mathsf{supp}(\mu_1).~\sigma\vDash_\Sigma A$. Let the other part of the distribution be $\mu_2$ (so $\mu=\mu_1+\mu_2$). Now, by construction, $\mu_1\vDash (\PP[A] = p)$ and $\mu_2\vDash\top$, so $\mu\vDash(\PP[A] = p)\oplus\top$, or equivalently $\mu\vDashD(\PP[A] = p)$.
}

\subsection{Derived Under-approximate Proof Rules}
\label{app:uxrules}

In this section, we provide derived inference rules that aid in reasoning about programs in an under-approximate manner. The first set of rules under-approximate nondeterministic choice by only exploring one of the paths and using the trivial post-condition $\top$ for the other path. Note that the unexplored path may diverge, $\top$ is a valid postcondition no matter what behavior it has.

\begin{lemma}
The following proof rules for under-approximating program paths can be derived for any nondeterministic Outcome Logic instance.
\[
\inferrule{\triple\varphi{C_1}\psi}{\triple{\varphi}{C_1+C_2}{\psi\oplus\top}}{\textsc{Under-Approx Left}}
\qquad
\inferrule{\triple\varphi{C_2}\psi}{\triple{\varphi}{C_1+C_2}{\psi\oplus\top}}{\textsc{Under-Approx Right}}
\]
\end{lemma}
\begin{proof}
We show the derivation for \textsc{Under-Approx Left} below. The derivation of \textsc{Under-Approx Right} is symmetric.
\[
\inferrule*[right=\textsc{Plus}]{
  \triple{\varphi}{C_1}\psi
  \\
  \inferrule*[right=\textsc{True}]{\;}{\triple{\varphi}{C_1}\top}
}
{\triple{\varphi}{C_1+C_2}{\psi\oplus\top}}
\]
\end{proof}

Once an under-approximate $-\oplus\top$ predicate has been introduced, it is also convenient to have inference rules that propagate it forward. The following two derived rules can be used to sequence under-approximate derivations together.

\begin{lemma} The following inference rules are derivable for any Outcome Logic instance.
\[
\inferrule{\triple\varphi{C}\psi}{\triple{\varphi\oplus\top}C{\psi\oplus\top}}{\textsc{Under-Approx Prop}}
\qquad
\inferrule{
  \triple{\varphi}{C_1}{\psi\oplus\top}
  \\
  \triple{\psi}{C_2}{\vartheta}
}
{\triple{\varphi}{C_1\fatsemi C_2}{\vartheta\oplus\top}}
{\textsc{Under-Approx Seq}}
\]
\end{lemma}
\begin{proof} These rules are derived as follows:
\[\footnotesize
\inferrule*[right=\textsc{Split}]{
  \triple\varphi{C}\psi
  \\
  \inferrule*[right=\textsc{True}]{\;}{\triple\top{C}\top}
}
{\triple{\varphi\oplus\top}C{\psi\oplus\top}}
\qquad
\inferrule*[right=\textsc{Seq}]{
  \triple{\varphi}{C_1}{\psi\oplus\top}
  \\
  \inferrule*[right=\textsc{Under-Approx Prop}]{
    \triple{\psi}{C_2}{\vartheta}
  }
  {\triple{\psi\oplus\top}{C_2}{\vartheta\oplus\top}}
}
{\triple{\varphi}{C_1\fatsemi C_2}{\vartheta\oplus\top}}
\]
\end{proof}


\section{Equivalence of Triples}
\label{app:triple-proofs}

In this section, we show that the nondeterministic instance of OL subsumes Hoare Logic~\cite{hoarelogic} and the Backward Under-Approximate Triples of \citet{ilalgebra} that were mentioned briefly in \Cref{sec:related}. We assume we have a nondeterministic program semantics over program states $\de{C}\colon\Sigma\to\bb{2}^\Sigma$ and an assertion logic where propositions $P,Q\in\prop$ are satisfied by program states, so $\mathord{\vDash_\Sigma}\subseteq \Sigma\times\prop$. Both of the aforementioned triple semantics are defined below where under-approximate triples use the notation $\lisbon PCQ$ due to \citet{realbugs}:
\[
\begin{array}{lllllllll}
\mktripple{\{}{\}} \quad\;& \text{iff}\quad\;& \forall \sigma\in\Sigma. & \sigma\vDash_\Sigma P &\Rightarrow  & {\color{purple}\forall}\tau\in \de{C}(\sigma). & \tau\vDash_\Sigma Q \\
\mktripple{\{\!|}{|\!\}} \quad\;& \text{iff}\quad\;& \forall \sigma\in\Sigma. & \sigma\vDash_\Sigma P &\Rightarrow & {\color{purple}\exists}\tau\in \de{C}(\sigma). & \tau\vDash_\Sigma Q
\end{array}
\]
Now, we will work with nondeterministic instances of OL using the evaluation model from \Cref{def:ndeval} and the logic of atomic assertions from \Cref{def:ndoc}. We let BI disjunctions $\varphi\vee\psi$ be syntactic sugar for $\lnot(\lnot \varphi \land \lnot\psi)$ (this encoding is typical in classical logics). We now prove our first result, that Hoare Triples are subsumed by OL. As we mentioned in \Cref{sec:mhldef}, since Hoare Triples are partial correctness specification, we have to use the postcondition $Q\vee\top^\oplus$ to express that $Q$ holds \emph{if} the program terminates\footnote{
Equivalent ways of expressing this include $\lnot Q\Rightarrow\top^\oplus$ (if $Q$ is false, then the program must diverge) or $\lnot\top^\oplus\Rightarrow Q$ (if the program terminates, then $Q$ holds). Alternatively, if we modified the semantics of atomic assertions (\Cref{def:ndoc}) to be $S\vDash P$ iff $\forall\sigma\in S.~\sigma\vDash_\Sigma P$ (without requiring that $S\neq\emptyset$), then we would have a more direct correspondence: $\vDash\hoare PCQ$ \,iff\, $\vDash\triple PCQ$, but then $P\oplus Q$ would behave more like $P\vee Q$, not guaranteeing reachability.
}.

\subhoare*
\iffpf{
Suppose that $\vDash\hoare PCQ$ or in other words, for any $\sigma\vDash_\Sigma P$ and $\tau\in\de{C}(\sigma)$, $\tau\vDash_\Sigma Q$. Now suppose that $S\vDash P$, or in other words, $S\neq\emptyset$ and $\forall\sigma\in S.~\sigma\vDash_\Sigma P$. Since $\vDash\hoare PCQ$, then for any $\tau\in\de{C}(\sigma)$, $\tau\vDash_\Sigma Q$. Now, we know that $\dem{C}{S} = \bind(S, \de{C}) = \bigcup_{\sigma\in S}\de{C}(\sigma)$. This means that for every $\tau\in\dem{C}{S}$, $\tau\vDash_\Sigma Q$. So, if $\dem CS \neq \emptyset$, then $\dem{C}{S}\vDash Q$. If $\dem CS = \emptyset$, then $\dem CS\vDash\top^\oplus$. Therefore $\dem CS\vDash Q\vee\top^\oplus$ and $\vDash\triple PC{Q\vee\top^\oplus}$.
}{
Suppose that $\vDash\triple PC{Q\vee\top^\oplus}$ and so if $S\vDash P$, then $\dem{C}{S}\vDash Q\vee\top^\oplus$. Now suppose that $\sigma\vDash_\Sigma P$. Then trivially $\{\sigma\}\vDash P$, so we can use our assumption to conclude that $\dem{C}{\{\sigma\}}\vDash Q\vee\top^\oplus$. This implies that $\forall\tau\in\de C(\sigma).~\tau\vDash_\Sigma Q$ (in the case where $\dem C{\{\sigma\}}\vDash\top^\oplus$, then $\de C(\sigma) = \emptyset$, so it holds vacuously).
}

Now, we will prove that OL triples subsume Backwards Under-Approximate Triples as well. This time, we use the under-approximate variant of OL which transforms the postcondition $Q$ into $Q\oplus\top$. This corresponds to existential quantification as we proved in \Cref{lem:uxsat}.

\begin{theorem}[Subsumption of Under-Approximate Triples]
\label{thm:ul}
$\mktripple{\{\!|}{|\!\}}$ \; iff \; $\vDashD\triple PC{Q}$
\end{theorem}
\iffpf{
Suppose that $\vDash\lisbon PCQ$ or in other words, for any $\sigma\vDash_\Sigma P$, there exists a $\tau\in\de{C}(\sigma)$ such that $\tau\vDash_\Sigma Q$. Now suppose that $S\vDash P$, or in other words $S\neq\emptyset$ and $\forall\sigma\in S.~\sigma\vDash_\Sigma P$. Pick one such $\sigma \in S$ (there must be at least one since $S\neq\emptyset$).
Since $\vDash\lisbon PCQ$, then $\exists\tau\in\de{C}(\sigma)$ such that $\tau\vDash_\Sigma Q$.
Since $\sigma\in S$, then $S = \{\sigma\}\cup S$ and therefore by linearity, $\dem{C}{S} = \dem{C}{\{\sigma\}\cup S} = \dem{C}{\{\sigma\}}\cup\dem{C}{S} = \de{C}(\sigma)\cup\dem{C}{S}$ and since $\tau\in\de{C}(\sigma)$ then $\tau\in\dem{C}{S}$ and so $\exists\tau\in \dem CS.~\tau\vDash_\Sigma Q$. By \Cref{lem:uxsat}, we can therefore conclude that $\dem CS\vDash Q\oplus\top$.
}{
Suppose that $\vDash\triple PC{Q\oplus\top}$ and so if $S\vDash P$, then $\dem{C}{S}\vDash Q\oplus\top$. Now suppose that $\sigma\vDash_\Sigma P$. Then trivially $\{\sigma\}\vDash P$, so we can use our assumption to conclude that $\dem{C}{\{\sigma\}}\vDash Q\oplus\top$. Now, by \Cref{lem:uxsat}, there is some $\tau\in\de{C}(\sigma)$ such that $\tau\vDash_\Sigma Q$.
}

The combination of \Cref{thm:ul} and \Cref{thm:mhlil} suggest that Backwards Under-Approximate triples can disprove any Hoare Triple as well (if the precondition $\varphi$ from \Cref{thm:ul} can be expressed as a basic assertion). \citet{ilalgebra} also stated this fact, although the proof was omitted.

\section{Falsification}

In this section we, prove the falsification results from \Cref{sec:falsification} of the main text. These theorems are inspired by that of \citet[Theorem 4.1]{ilalgebra}, who proved that if some Hoare triple is false $\not\vDash\hoare PCQ$, then there is some other Incorrectness triple $\vDash\inc{P'}C{Q'}$ that disproves it.
\[
\not\vDash\hoare PCQ
\qquad\text{iff}\qquad
\exists P', Q'.\quad
P'\Rightarrow P
\quad\text{and}\quad
Q' \centernot\Rightarrow Q
\quad\text{and}\quad
\vDash\inc{P'}C{Q'}
\]
The proof given by \citet{ilalgebra} is \emph{semantic}; it does not witness the construction of $P'$ and $Q'$ as \emph{syntactic} assertions. We give any analogous result in \Cref{sec:semtriples}. \Cref{thm:falsification} proves that any false OL triple (with semantic assertions) can be disproven by another OL triple.

While this result shows the strength of the OL model, we are also interested to know if our \emph{syntactic} assertion logic is powerful enough to express the pre- and postconditions needed to disprove other triples. We answer this question in the affirmative, although the result is less general. While the semantic proof in \Cref{sec:semtriples} applies to \emph{any} OL instance, the syntactic proofs rely on some additional properties of the particular evaluation model. We lay out the requirements for a falsifiable instance of OL in \Cref{app:synfalse} and prove that the nondeterministic and probabilistic instances are falsifiable in \Cref{app:ndfalse,app:probfalse} respectively.

\subsection{Falsification Proof with Semantic Assertions}\label{sec:semtriples}

We first introduce the notion of a semantic OL triple. A semantic assertion is simply a set of satisfying models. We will use the uppercase greek metavariables to denote semantic assertions $\Phi,\Psi\in\bb{2}^{M\Sigma}$. The semantic interpretation of a syntactic assertion is the set of models that satisfies it $\sem\varphi \triangleq \{ m \mid m\in M\Sigma, m\vDash\varphi \}$. Logical implication $\Phi\Rightarrow\Psi$ is given by set inclusion $\Phi\subseteq\Psi$. Note that $\Phi\Rightarrow\Psi$ is a proposition, \emph{not} a semantic assertion (\ie it is not a set). Negation is given by $\lnot\Phi = \bb{2}^{M\Sigma}\setminus \Phi$ and we say that an assertion is satisfiable $\mathsf{sat}(\Phi)$ iff $\Phi\neq\emptyset$. This gives us the following expected properties:
\[ 
\sem\varphi \Rightarrow \sem\psi \quad\text{iff}\quad \varphi\Rightarrow\psi
\qquad\qquad
m\in \lnot\Phi \quad\text{iff}\quad m\notin \Phi
\qquad\qquad
\mathsf{sat}(\sem\varphi)\quad\text{iff}\quad \exists m\in M\Sigma.~ m\vDash\varphi
\]
We also define the notion of a semantic OL triple as follows:
\[
\vDash_S\triple\Phi{C}\Psi
\qquad\text{iff}\qquad
\forall m\in M\Sigma.\quad m\in\Phi \quad\implies\quad \dem Cm\in\Psi
\]
The correspondence between semantic and syntactic triples is given by the following lemma.

\begin{lemma}[Equivalence of Semantic and Syntactic triples]\label{lem:strip}
If \;$\Phi = \sem\varphi$ and $\Psi=\sem\psi$, then:
\[
\vDash_S\triple{\Phi}C{\Psi}
\qquad\text{iff}\qquad
\vDash\triple \varphi{C}\psi
\]
\end{lemma}
\iffpf{
Suppose that $m\vDash\varphi$, then $m\in\sem{\varphi}=\Phi$, so using $\vDash_S\triple{\Phi}C{\Psi}$, we know that $\dem{C}m\in\Psi=\sem{\psi}$. This means that $\dem Cm\vDash\psi$, so $\vDash\triple{\varphi}C\psi$.
}{
Suppose that $m\in\Phi=\sem\varphi$, then it must be that $m\vDash\varphi$, so using $\vDash\triple{\varphi}C\psi$, we can conclude that $\dem Cm\vDash\psi$. This means that $\dem Cm\in\sem\psi = \Psi$, so therefore $\vDash_S\triple{\Phi}C{\Psi}$.
}
We can now prove the Semantic Falsification theorem and the Principle of Denial, which were introduced in \Cref{sec:falsification}.

\falsification*
\iffpf{
Assume that $\not\vDash_S\triple\Phi C\Psi$, so that means that there is some $m\in M\Sigma$ such that $m\in\Phi$ and $\dem Cm\notin\Psi$. By definition, this also means that $\dem Cm\in \lnot\Psi$. Now, let $\Phi' = \{m\}$, so clearly $\Phi' \Rightarrow\Phi$ (since $\Phi'\subseteq\Phi$) and $\mathsf{sat}(\Phi')$. To see that $\vDash_S\triple{\Phi'}C{\lnot\Psi}$, suppose that $m'\in\Phi'$. By construction, it must be that $m'=m$, so therefore $\dem C{m'}\in\lnot\Psi$ (since we already know that $\dem Cm\in\lnot\Psi$.
}{
Assume that there is some $\Phi'$ such that $\Phi'\Rightarrow\Phi$, $\mathsf{sat}(\Phi')$, and $\vDash_S\triple{\Phi'}C{\lnot\Psi}$. Then, there must be some $m\in\Phi'$ and so $m\in\Phi$ as well. Since $\vDash_S\triple{\Phi'}C{\lnot\Psi}$, then $\dem Cm\in\lnot\Psi$ and so $\dem Cm\notin\Psi$. We therefore know that $m\in\Phi$ and $\dem Cm\notin\Psi$, so $\not\vDash_S\triple\Phi C\Psi$.
}

\truepos*
\begin{proof}
Let $\Phi' = \sem{\varphi'}$, $\Phi = \sem\varphi$, and $\Psi = \sem\psi$. From our assumptions, we can conclude that $\Phi'\Rightarrow\Phi$ and $\mathsf{sat}(\Phi')$ and by \Cref{lem:strip} we can conclude that $\vDash_S\triple{\Phi'}C{\lnot\Psi}$. Therefore by \Cref{thm:falsification}, this implies that $\not\vDash_S\triple\Phi C\Psi$. Using \Cref{lem:strip} again, we conclude that $\not\vDash\triple\varphi C\psi$.

\end{proof}

\subsection{Falsification Proof with Syntactic Assertions}
\label{app:synfalse}

The syntactic version of the forward direction of the Falsification Theorem imposes more specific constraints on the assertions and execution model. We first lay out the general strategy for the proof, and then provide the formal details.

If we start with $\not\vDash\triple\varphi{C}\psi$, then we know that there exists some $m$ such that $m\vDash\varphi$ and $\dem Cm\not\vDash\psi$ and this implies that $\dem Cm\vDash\lnot\psi$ since we are working with classical interpretations of BI. Now, we have a single program execution starting at $\varphi$ and ending at $\lnot\psi$, and we would like to extrapolate a valid OL triple from this (possibly with a precondition stronger than $\varphi$ since the bad outcome $\lnot\psi$ may only occur under some more specific constraints).

We are going to do this by induction on the program $C$. However, in cases that involve choice (\eg $C=C_1+C_2$), we need to be able to split the postcondition into the components corresponding to the two choices ($C_1$ or $C_2$). This is possible, but only if the postcondition contains no implications. Logical negation is an implication ($\lnot\psi$ is shorthand for $\psi\Rightarrow\bot$), therefore we need a different postcondition $\psi'$ that implies $\lnot\psi$, but is syntactically valid. The precise form of $\psi'$ will depend on the BI instance.

In addition, the program $C$ must terminate after finitely many steps, otherwise the precondition that we generate may be infinitely large. Possible ways around this include using a fixed point logic, however we are not aware of any versions of BI that have a fixed point operator. Instead, we will assume going forward that every program terminates after finitely many steps.

In order to make the argument formal, we first introduce the notion of a falsifiable OL instance which adds the constraints needed to split assertions and extrapolate triples. Next, we prove a couple of intermediate lemmas before giving the main result. In the next sections, we will instantiate this result to the nondeterministic and probabilistic evaluation models.

\begin{definition}[Falsifiable OL]\label{def:fls}
 An instance of OL is falsifiable if it has the following properties:
 \begin{enumerate}
 \item The PCM operation has the properties:
 \begin{enumerate}
 \item If $m_1\monoid m_2 = \varnothing$, then $m_1=m_2=\varnothing$
 \item If $m_1\monoid m_2 = n_1\monoid n_1$, then there exist $s_1,s_2,t_1,t_2$ such that $s_1\monoid s_2 = n_1$, $t_1\monoid t_2 = n_2$, $s_1\monoid t_1 = m_1$ and $s_2\monoid t_2 = m_2$.
\end{enumerate}
 \item Atomic assertions $P$ are splittable, that is if $m_1\monoid m_2\vDash P$, then there exist $\varphi_1$ and $\varphi_2$ such that $m_1\vDash\varphi_1$ and $m_2\vDash\varphi_2$ and $\varphi_1\oplus\varphi_2 \Rightarrow P$.
 \item Sequences of outcomes are falsifiable, $m\not\vDash Q_1\oplus\cdots\oplus Q_n$, iff $\exists \psi$ containing no implications such that $m\vDash\psi$ and $\psi \Rightarrow \lnot\bigoplus_{i=1}^n Q_j$.
 \item Atomic commands have trace extrapolation, if $\dem cm\vDash\psi$, then there exists $\varphi$ such that $m\vDash\varphi$ and $\vDash\triple{\varphi}c\psi$ (where $\varphi$ and $\psi$ have no implications).
 \end{enumerate}
 \end{definition}

\begin{lemma}[Splitting]\label{lem:split}
For any BI assertion $\varphi$ that contains no implications and where the BI frame comes from a falsifiable OL instance, if $m_1\monoid m_2\vDash \varphi$, then there exist $\varphi_1$ and $\varphi_2$ such that $m_1\vDash\varphi_1$ and $m_2\vDash\varphi_2$ and $\varphi_1\oplus\varphi_2 \Rightarrow \varphi$.
\end{lemma}
\begin{proof}
By induction on the structure of $\varphi$ (\Cref{fig:outcomes}).
\begin{itemize}
\item $\varphi = \top$. Clearly $m_1\vDash\top$ and $m_2\vDash\top$ and $\top\oplus\top\Rightarrow\top$.
\item $\varphi = \bot$. Vacuous since $m_1\monoid m_2\vDash\bot$ is impossible.
\item $\varphi = \top^\oplus$. If $m_1 \monoid m_2\vDash\top^\oplus$, then it must be the case that $m_1 = m_2 = \varnothing$ (by property (1a) of \Cref{def:fls}). So, $m_1\vDash\top^\oplus$ and $m_2\vDash\top^\oplus$ and $\top^\oplus\oplus\top^\oplus\Rightarrow\top^\oplus$.
\item$\varphi = \psi'\land \psi $. We know that $m_1\monoid m_2\vDash \psi'\land\psi$, so $m_1\monoid m_2\vDash \psi'$ and $m_1\monoid m_2\vDash \psi$. By the induction hypotheses, There are $\varphi_1$, $\varphi_2$, $\psi_1$, and $\psi_2$ such that $m_1\vDash \varphi_1$ and $m_2\vDash\varphi_2$ and $m_1\vDash\psi_1$ and $m_2\vDash\psi_2$ and $\varphi_1\oplus\varphi_2 \Rightarrow \psi'$ and $\psi_1\oplus\psi_2\Rightarrow \psi$. Therefore, $m_1\vDash \varphi_1\land\psi_1$ and $m_2\vDash\varphi_2\land\psi_2$. Now, suppose $m'\vDash (\varphi_1\land\psi_1) \oplus  (\varphi_2\land\psi_2)$. Then $m'_1\vDash\varphi_1$, $m'_1\vDash\psi_1$, $m'_2\vDash\varphi_2$, and $m'_2\vDash\psi_2$ such that $m'_1\monoid m'_2 = m'$. So, $m'\vDash \varphi_1\oplus\varphi_2$ and $m'\vDash \psi_1\oplus\psi_2$ and by the implications from the induction hypotheses, $m'\vDash \psi'$ and $m'\vDash\psi$, so $m'\vDash\psi'\land\psi$.
\item$\varphi = \psi'\oplus \psi $. We know that $m_1\monoid m_2\vDash \psi'\oplus\psi$, so $n_1\vDash\psi'$ and $n_2\vDash \psi$ such that $n_1\monoid n_2 = m_1\monoid m_2$. By property (1b) of \Cref{def:fls}, there must be $s_1$, $s_2$, $t_1$ and $t_2$ such that $s_1\monoid s_2 = n_1$, $t_1\monoid t_2 = n_2$, $s_1\monoid t_1 = m_1$ and $s_2\monoid t_2 = m_2$. So, $s_1\monoid s_2\vDash \psi'$ and $t_1\monoid t_2\vDash\psi$, and by the induction hypothesis, $s_1\vDash\varphi_1$, $s_2\vdash\varphi_2$, $t_1\vDash\psi_1$ and $t_2\vDash\psi_2$ such that $\varphi_1\oplus\varphi_2\Rightarrow\psi'$ and $\psi_1\oplus\psi_2\Rightarrow \psi$. Recombining terms, we get that $m_1\vDash \varphi_1\oplus\psi_1$ and $m_2\vDash\varphi_2\oplus\psi_2$ and it is easy to see that $\varphi_1\oplus\varphi_2\oplus\psi_1\oplus\psi_2 \Rightarrow \psi'\oplus\psi$.
\item $\varphi = \psi'\Rightarrow\psi$. Vacuous since we assumed $\varphi$ has no implications.
\item $\varphi=P$. By property (2) from \Cref{def:fls}.

\end{itemize}
\end{proof}

\begin{lemma}[Trace Extrapolation]\label{lem:trace}
For any falsifiable OL instance, if there exists $m$ such that $\dem Cm\vDash\psi$ (where $\psi$ contains no implications), then there exists $\varphi$ (also with no implications) such that $m\vDash\varphi$ and $\vDash\triple{\varphi}C\psi$.
\end{lemma}
\begin{proof}
By induction on the structure of the program $C$ (see \Cref{fig:comlang}).
\begin{itemize}
\item $C = \zero$. Assume that $\dem{\zero}{m} \vDash \psi$. Since $\dem{\zero}{m} = \varnothing$, then this assumption gives us $\varnothing\vDash\psi$. We can then take $\varphi= \top$ and derive $\vDash\triple\varphi\zero\psi$: for any $m'\vDash \top$ we have $\dem{\zero}{m'} \vDash \psi$ since we know $\varnothing \vDash \psi$ and $\dem{\zero}{m'} =\varnothing$, for any $m'$. We also clearly have $m\vDash\top$.

\item $C = \bb 1$. Assume $\dem{\bb 1}{m} \vDash \psi$. Since $\dem{\bb 1}{m} = m$, then this assumption gives us $m \vDash\psi$. We can take $\varphi= \psi$ and immediately derive $m\vDash\varphi$. We can then also derive $\vDash\triple{\varphi}{\bb{1}}\psi$: for any $m'\vDash \varphi$ we have $\dem{\bb 1}{m'} = m' \vDash \psi$ since we know $\varphi = \psi$.

\item $C=C_1+C_2$. Assume $\dem{C_1+C_2}{m} \vDash \psi$. We know that $\dem{C_1+C_2}m = \dem{C_1}m\monoid \dem{C_2}m$, so $\dem{C_1}m\monoid \dem{C_2}m\vDash\psi$. By \Cref{lem:split}, we know that there exist $\psi_1,\psi_2$ such that  $\dem{C_1}m\vDash\psi_1$ and $\dem{C_2}m\vDash\psi_2$ and $\psi_1\oplus\psi_2\Rightarrow\psi$. By induction, there exist $\varphi_i$ such that $\vDash\triple{\varphi_i}{C_i}{\psi_i}$ for $i\in\{1,2\}$ and $m\vDash\varphi_i$. Now, we pick the precondition $\varphi = \varphi_1\land \varphi_2$ (so $m\vDash \varphi$). It remains to argue that $\vDash \triple{\varphi}{C_1+C_2}{\psi}$: for any $m'\vDash\varphi_1\land\varphi_2$, we know that $m'\vDash \varphi_i$ and using the fact that $\vDash \triple{\varphi_i}{C_i}{\psi_i}$, we conclude that $\dem{C_i}{m'}\vDash\psi_i$ (for $i=1,2$). Hence, $\dem{C_1+C_2}{m'}\vDash \psi_1\oplus\psi_2$ and since $\psi_1\oplus\psi_2\Rightarrow\psi$ we can conclude $\dem{C_1+C_2}{m'}\vDash \psi$.

\item $C=C_1\fatsemi C_2$. Assume $\dem{C_1\fatsemi C_2}m \vDash\psi$. We know that $\dem{C_1\fatsemi C_2}m = \dem{C_2}{\dem{C_1}m}$, so $\dem{C_2}{\dem{C_1}m}\vDash\psi$. By the induction hypothesis, we conclude that there exists $\vartheta$ such that $\vDash\triple\vartheta{C_2}\psi$ and $\dem{C_1}m\vDash\vartheta$. By induction again, we get that $\vDash\triple{\varphi}{C_1}\vartheta$ such that $m\vDash\varphi$. Now, to show that $\vDash\triple{\varphi}{C_1\fatsemi C_2}\psi$, suppose that $m'\vDash\varphi$, then we know that $\dem{C_1}{m'}\vDash\vartheta$ (from $\vDash\triple{\varphi}{C_1}\vartheta$), and we know that $\dem{C_2}{\dem{C_1}{m'}}\vDash\psi$ (from $\vDash\triple\vartheta{C_2}\psi$), so therefore $\vDash\triple{\varphi}{C_1\fatsemi C_2}\psi$.

\item $C=C^\star$. We will first show that for any $n$, there is a $\varphi$ such that $\vDash\triple{\varphi}{C^n}{\psi}$ and $m\vDash\varphi$. The proof is by induction on $n$. The case where $n=0$ follows from the $\bb{1}$ case above. Now, by the induction hypothesis we know that there is some $\vartheta$ such that $\vDash\triple\vartheta{C^n}\psi$ and $\dem Cm\vDash\vartheta$. By the previous induction hypothesis, we know that $\vDash\triple{\varphi}C\vartheta$, such that $m\vDash\varphi$. So, combining these results, we get that $\vDash\triple{\varphi}{C^{n+1}}\psi$.

Now, since we assumed that the program terminates after finitely many steps, there must be some $n$ such that $\dem{C^\star}m = \dem{\bb{1}}m \monoid \dem Cm \monoid \cdots\monoid \dem {C^n}m$. By repeatedly using \Cref{lem:split}, we can split $\psi$ into $\psi_0,\ldots,\psi_n$ such that $\dem{C^k}m\vDash\psi_k$ (for each $k\in\{0,\ldots,n\}$) and $\psi_0\oplus \cdots \oplus \psi_n\Rightarrow\psi$. By the inductive proof above, for each $k$, there is a $\varphi_k$ such that $\vDash\triple{\varphi_k}{C^k}{\psi_k}$  and $m\vDash\varphi_k$. Now, let $\varphi = \varphi_0 \land\cdots\land\varphi_n$, so clearly $m\vDash\varphi$. We can conclude that $\vDash\triple{\varphi}{C^\star}{\psi_0\oplus\cdots\oplus\psi_n}$ by an argument analogous to the $C=C_1+C_2$ case. Finally, since $\psi_0\oplus \cdots \oplus \psi_n\Rightarrow\psi$, we can weaken the postcondition to obtain $\vDash\triple{\varphi}{C^\star}{\psi}$.

\item $C=c$. By property (4) of \Cref{def:fls}.
\end{itemize}
\end{proof}

\begin{theorem}[Falsification]\label{thm:flssyn}
For any falsifiable OL instance,
\[
\not\vDash\triple{\varphi}C{\bigoplus_{i=1}^nQ_i}
\qquad\text{iff}\qquad
\exists \varphi'\Rightarrow\varphi\quad\text{and}\quad \exists\psi\Rightarrow \lnot\bigoplus_{i=1}^nQ_i
\quad\text{such that}\quad
\vDash\triple{\varphi'}C{\psi}
\]
Where $\psi$ has no implications and $\mathsf{sat}(\varphi')$.
\end{theorem}
\iffpf{
Since $\not\vDash\triple{\varphi}C{\bigoplus_{i=1}^n Q_i}$, then there is an $m$ such that $m\vDash\varphi$ and $\dem Cm\not\vDash\bigoplus_{i=1}^n Q_i$.
By property (3) of \Cref{def:fls}, we know that there exists a $\psi$ with no implications such that $\psi\Rightarrow\lnot\bigoplus_{i=1}^n Q_i$ and $\dem Cm\vDash\psi$. We can now use \Cref{lem:trace} to conclude that there is a $\vartheta$ such that $\vDash\triple\vartheta C\psi$ and $m\vDash\vartheta$. Now, let $\varphi' = \vartheta\land\varphi$ ($\varphi'$ is satisfiable since $m\vDash\vartheta$ and $m\vDash\varphi$). Clearly also $\varphi'\Rightarrow \varphi$. It just remains to show that $\vDash\triple{\varphi'}C\psi$: for any $m'\vDash\varphi'$, then $m'\vDash\vartheta$ and so $\dem{C}{m'}\vDash\psi$ (since $\vDash\triple\vartheta C\psi$).
}{
Assume that $\varphi'\Rightarrow \varphi$ and $\psi\Rightarrow\lnot\bigoplus_{i-1}^n Q_i$ and $\vDash\triple{\varphi'}C{\psi}$. By weakening, we can also conclude that $\vDash\triple{\varphi'}C{\lnot\bigoplus_{i-1}^n Q_i}$. By \Cref{thm:denial}, we can conclude that $\not\vDash\triple{\varphi}C{\bigoplus_{i-1}^n Q_i}$.
}

\begin{remark}
The restrictions laid out in \Cref{def:fls} are quite specific, but we will see in the following sections that they are naturally satisfied in both the nondeterministic and probabilistic models. While the Trace Extrapolation (\Cref{lem:trace}) property may seem unconventional, it can be thought of as a more specialized version of a weakest precondition transformer~\cite{Dijkstra76}. Indeed, if we had such a predicate transformer, we would know that $\dem Cm\vDash \psi$ implies that $m\vDash\wpre(C, \psi)$ and that $\vDash\triple{\wpre(C,\psi)}C{\psi}$ is a valid triple. Unfortunately, weakest preconditions have complex interactions with choice mechanisms such as $C_1+C_2$ and $x\samp\eta$ (refer to \citet{kaminski} for a more in-depth discussion of $\wpre$ and choice). The question of whether a weakest precondition exists for OL remains open. We plan to explore this more in future work, but for now we use the more specialized Trace Extrapolation property to complete the falsification proof.
\end{remark}

Before moving on to the evaluation-model-specific falsification results, we prove a useful lemma about trace extrapolation for pure commands.

\begin{lemma}[Trace Extrapolation for Pure Commands]\label{lem:tepure}
For any OL instance where trace extrapolation holds for pure commands $c$ and basic assertions $P,Q$, that is:
\[
\dem cm\vDash Q\quad\Rightarrow\quad
\exists P.\quad m\vDash P
\quad\text{and}\quad
\vDash\triple PcQ
\]
Then trace extrapolation holds for pure commands and any assertions that do not have implications:
\[
\dem cm\vDash \psi\quad\Rightarrow\quad
\exists \varphi ~\text{with no implications}.\quad m\vDash \varphi
\quad\text{and}\quad
\vDash\triple{\varphi}c{\psi}
\]
\end{lemma}
\begin{proof} By induction on the structure of $\psi$:
\begin{itemize}
\item $\psi = \top$. Let $\varphi = \top$. Clearly $m\vDash \top$ and $\vDash\triple{\top}c{\top}$: suppose $m'\vDash\top$, then clearly $\dem c{m'}\vDash\top$.
\item $\psi = \bot$. Vacuous since $\dem cm\vDash\bot$ is impossible.
\item $\psi = \psi_1\land\psi_2$. Assume $\dem cm\vDash \psi_1\land \psi_2$. By induction, we know that there is a $\varphi_i$ such that $m\vDash\varphi_i$ and $\vDash\triple{\varphi_i}c{\psi_i}$ for $i\in\{1,2\}$. Now, let $\varphi = \varphi_1\land\varphi_2$, so clearly $m\vDash\varphi$. We now show that $\vDash\triple{\varphi}c{\psi_1\land\psi_2}$: suppose $m'\vDash\varphi$, then $m'\vDash\varphi_i$ for $i\in\{1,2\}$. Since $\vDash\triple{\varphi_i}c{\psi_i}$, then $\dem c{m'}\vDash\psi_i$, therefore $\dem c{m'}\vDash \psi_1\land\psi_2$.
\item $\psi=\psi_1\oplus\psi_2$. Assume $\dem cm\vDash\psi_1\oplus\psi_2$. Since $c$ is pure, it cannot split into multiple outcomes, therefore the fact that $\dem cm$ has multiple outcomes means that $m$ also must have multiple outcomes, so there must be $m_1$, $m_2$ such that $m=m_1\monoid m_2$ and $\dem c{m_1}\vDash\psi_1$ and $\dem c{m_2}\vDash\psi_2$. By induction, we know that there is a $\varphi_i$ such that $m_i\vDash\varphi_i$ and $\vDash\triple{\varphi_i}c{\psi_i}$ for $i\in\{1,2\}$. Now, let $\varphi=\varphi_1\oplus\varphi_2$, so clearly $m\vDash\varphi$. Now we show that $\vDash\triple\varphi c{\psi_1\oplus\psi_2}$: suppose $m'\vDash\varphi$, so $m'_1\vDash\varphi_1$ and $m'_2\vDash\varphi_2$ such that $m'_1\monoid m'_2 = m'$. For each $i\in\{1,2\}$, we know that $\vDash\triple{\varphi_i}c{\psi_i}$ for $i\in\{1,2\}$ and so $\dem c{m'_i}\vDash\psi_i$. Combining these, we get $\dem c{m'}\vDash\psi_1\oplus\psi_2$.
\item $\psi = \psi_1\Rightarrow\psi_2$. Vacuous since we assumed that $\psi$ has no implications.
\item $\psi = Q$. By assumption.
\end{itemize}

\end{proof}

\subsection{Nondeterministic Falsification}
\label{app:ndfalse}

This section contains proofs for the falsification results in \Cref{sec:ndfalse}. The goal is to show that nondeterministic instances of OL are falsifiable by showing that \Cref{def:fls} holds for instances of OL using the nondeterministic evaluation model and outcome logic. We first prove that assertions can be falsified, then we prove trace extrapolation, and then we prove \Cref{def:fls}.

\ndfalse*
\iffpf{
Suppose that $S\not\vDash Q_1\oplus\cdots\oplus Q_n$, so for all $S_1, \ldots, S_n$, if $S = \bigcup_{i=1}^n S_i$, there exists some $i$ such that $S_i\not\vDash Q_i$.
Now, for each $i$, let $T_i = \{ \sigma \mid \sigma\in S, \sigma\vDash_\Sigma Q_i\}$, so by construction $T_i\subseteq S$ and therefore $\bigcup_{i=1}^n T_i \subseteq S$. If $S\neq \bigcup_{i=1}^n T_i$, then $S \supset \bigcup_{i=1}^n T_i$, so there must be some $\tau\in S$ such that for all $i$, $\tau\notin T_i$ and so for all $i$, $\tau\not\vDash_\Sigma {Q}_i$, or in other words, $\tau\vDash_\Sigma\overline{Q}_i$. So, $S\vDash (\overline{Q}_1 \land\cdots\land\overline{Q}_n)\oplus\top$. Otherwise, it must be the case that $S= \bigcup_{i=1}^n T_i$ and we therefore know that there exists some $i$ such that $T_i\not\vDash Q_i$. By construction, $\sigma\vDash_\Sigma Q_i$ for every $\sigma\in T_i$, so it must be that $T_i = \emptyset$. This means that $\sigma\not\vDash_\Sigma Q_i$ for every $\sigma\in S$, so $S\vDash \overline{Q}_i$. Or, if $S=\emptyset$, then $S\vDash\top^\oplus$.
}{
There are three cases:
\begin{itemize}
\item Suppose there is some $i$ such that $S\vDash\overline{Q}_i$. This means that $S\neq\emptyset$ and $\forall\sigma\in S.~\sigma\not\vDash_\Sigma Q_i$. Since there are no states satisfying $Q_i$, then there is no $S_i\subseteq S$ such that $S_i \vDash Q_i$ and therefore $S\not\vDash Q_1\oplus\cdots\oplus Q_n$.

\item Suppose that $S\vDash(\overline{Q}_1\land\cdots\land\overline{Q}_n)\oplus\top$, so by \Cref{lem:uxsat}, there is some state $\sigma\in S$ such that $\sigma$ does not satisfy any $Q_i$. Therefore $S\not\vDash Q_1\oplus\cdots\oplus Q_n$: there is now way to break $S$ into $n$ parts each of which satisfying a $Q_i$ because none of those sets can contain $\sigma$.

\item Suppose $S\vDash\top^\oplus$ and so $S=\emptyset$. Then clearly $S\not\vDash Q_1\oplus\cdots\oplus Q_n$: we cannot witness each $Q_i$ because there are no states at all.
\end{itemize}
}

\begin{lemma}[Nondeterministic Trace Extrapolation]\label{lem:ndte}
If $\dem{c}S\vDash\psi$ and $\psi$ has no implications, then there is some $\varphi$ such that $S\vDash\varphi$ and $\vDash\triple\varphi{c}\psi$.
\end{lemma}
\begin{proof} By cases on the structure of $c$.
\begin{itemize}

\item $c=(\assume{e})$. We know that $\dem{\assume e}S\vDash\psi$. Let $S_1 = \dem{\assume e}S = \{\sigma\mid \sigma\in S, \sigma\vDash e\}$ and $S_2 = \dem{\assume{\lnot e}}S = \{\sigma\mid\sigma\in S, \sigma\not\vDash e\}$. Clearly $S_1\cup S_2 = S$, since the two assume statements partition $S$ into two parts. We now define $\varphi$ as follows:
\[
\varphi = \varphi_1\oplus\varphi_2
\qquad\text{where}\qquad
\varphi_1 = \left\{\begin{array}{ll}
\psi \land e & \text{if}~S_1\neq\emptyset \\
\psi \land \top^\oplus &\text{if}~S_1=\emptyset
\end{array}\right.
\qquad
\varphi_2 = \left\{\begin{array}{ll}
\overline e & \text{if}~S_2\neq\emptyset \\
\top^\oplus &\text{if}~S_2=\emptyset
\end{array}\right.
\]
We now show that $S_1\vDash \varphi_1$: we already know that $S_1\vDash\psi$ by assumption. If $S_1\neq\emptyset$, then it must satisfy $e$, since by construction all the states in $S_1$ satisfy $e$. If not, then $S_1=\emptyset\vDash\top^\oplus$. A similar argument shows that $S_2\vDash\varphi_2$. Given this, $S\vDash\varphi$.

It remains to show that $\vDash\triple\varphi{\assume e}\psi$. Suppose $T\vDash\varphi$, so $T_1\vDash\varphi_1$ and $T_2\vDash \varphi_2$ such that $T_1\cup T_2 = T$. Since all the states satisfying $e$ from $T$ are in $T_1$, then $\dem{\assume e}T = T_1$. Since $T_1\vDash\varphi_1$, then $T_1\vDash\psi$.

\item $c$ is a pure command. It suffices to show that the property holds for basic assertions $Q$, we can then use \Cref{lem:tepure} to complete the proof.

Suppose that $\dem cS\vDash Q$. This means that $\dem cS \neq\emptyset$ and $\forall\sigma \dem cS.~\sigma\vDash Q$. For pure commands, there are well known weakest precondition predicate transformations that satisfy $\tau\vDash Q$ iff $\sigma\vDash\wpre(c, Q)$ such that $\{\tau\} = \de c(\sigma)$. This includes the rules for variable assignment ($\wpre(x:=v, Q) = Q[v/a]$) as well as the backwards reasoning rules for Separation Logic given by \citet{sl}. So, it must be the case that $S\vDash \wpre(c, Q)$: since $c$ is pure, it cannot change the magnitude of the set, so $S\neq\emptyset$. In addition, since all the states in the output set satisfy $Q$, then the states in $S$ must all satisfy $\wpre(c,Q)$. Finally, we conclude that $\vDash\triple{\wpre(c,Q)}c{Q}$: suppose that $T\vDash \wpre(c,Q)$, then $T\neq\emptyset$ and $\forall \sigma \in T.~\sigma\vDash \wpre(c,Q)$. By the properties of weakest preconditions, we know that $\tau\vDash Q$ if $\{\tau\} = \de c(\sigma)$, so everything in $\dem cT$ must satisfy $Q$. Additionally, $c$ is pure and cannot change the size of $T$, so $\dem cT\neq\emptyset$. This means that $\dem cT\vDash Q$.

\end{itemize}
\end{proof}

\begin{lemma}
\label{lem:ndfalsifiable}
The nondeterministic instance of OL is falsifiable
\end{lemma}
\begin{proof}\;

\begin{enumerate}
\item Properties of the PCM $\langle \bb{2}^\Sigma, \cup, \emptyset\rangle$:
\begin{enumerate}
\item If $S \cup T = \emptyset$, then it must be the case that $S=T=\emptyset$.
\item Suppose that $S_1\cup S_2 = T_1\cup T_2$. Now, let $U_1 = S_1\cap T_1$, $U_2 = S_2 \cap T_1$, $V_1 = S_1\cap T_2$, and $V_2 = S_2\cap T_2$. It is easy to see that $U_1\cup U_2 = T_1$ and $V_1\cup V_2 = T_2$ and $U_1\cup V_1 = S_1$ and $U_2\cup V_2 = S_2$.
\end{enumerate}
\item Basic assertion splitting: If $S_1\cup S_2\vDash P$, then there are three options. If $S_1 = \emptyset$, then $S_1\vDash\top^\oplus$ and $S_2\vDash P$ and clearly $\top^\oplus\oplus P \Rightarrow P$. The case where $S_2=\emptyset$ is symmetrical. Finally, if both $S_1$ and $S_2$ are nonempty, then $S_1\vDash P$ and $S_2\vDash P$ and $P\oplus P\Rightarrow P$.
\item Assertion falsification: Follows from \Cref{lem:ndfalse}.
\item Trace extrapolation: By \Cref{lem:ndte}.
\end{enumerate}

\end{proof}

The following theorem is a more specific version of \Cref{thm:flssyn} where we include a more specific postcondition (following from \Cref{lem:ndfalse}) instead of existentially quantifying the postcondition.

\ndfalsethm*
\begin{proof}
Follows directly from \Cref{lem:ndfalsifiable,lem:ndfalse,thm:flssyn}.
\end{proof}

Now, we show that OL can disprove any Hoare Triple, which means that it fully subsumes the use case of Incorrectness Logic.

\mhlil*
\iffpf{
Assume $\not\vDash\hoare PCQ$. From \Cref{thm:hl}, we get that $\not\vDash\triple PC{Q\vee\top^\oplus}$. This means that there is some $S$ such that $S\vDash P$ and $\dem CS\not\vDash Q\vee\top^\oplus$, which implies that $\dem CS\not\vDash Q$ and $\dem CS \not\vDash\top^\oplus$, which implies that $\dem CS\vDash \overline Q\oplus\top$. Now, we can use \Cref{lem:trace} to conclude that there is an assertion $\vartheta$ such that $S\vDash\vartheta$ and $\vDash\triple\vartheta{C}{\overline Q\oplus\top}$. Now let $\varphi = \vartheta\land P$, so clearly $S\vDash\varphi$ and since $\varphi\Rightarrow\vartheta$, then $\vDash\triple\varphi C{\overline Q\oplus\top}$, or equivalently, $\vDashD\triple\varphi C{\overline Q}$.
}{
Since $\mathsf{sat}(\varphi)$, there is some $S\vDash\varphi$ and since $\varphi\Rightarrow P$, then also $S\vDash P$. From $\vDashD\triple\varphi C{\overline Q}$, we know that there is a $\tau\in \dem CS$ such that $\tau\not\vDash_\Sigma Q$. There must also be some $\sigma\in S$ such that $\tau\in\de C(\sigma)$, and since $\sigma\in S$ and $S\vDash P$, then $\sigma\vDash_\Sigma P$. So, we have now shown that $\sigma\vDash_\Sigma P$ and $\tau\in\de C(\sigma)$, and $\tau\not\vDash_\Sigma Q$, therefore $\not\vDash\hoare PCQ$.

%
}

\subsection{Probabilistic Falsification}\label{app:probfalse}

In this section, we prove the claims from \Cref{sec:probfalse} pertaining to the falsifiability of probabilistic OL triples. The goal is to prove that probabilistic instances of OL uphold \Cref{def:fls}. We begin by showing that sequences of assertions can be falsified (\Cref{lem:probastfls}), then we we show Trace Extrapolation (\Cref{lem:probte}), and finally we show that probabilistic OL is falsifiable (\Cref{lem:probfls}) implying that \Cref{thm:flssyn} holds.

\begin{lemma}[Falsifying Probabilistic Assertions]\label{lem:probastfls}
\[
\mu\not\vDash \bigoplus_{i=1}^n(\PP[A_i] = p_i)
\qquad\text{iff}\qquad
\exists\psi~\text{(with no implications)}.\quad
\mu\vDash\psi
\quad\text{and}\quad
\psi\Rightarrow \lnot\bigoplus_{i=1}^n(\PP[A_i] = p_i)
\]
\end{lemma}
\iffpf{
We begin by defining $\psi$ as follows:
\[
\psi \triangleq \smashoperator{\bigoplus_{x\in \{0,1\}^n}} (\PP[B_x] = \PP_\mu [ B_x ])
\qquad\text{where}\qquad
B_x \triangleq \bigwedge_{i = 1}^n \mathsf{xor}(A_i, x_i)
\]
In the above, $x_i$ denotes the $i^\text{th}$ bit of the string $x$, so if $x_i = 0$, then the $i^\text{th}$ conjunct of $B_x$ is $A_i$ and if $x_i=1$, then it is $\lnot A_i$. Now, for any $\sigma\in\mathsf{supp}(\mu)$, there must be exactly one $B_x$ such that $\sigma\vDash_\Sigma B_x$. This is because for each $A_i$, either $\sigma\vDash A_i$ or $\sigma\vDash \lnot A_i$, and so a unique $B_x$ corresponds to these choices. That means that $\mu$ can be partitioned by its support into sub-distributions $\mu_x$ such that $\forall\sigma\in\mathsf{supp}(\mu_x).~\sigma\vDash_\Sigma B_x$ and $\mu = \sum_{x\in\{0,1\}^n} \mu_x$. Additionally, $|\mu_x| = \PP_\mu[B_x]$ since $\mathsf{supp}(\mu_x)$ contains exactly those states that satisfy $B_x$. Therefore, for each $x$, $\mu_x\vDash (\PP[B_x] = \PP_\mu[B_x])$ and so $\mu\vDash\psi$.

Now we must show that $\psi\Rightarrow \lnot\bigoplus_{i=1}^n(\PP[A_i] = p_i)$. Suppose that $\eta\vDash\psi$. This means that $\forall x\in \{0,1\}^n$ there is an $\eta_x$ such that $\eta_x\vDash(\PP[B_x] = \PP_\mu[B_x])$ and $\eta=\sum_{x\in\{0,1\}^n}\eta_x$. This also implies that for all $x$, $|\eta_x| = |\mu_x|$. 

For the sake of contradiction, suppose $\eta\vDash \bigoplus_{i=1}^n(\PP[A_i] = p_i)$. In order for this to be true, then for each $i$, we would need the following:
\[
\smashoperator{\sum_{x\in\{0,1\}^n, x_i = 0}} \alpha_{x,i}\cdot\eta_x \vDash \PP[A_i] = p_i
\]
Where each $\alpha_{x,i}$ is a coefficient between $0$ and $1$ such that for all $x$, $\sum_{i=1}^n \alpha_{x,i}= 1$. Essentially, this distributes the probability mass of each $\eta_x$ among all the $A_i$ assertions that it is compatible with. Now, since for each $x$, $|\eta_x| = |\mu_x|$ and if $x_i=0$, then every $\sigma\in\mathsf{supp}(\mu_x)$ must satisfy $A_i$, we also have the following:
\[
\smashoperator{\sum_{x\in\{0,1\}^n, x_i = 0}} \alpha_{x,i}\cdot\mu_x \vDash \PP[A_i] = p_i
\]
And this implies that $\mu\vDash \bigoplus_{i=1}^n(\PP[A_i] = p_i)$, which is a contradiction, therefore it must be the case that $\eta\vDash\lnot \bigoplus_{i=1}^n(\PP[A_i] = p_i)$.
}{
Suppose that there is some $\psi$ such that $\mu\vDash\psi$ and $\psi \Rightarrow \lnot \bigoplus_{i=1}^n(\PP[A_i] = p_i)$. By modus ponens, $\mu\vDash \lnot \bigoplus_{i=1}^n(\PP[A_i] = p_i)$ and therefore $\mu\not\vDash \bigoplus_{i=1}^n(\PP[A_i] = p_i)$.
}

The following lemma is needed for trace extrapolation.

\begin{lemma}[Assertion Scaling]\label{lem:asstscale} For any scalar $\alpha\neq 0$ and assertion $\varphi$, there exists a $\psi$ such that for any $\mu$ if $\alpha\cdot |\mu| \le 1$, then $\mu\vDash\varphi$ iff $\alpha\cdot\mu\vDash\psi$
\end{lemma}
\begin{proof} By induction on the structure of $\varphi$.
\begin{itemize}
\item $\varphi = \top$. Let $\psi =\top$. Clearly $\mu\vDash\top$ iff $\alpha\cdot\mu\vDash\top$ since both are always true.
\item $\varphi = \bot$. Let $\psi =\bot$. Clearly $\mu\vDash\bot$ iff $\alpha\cdot\mu\vDash\bot$ since both are always false.
\item $\varphi = \top^\oplus$. Let $\psi =\top^\oplus$. Clearly $\mu\vDash\top^\oplus$ iff $\alpha\cdot\mu\vDash\top^\oplus$ since both are true iff $\mu = \varnothing$.
\item $\varphi = \varphi_1\land \varphi_2$. By the induction hypothesis, there exist $\psi_1$ and $\psi_2$ such that for any $\mu$, $\mu\vDash\varphi_i$ iff $\alpha\cdot\mu\vDash\psi_i$ for $i\in\{1,2\}$. Now, let $\psi=\psi_1\land\psi_2$, so clearly $\mu\vDash\varphi_1\land\varphi_2$ iff $\alpha\cdot\mu\vDash\psi_1\land\psi_2$.
\item $\varphi = \varphi_1\oplus \varphi_2$. By the induction hypothesis, there exist $\psi_1$ and $\psi_2$ such that for any $\mu$, $\mu\vDash\varphi_i$ iff $\alpha\cdot\mu\vDash\psi_i$ for $i\in\{1,2\}$. Now, let $\psi=\psi_1\oplus\psi_2$. It must be that $\mu\vDash\varphi_1\oplus\varphi_2$ iff $\mu_1\vDash\varphi_1$ and $\mu_2\vDash\varphi_2$ such that $\mu=\mu_1+\mu_2$ iff $\alpha\cdot\mu_1\vDash \psi_1$ and $\alpha\cdot\mu_2\vDash \psi_2$ such that $\alpha\cdot\mu = \alpha\cdot\mu_1 +\alpha\cdot\mu_2$ iff $\alpha\cdot\mu\vDash \psi_1\oplus\psi_2$.
\item $\varphi = \varphi_1\Rightarrow \varphi_2$. By the induction hypothesis, there exist $\psi_1$ and $\psi_2$ such that for any $\mu$, $\mu\vDash\varphi_i$ iff $\alpha\cdot\mu\vDash\psi_i$ for $i\in\{1,2\}$. Now, let $\psi=\psi_1\Rightarrow\psi_2$, so clearly $\mu\vDash\varphi_1\Rightarrow\varphi_2$ iff $\alpha\cdot\mu\vDash\psi_1\Rightarrow\psi_2$.
\item $\varphi = (\PP[A] = p)$. Let $\psi = (\PP[A] = \alpha\cdot p)$. It is easy to see that $\mu\vDash(\PP[A] = p)$ iff $\alpha\cdot\mu\vDash(\PP[A] = \alpha\cdot p)$ since $|\mu| = p$ iff $\alpha\cdot |\mu| = \alpha\cdot p$.
\end{itemize}
\end{proof}

\begin{lemma}[Probabilistic Trace Extrapolation]\label{lem:probte}
If $\dem{c}\mu\vDash\psi$ and $\psi$ has no implications, then there is some $\varphi$ such that $\mu\vDash\varphi$ and $\vDash\triple\varphi{c}\psi$.
\end{lemma}
\begin{proof} By cases on the structure of $c$.
\begin{itemize}
\item $c = (x\samp\eta)$.
First, we know that:
\[
\dem{x\samp\eta}\mu
= \smashoperator{\sum_{\sigma\in\mathsf{supp}(\mu)}} \mu(\sigma)\cdot\smashoperator{\sum_{v\in\mathsf{supp}(\eta)}} \eta(v)\cdot \de{x:=v}(\sigma)
= \smashoperator{\sum_{v\in\mathsf{supp}(\eta)}} \dem{x:=v}{\eta(v)\cdot \mu}
\]
So, we can apply \Cref{lem:split} many times to get a $\psi_v$ for each $v$ such that $\dem{x:=v}{\eta(v)\cdot\mu}\vDash \psi_v$ and $(\bigoplus_{v\in\mathsf{supp}(\eta)}\psi_v)\Rightarrow \psi$. Now, since $x:=v$ is pure, we can use the next case of this proof to conclude that there is a $\varphi_v$ such that $\eta(v)\cdot\mu\vDash\varphi_v$ and $\vDash\triple{\varphi_v}{x:=v}{\psi_v}$. Using \Cref{lem:asstscale} (with $\alpha=\frac1{\eta(v)}$) we can get a $\varphi'_v$ such that $\mu'\vDash\varphi'_v$ iff $\eta(v)\cdot\mu'\vDash\varphi_v$ (and therefore $\mu\vDash\varphi'_v$). Now, let $\varphi = \bigwedge_{v\in\mathsf{supp}(\eta)}\varphi'_v$, so clearly $\mu\vDash\varphi$. We can also show that $\vDash\triple{\varphi}{x\samp\eta}{\psi}$:

Suppose that $\mu'\vDash\varphi$. Then, $\mu'\vDash\varphi'_v$ for each $v$. This also means that $\eta(v)\cdot\mu'\vDash\varphi_v$. Now, using $\vDash\triple{\varphi_v}{x:=v}{\psi_v}$, we know that $\dem{x:=v}{\eta(v)\cdot\mu'}\vDash\psi_v$. Combining these, we get $\dem{x\samp\eta}{\mu'}\vDash \bigoplus_{v\in\mathsf{supp}(\eta)}\psi_v$. This implies that $\dem{x\samp\eta}{\mu'}\vDash\psi$.

\item $c=(\assume{e})$. We know that $\dem{\assume e}\mu\vDash\psi$. Now, let $\varphi = (\psi \land \PP[e] = p)\oplus (\PP[\lnot e] = |\mu|-p)$ where $p = |\dem{\assume e}\mu|$.

Let $\mu_1 = \dem{\assume e}\mu$ and $\mu_2 = \dem{\assume{\lnot e}}\mu$. Clearly $\mu_1+\mu_2 = \mu$, since the two assume statements partition the support of $\mu$ into two parts. It is also the case that $\mu_1\vDash(\psi\land\PP[e]=p)$ since we took $\mu_1\vDash\psi$ as an assumption and $\PP_{\mu_1}[e] = |\dem{\assume e}\mu|$ by construction. Similarly, $\mu_2\vDash(\PP[\lnot e] = |\mu|-p)$ since $\mu_2$ contains all the states where $e$ is false by construction and must have mass equal to $|\mu|-|\mu_1|$. Therefore, $\mu\vDash\varphi$.

We now show that $\vDash\triple\varphi{\assume e}\psi$. Suppose that $\mu'\vDash\varphi$. Therefore, $\mu_1\vDash \psi\land \PP[e]=p$ and $\mu_2\vDash\PP[\lnot e]=(1-p)$ such that $\mu'=\mu_1+\mu_2$. It must be the case that $\dem{\assume e}{\mu'} = \mu_1$, since $\lnot e$ holds for every state in the support of $\mu_2$. We already know that $\mu_1\vDash\psi$, so $\dem{\assume e}{\mu'}\vDash\psi$.

\item $c$ is a pure command. It suffices to show that the property holds for basic assertions $\PP[A]=p$, we can then use \Cref{lem:tepure} to complete the proof.

Suppose that $\dem c\mu\vDash (\PP[A] = p)$. This means that $|\dem c\mu| = p$ and $\forall\sigma\in\mathsf{supp}(\dem c\mu).~\sigma\vDash A$. For pure commands, there are well known weakest precondition predicate transformations that satisfy $\tau\vDash A$ iff $\sigma\vDash\wpre(c, A)$ such that $\unit(\tau) = \de c(\sigma)$. This includes the rules for variable assignment ($\wpre(x:=v, A) = A[v/a]$) as well as the backwards reasoning rules for Separation Logic given by \citet{sl}. So, it must be the case that $\mu\vDash\PP[\wpre(c, A)] = p$: since $c$ is pure, it cannot change the mass of the distribution, so $|\mu| = |\dem c\mu| = p$. In addition, since all the states in the output distribution satisfy $A$, then the states in $\mu$ must all satisfy $\wpre(c,A)$. Finally, we conclude that $\vDash\triple{\PP[\wpre(c,A)] =p}c{\PP[A] = p}$: suppose that $\mu'\vDash\PP[\wpre(c,A)]=p$, then $|\mu'|=p$ and $\forall \sigma\in\mathsf{supp}(\mu').~\sigma\vDash \wpre(c,A)$. By the properties of weakest preconditions, we know that $\tau\vDash A$ if $\unit(\tau) = \de c(\sigma)$, so everything in the support of $\dem c{\mu'}$ must satisfy $A$. Additionally, $c$ is pure and cannot change the mass of the distribution, so $|\dem c{\mu'}| = |\mu'| = p$. This means that $\dem c{\mu'}\vDash \PP[A] =p$.

\end{itemize}
\end{proof}

\begin{lemma}\label{lem:probfls}
The probabilistic instance of OL is falsifiable
\end{lemma}
\begin{proof}\;

\begin{enumerate}
\item Properties of the PCM $\langle \mathcal D\Sigma, +, \varnothing\rangle$:
\begin{enumerate}
\item If $\mu_1+\mu_2 = \varnothing$, then it must be the case that $\mu_1=\mu_2 =\varnothing$ since the monoid operation $+$ can only add probability mass, not remove it.
\item Suppose that $\mu_1 + \mu_2 = \eta_1+\eta_2$. We now define the following:
\begin{align*}
\alpha_1 &\triangleq \lambda x. \min(\mu_1(x), \eta_1(x)) &
\alpha_2 &\triangleq \lambda x. \eta_1(x) - \alpha_1(x) \\
\beta_2 &\triangleq \lambda x. \min(\mu_2(x), \eta_2(x))  &
\beta_1 &\triangleq \lambda x.\eta_2(x) - \beta_2(x)
\end{align*}
By construction $\alpha_1 + \alpha_2 = \eta_1$ and $\beta_1+\beta_2 = \eta_2$. We now show that for any $x$, $\alpha_1(x) + \beta_1(x) = \mu_1(x)$:
\begin{align*}
\alpha_1(x) + \beta_1(x) &= \min(\mu_1(x), \eta_1(x)) + \eta_2(x) - \beta_2(x) \\
&= \min(\mu_1(x), \eta_1(x)) + \eta_2(x) - \min(\mu_2(x), \eta_2(x)) \\
&= \min(\mu_1(x), \eta_1(x)) + \max(\eta_2(x) - \mu_2(x), 0) \\
&= \min(\mu_1(x), \eta_1(x)) + \max(\mu_1(x) - \eta_1(x), 0) \\
\intertext{So, if $\mu_1(x) \le \eta_1(x)$, then this equals $\mu_1(x) + 0$, otherwise it is $\eta_1(x) + \mu_1(x) - \eta_1(x)$.}
&= \mu_1(x)
\end{align*}
It is also true that $\alpha_2(x) + \beta_2(x) = \mu_2(x)$ by a symmetric argument.
\end{enumerate}
\item Basic assertion splitting: We know that $\mu_1 \monoid \mu_2\vDash \PP[A]=p$, so $|\mu_1| + |\mu_2| = p$ and all the states in both supports satisfy $A$. That means that $\mu_1\vDash (\PP[A] = |\mu_1|)$ and $\mu_2\vDash(\PP[A] = |\mu_2|)$ and $(\PP[A] = |\mu_1|)\oplus (\PP[A] = |\mu_2|) \Rightarrow (\PP[A] = p)$.
\item Assertion falsification: Follows from \Cref{lem:probastfls}.
\item Trace extrapolation: By \Cref{lem:probte}.
\end{enumerate}

\end{proof}

While we have already shown that probabilistic OL is falsifiable, the result in \Cref{lem:probfls} gives us a falsifying postcondition that is exponentially large. If the original specification had $n$ outcomes in the postcondition, then the specification that disproves it will have $2^n$ outcomes. We now show that in the common case where the outcomes are disjoint, the incorrectness specification only needs $n+1$ outcomes.

\probfalse*
\begin{proof}
In general, if all the $B_j$s are disjoint, then $\mu\vDash\bigoplus_{j=1}^m (\PP[B_j] = r_j)$ iff for each $j$, $\PP_\mu[B_j] = r_j$ and $|\mu| = \sum_{j=1}^m r_j$. This is easy to see, since the disjointness condition partitions the support of $\mu$. It will now suffice to prove the following claim, the remainder of the proof then follows from \Cref{thm:falsification}.
Claim: $\mu \not\vDash \bigoplus_{i=1}^n(\PP[A_i] = p_i)$ iff $\exists\vec q.~\mu\vDash\bigoplus_{i=0}^n(\PP[A_i] = q_i)$. Due to disjointness, this is equivalent to saying that there is some $i$ such that $\PP_\mu[A_i] \neq p_i$ or $|\mu| \neq \sum_{i=1}^n p_i$ iff there exist $\vec q$ such that for each $i$, $\PP_\mu[A_i] = q_i$ and $|\mu| = \sum_{i=0}^n q_i$ and either $q_0\neq 0$ or there is some $i$ such that $q_i\neq p_i$.
\iffcases{
Let each $q_i = \PP_\mu[A_i]$, with the addition of $A_0$, the $A_i$s form a tautology, so they account for all the states in $\mu$ and therefore $\sum_{i=0}^n q_i = |\mu|$. By assumption, either $\PP_\mu[A_i] \neq p_i$ or $|\mu| \neq \sum_{i=1}^n p_i$. If $\PP_\mu[A_i] \neq p_i$, then clearly $p_i\neq q_i$. If every $\PP_\mu[A_i] =p_i$, then it must be that $|\mu| \neq \sum_{i=1}^n p_i = \sum_{i=1}^n q_i$, and so it must be that $q_0 \neq 0$.
}{
Suppose that every $\PP[A_i] = q_i$ and $\sum_{i=0}^n q_i = |\mu|$ and either $p_i \neq q_i$ for some $i$ or $q_0\neq 0$. If there is an $i$ such that $p_i\neq q_i$, then clearly $\PP[A_i] \neq p_i$. If each $p_i=q_i$, then it must be that $q_0 \neq 0$, and then $\sum_{i=1}^n p_i = (\sum_{i=0}^n q_i) - q_0 = |\mu|-q_0 \neq |\mu|$.
}
%
%

\end{proof}

Going further, some specifications can be disproven using a single lower bound:

\lbfalse*
\begin{proof}
We first show that $(\PP[\lnot A] \ge q) \Rightarrow \lnot (\PP[A] \ge p)$. Suppose that $\mu\vDash \PP[\lnot A] \ge q$, so by \Cref{lem:problb}, $\PP_\mu[\lnot A] \ge q$ and therefore $\PP_\mu[\lnot A] > 1-p$. This also means that $\PP_\mu[A] < |\mu| - (1-p)$ and since $|\mu| \le 1$, then $\PP_\mu[A] < 1 - (1 - p) = p$. It follows that $\PP_\mu[A] \not\ge p$.

Now, given the implication that we just proved, we can conclude that $\vDash\triple{\varphi'}C{\lnot(\PP[A]\ge p)}$. Therefore, the original claim holds by \Cref{thm:denial}.
\end{proof}

\section{Separation Logic}
\label{app:sl}
In this section we define the semantics of the assertion logic and atomic commands defined in \Cref{sec:sep}.

\subsection{Semantics of the Assertion Logic}
\label{app:sl-assert}
Recall the syntax for separation logic.
$$p\in\textsf{SL} ::= \textbf{emp}\mid \exists x.p\mid p\land q\mid p\lor q\mid p\Rightarrow q \mid p \sep q\mid p \wand q \mid e\mid e_1\mapsto e_2\mid e\mapsto-\mid e\not\mapsto$$
First we define the disjoint union of two heaps $\uplus\colon\mathcal{H}\to\mathcal{H}\rightharpoonup\mathcal{H}$ as follows:
\[ h_1 \uplus h_2 \quad\triangleq\quad \lambda \ell.\left\{
\begin{array}{ll}
h_1(\ell) &\text{if}~\ell\in\textsf{dom}(h_1)\\
h_2(\ell) &\text{if}~\ell\in\textsf{dom}(h_2)
\end{array}\right.
\qquad
\text{if}
\qquad
\textsf{dom}(h_1) \cap\textsf{dom}(h_2)=\emptyset
\]
The satisfaction relation ${\vDash}\subseteq(\mathcal{S}\times\mathcal{H})\times \textsf{SL}$ is defined as follows.
\[
\begin{array}{lll}
(s, h)\vDash \textbf{emp} &\text{iff}& \textsf{dom}(h) = \emptyset \\
(s,h)\vDash \exists x.p &\text{iff}& (s,h)\vDash p[v/x]~\text{for some}~v\\
(s,h)\vDash p\land q &\text{iff}& (s,h)\vDash p ~\text{and}~(s,h)\vDash q\\
(s,h)\vDash p\lor q &\text{iff}& (s,h)\vDash p ~\text{or}~(s,h)\vDash q\\
(s,h)\vDash p\Rightarrow q& \text{iff}& \text{if}~(s,h)\vDash p~\text{then}~(s,h)\vDash q\\
(s, h)\vDash p\sep q &\text{iff}& \exists h_1, h_2 ~\text{such that}~h = h_1\uplus h_2~\text{and}~(s,h_1)\vDash p~\text{and}~(s,h_2)\vDash q\\
(s, h)\vDash p\wand q & \text{iff}& \forall h_1,h_2~\text{such that}~ h_2 = h\uplus h_1~\text{if}~(s,h_1)\vDash p~\text{then}~(s,h_2)\vDash q\\
(s,h)\vDash e &\text{iff} & \de{e}(s) = \tru \\
(s,h)\vDash e_1\mapsto e_2 &\text{iff}& \de{e_1}(s) = \ell~\text{and}~\textsf{dom}(h) = \{\ell\}~\text{and}~ h(\ell) = \de{e_2}(s) \\
(s,h)\vDash e\mapsto- &\text{iff}& \de{e}(s) = \ell~\text{and}~\textsf{dom}(h) = \{\ell\}~\text{and}~ h(\ell) \neq\bot \\
(s,h)\vDash e\not\mapsto &\text{iff}&  \de{e}(s)=\textsf{null}  ~\text{and}~\textsf{dom}(h)=\emptyset~\text{or}\\
&&\de{e}(s) = \ell~\text{and}~\textsf{dom}(h) = \{\ell\}~\text{and}~ h(\ell) =\bot
\end{array}
\]
%
Note that this is a \emph{classical} interpretation of separation logic where the points-to predicate $x\mapsto v$ is satisfied only by a singleton heap. We can add the \emph{intuitionistic} points-to predicate $x\hookrightarrow v$ as syntactic sugar for $x\mapsto v\sep \tru$ which is satisfied by \emph{any} stack--heap pair $(s,h)$ where $h(\de{x}(s)) = \de{v}(s)$. The difference between $x\mapsto v$ and $x\hookrightarrow v$ is very similar to the difference between over- and under-approximate versions of outcomes that we saw in \Cref{sec:logic}, where we defined under-approximation $m\vDash^\downarrow P$ to be $m\vDash P\oplus\top$ .

\subsection{Logical Operations on Errors}\label{sec:erlogic}

Let ${\vDash_A}\subseteq A\times\prop_A$ and ${\vDash_B}\subseteq A\times\prop_B$ be two logical satisfaction relations in which the assertion syntaxes ($\prop_A$ and $\prop_B$) contain the usual logical constructs $\tru$, $\fls$, $\land$, $\lor$, and $\lnot$. In addition, let $\prop = \prop_A\times\prop_B$ and ${\vDash}\subseteq (B+A)\times\prop$ is the satisfaction relation from \Cref{def:erassert}. We now add the following logical operations:
\[
\begin{array}{lcllcl}
\tru &\triangleq& (\tru, \tru) \qquad\qquad\;&
(p,q)\land (p',q') & \triangleq & (p\land p', q\land q') \\
\fls & \triangleq & (\fls,\fls) &
(p,q)\lor (p',q') & \triangleq & (p\lor p', q\lor q')\\
\lnot (p,q) &\triangleq & (\lnot p, \lnot q)
\end{array}
\]
To provide some justification for these definitions, we prove the following sanity checks.

\begin{lemma}[Sanity checks for logical operations]
The following statements hold for all $m$, $p$, $p'$, $q$, and $q'$.
\begin{itemize}
\item \textsc{True:} $m \vDash \tru$
\item \textsc{False:} $m\not\vDash\fls$
\item \textsc{Conjunction:} $m\vDash (p, q)\land (p',q')$ iff $m\vDash(p,q)$ and $m\vDash (p',q')$
\item \textsc{Disjunction:} $m\vDash (p, q)\lor (p',q')$ iff $m\vDash(p,q)$ or $m\vDash (p',q')$
\item \textsc{Negation:} $m\vDash \lnot (p,q)$ iff $m\not\vDash (p,q)$
\item \textsc{Sugar Syntax:} $(\ok:p) \vee (\er:q)$ iff $(p,q)$
\end{itemize}
\end{lemma}
\begin{proof} We prove each case assuming that $m=\inj_L(b)$. The cases where $m=\inj_R(a)$ are symmetric.

\begin{itemize}
\item \textsc{True:} $\inj_L(b)\vDash (\tru,\tru)$ since $b\vDash\tru$
\item \textsc{False:} $\inj_L(b)\not\vDash(\fls,\fls)$ since $b\not\vDash\fls$
\item \textsc{Conjunction:} $\inj_L(b)\vDash (p\land p', q\land q')$ iff $b\vDash q\land q'$ iff $b\vDash q$ and $b\vDash q'$ iff $\inj_L(b)\vDash (p,q)$ and $\inj_L(b)\vDash (p',q')$.
\item \textsc{Disjunction:} $\inj_L(b)\vDash (p\lor p', q\lor q')$ iff $b\vDash q\lor q'$ iff $b\vDash q$ or $b\vDash q'$ iff $\inj_L(b)\vDash (p,q)$ or $\inj_L(b)\vDash (p',q')$.
\item \textsc{Negation:} $\inj_L(b)\vDash (\lnot p, \lnot q)$ iff $b\vDash \lnot q$ iff $b\not\vDash q$ iff $\inj_L(b)\not\vDash (p, q)$
\item \textsc{Sugar Syntax:} $(\ok:p)\vee(\er:q) = (p, \fls)\vee (\fls, q) = (p\vee\fls, \fls\vee q)$ iff $(p,q)$.
\end{itemize}
\end{proof}

\subsection{Semantics of Programs}\label{app:sl-prog}
Recall the syntax of the atomic \textsf{mGCL} commands.
\[c\in\textsf{mGCL} ::= \textsf{assume}~e\mid x:=e \mid x := \textsf{alloc}() \mid\textsf{free}(e)\mid x\leftarrow [e] \mid [e_1] \leftarrow e_2 \mid\textsf{error}()\]
The semantics $\de{c}\colon \mathcal{S}\times\mathcal{H}\to M((\mathcal{S}\times\mathcal{H})+(\mathcal{S}\times\mathcal{H}))$ is given below, parameterized by any monad $M$. Note that often the semantics of $\textsf{alloc}()$ is nondeterministic and, in particular, it might reallocate some location $\ell$ such that $h(\ell)=\bot$. We have chosen to make the semantics fully deterministic so as to allow \textsf{mGCL} to be embedded into, for example, a probabilistic evaluation context.
{\small
\begin{align*}
\de{\textsf{assume}~e}(s,h) ~&= \left\{\begin{array}{ll}
\unit_\er((s,h)) & \text{if}~\de{e}(s) \neq 0 \\
\ident & \text{otherwise}
\end{array}\right.\\
{\de{x:=e}}(s, h) ~&= \unit_\er((s[x\mapsto\de{e}(s)], h)) \\
{\de{x := \textsf{alloc}()}}(s, h) ~&= \unit_\er((s[x\mapsto\ell], h[\ell\mapsto\textsf{null}]))~\text{where}~\ell = \max(\textsf{dom}(h))+1 \\
{\de{\textsf{free}(e)}}(s, h) ~&= \left\{\begin{array}{ll}
\unit_\er(s, h[\ell\mapsto\bot]) & \text{if}~\de{e}(\sigma) = \ell~\text{and}~\ell\in\textsf{dom}(h)~\text{and}~h(\ell)\neq \bot\\
\unit_M(\inj_L((s,h))) & \text{if}~\de{e}(\sigma) = \ell~\text{and}~h(\ell)= \bot\\
\ident & \text{if}~\de{e}(\sigma) \notin\mathsf{dom}(h)
\end{array}\right.\\
{\de{[e_1]\leftarrow e_2}}(s, h) ~&= \left\{\begin{array}{ll}
\unit_\er(s, h[\ell\mapsto\de{e_2}(\sigma)]) & \text{if}~\de{e_1}(\sigma) = \ell~\text{and}~\ell\in\textsf{dom}(h)~\text{and}~h(\ell)\neq \bot \\
\unit_M(\inj_L((s,h))) & \text{if}~\de{e_1}(\sigma) = \ell~\text{and}~h(\ell)= \bot \\
\ident & \text{if}~\de{e_1}(\sigma) \notin\mathsf{dom}(h)
\end{array}\right.\\
{\de{x\leftarrow [e]}}(s, h) ~&= \left\{\begin{array}{ll}
\unit_\er(s[x\mapsto h(\ell)], h) & \text{if}~\de{e}(s) = \ell~\text{and}~\ell\in\textsf{dom}(h)~\text{and}~h(\ell)\neq\bot \\
\unit_M(\inj_L((s,h))) & \text{if}~\de{e}(s) = \ell~\text{and}~h(\ell)=\bot \\
\ident & \text{if}~\de{e}(\sigma) \notin\mathsf{dom}(h)
\end{array}\right.\\
{\de{\textsf{error}()}}(s, h) ~&= \unit_M(\inj_L((s,h)))
\end{align*}}
We can define the usual semantics of $\textsf{alloc}()$ if we specialize $M$ to the powerset monad.
\[\small
\de{x := \textsf{alloc}()}(s,h) = \{ \inj_R((s[x\mapsto \ell], h[\ell\mapsto v])) \mid \ell \in \mathbb{N}^+, v\in\textsf{Val}, \ell \notin\textsf{dom}(h)\vee h(\ell)=\bot \}
\]

\subsection{Manifest Errors}

\manifest*
\begin{proof}
First, recall that by definition, $\vDash\inc pC{\er:q}$ is a manifest error iff $\forall\sigma.~\exists\tau\in\de C(\sigma).~ \tau\vDash(\er:q\sep\tru)$.

\iffcases{
Suppose that $S\vDash(\ok:\tru)$. This means that there must be a $\sigma$ such that $\inj_R(\sigma)\in S$. By the definition of manifest errors, we know that $\exists\tau\in\de{C}(\sigma)$ such that $\tau\vDash(\er:q\sep\tru)$. Now, since $\de{C}(\sigma) = \dem C{\{\inj_R(\sigma)\}}$ and $\{\inj_R(\sigma)\} \subseteq S$, then $\tau\in \dem CS$ and so $\dem CS\vDashD (\er:q\sep\tru)$. Therefore, $\vDashD\triple{\ok:\tru}C{\er:q\sep\tru}$.
}{
Let $\sigma$ be any program state. From $\vDashD\triple{\ok:\tru}C{\er:q\sep\tru}$, we know that $\dem C{\{\inj_R(\sigma)\}}\vDashD(\er:q\sep\tru)$. So, by \Cref{lem:uxsat} there must be some $\tau\in\de C(\sigma)$ such that $\tau\vDash (\er:q\sep\tru)$ (since $\dem C{\{\inj_R(\sigma)\}} = \de C(\sigma)$).
}
\end{proof}

\section{Soundness Proofs}
\label{app:soundness}

\begin{lemma}[Soundness of generic rules in Figure~\ref{fig:baserules}]
If $\vdash\triple PCQ$ then $\vDash\triple PCQ$.
\end{lemma}
\begin{proof}
By induction on the derivation $\vdash\triple PCQ$.
\begin{itemize}
\item\textsc{Zero}. Suppose that $m\vDash\varphi$. We know that $\de{\zero}^\dagger(m) = \ident$ and $\ident\vDash\top^\oplus$, therefore $\vDash\triple\varphi\zero{\top^\oplus}$
\item\textsc{One}. Suppose that $m\vDash\varphi$. The know that $\de{\bb{1}}^\dagger(m) = m$ and we assumed that $m\vDash\varphi$, so $\vDash\triple\varphi{\bb{1}}\varphi$
\item\textsc{Seq}. Suppose that $m\vDash\varphi$. By induction, we know that $\de{C_1}^\dagger(m)\vDash\psi$. By induction again, we know that $\de{C_2}^\dagger(\de{C_1}^\dagger(m))\vDash\vartheta$. In addition:
\begin{align*}
\dem{C_2}{\dem{C_1}m}
&= \bind(\dem{C_1}m, \de{C_2}) \\
&= \bind(\bind(m, \de{C_1}), \de{C_2}) \\
&= \bind(m, \lambda\sigma.\bind(\de{C_1}(\sigma), \de{C_2})) \\
&= \bind(m, \de{C_1\fatsemi C_2})\\
&= \dem{C_1\fatsemi C_2}m
\end{align*}
So, $\dem{C_1\fatsemi C_2}m\vDash \vartheta$ and therefore $\vDash\triple\varphi{C_1\fatsemi C_2}\vartheta$

\item\textsc{For}. Since $\mathsf{for}~N~do~C$ is syntactic sugar for $C^n$ (or, equivalently, $C\fatsemi\cdots\fatsemi C$), this rule can be derived by induction on $N$ using the \textsc{Seq} use.

\item\textsc{Split}.
Suppose $m\vDash \varphi_1\oplus\varphi_2$, then there exists $m_1$ and $m_2$ such that $m_1\monoid m_2 = m$ and $m_1\vDash\varphi_1$ and $m_2\vDash\varphi_2$. By induction, we know that $\dem{C}{m_1}\vDash\psi_1$ and $\dem{C}{m_2}\vDash \psi_2$. By linearity, we know that $\dem{C}{m_1}\monoid\dem{C}{m_2} = \dem{C}{m_1\monoid m_2} = \dem Cm$. Note that this does not necessarily mean that $\dem{C}{m}$ is defined, but if we limit $C$ to be syntactically valid (as described in \Cref{app:totality}), then it must be defined and so $\de{C}(m)\vDash\psi_1\oplus\psi_2$


\item\textsc{Consequence}. Suppose that $m\vDash\varphi'$. By the assumption that $\varphi'\Rightarrow\varphi$, this means that $m\vDash\varphi$. By induction, we know that $\dem Cm\vDash\psi$ and so $\dem Cm\vDash\psi'$ (since $\psi\Rightarrow\psi'$) and therefore $\vDash\triple{\varphi'}C{\psi'}$
\item\textsc{Empty}. Suppose that $m\vDash\top^\oplus$, then $m=\varnothing$. We also know that $\bind(\varnothing,f) = \varnothing$ for any $f$, so $\dem{C}\varnothing = \varnothing$, therefore $\dem Cm\vDash\top^\oplus$.
\item\textsc{True}. Suppose that $m\vDash \varphi$. It is trivially true that $\dem Cm\vDash \top$.
\item\textsc{False}. The premise that $m\vDash\bot$ is impossible, therefore this case vacuously holds.
\end{itemize}
\end{proof}

\begin{lemma}[Soundness of nondeterministic rules in Figure~\ref{fig:baserules}]
If $\vdash\triple PCQ$ then $\vDash\triple PCQ$.
\end{lemma}
\begin{proof}By induction on the derivation $\vdash\triple PCQ$.
\begin{itemize}
\item\textsc{Plus}. Suppose that $m\vDash\varphi$. By induction, we know that $\dem{C_1}m\vDash \psi_1$ and $\dem{C_2}m\vDash\psi_2$. By the definition of $\de{-}$ we also know that $\dem{C_1}m\cup\dem{C_2}m=\dem{C_1+C_2}m$ and therefore $\dem{C_1+C_2}m\vDash \psi_1\oplus\psi_2$.
\item\textsc{Induction}. Suppose that $m\vDash\varphi$. We know by induction that $\dem{\bb{1}+C\fatsemi C^\star}m\vDash\psi$. Let $F = \lambda f.\lambda \sigma. f^\dagger(\de{C}(\sigma))\cup \unit(\sigma)$ and note that:
\begin{align*}
\dem{\bb{1}+C\fatsemi C^\star}{m} ~&= \dem{\bb{1}}{m} \cup \dem{C\fatsemi C^\star}{m} \\
&= m \cup\dem{C^\star}{\dem{C}{m}})\\
&= m\cup \bigcup_{n\in\mathbb{N}}F^n(\lambda x.\varnothing)^\dagger(\dem{C}{m})\\
&= F(\lambda x.\varnothing)^\dagger(m) \cup \bigcup_{n\ge 1}F^n(\lambda x.\varnothing)^\dagger(m)\\
&= \dem{C^\star}{m}
\end{align*}
So, $\dem{C^\star}m\vDash\psi$.
\end{itemize}
\end{proof}

\begin{lemma}[Soundness of expression-based rules in \Cref{fig:baserules}]
If $\vdash\triple PCQ$ then $\vDash\triple PCQ$.
\end{lemma}
\begin{proof}
By induction on the derivation $\vdash\triple PCQ$.
\begin{itemize}
\item\textsc{Assume}. Suppose that $m\vDash P_1\oplus P_2$. Since $P_1\vDash e$ and $P_2\vDash\lnot e$, we know by the definition of expression entailment that that $\dem{\assume{e}}{m}\vDash P$.
\item\textsc{Assign}. Suppose that $m\vDash P[e/x]$. By the required properties of substitution, we know that $\dem{x:=e}{m}\vDash P$.
\item\textsc{If}. Suppose that $m\vDash P_1\oplus P_2$. Now observe that:
\begin{align*}
\dem{\iftf e{C_1}{C_2}}{m} ~&= \dem{(\assume{e}\fatsemi C_1) + (\assume{\lnot e}\fatsemi C_2)}{m} \\
&= \dem{\assume{e}\fatsemi C_1}{m} \monoid \dem{\assume{\lnot e}\fatsemi C_2}{m} \\
&= \dem{C_1}{\dem{\assume{e}}{m}} \monoid \dem{C_2}{\dem{\assume{\lnot e}}{m}} \\
\intertext{Now, let $m_1=\dem{\assume{e}}{m}$ and $m_2=\dem{\assume{\lnot e}}{m}$. Since we know that $P_1\vDash e$ and $P_2\vDash \lnot e$ and $m\vDash P_1\oplus P_2$, then $m_1\vDash P_1$ and $m_2\vDash P_2$ (by the required properties of expression entailment).}
&= \dem{C_1}{m_1} \monoid \dem{C_2}{m_2}
\end{align*}
By the induction hypotheses, we also know that $\dem{C_1}{m_1}\vDash Q_1$ and $\dem{C_2}{m_2}\vDash Q_2$, therefore $\dem{C_1}{m_1}\diamond\dem{C_2}{m_2}\vDash Q_1\oplus Q_2$. Note that this composition with $\diamond$ is valid in all the execution models we have presented since $\cup$ is total and we have already shown that $+$ on distributions is defined in the semantics of if statements.
\end{itemize}
\end{proof}

\begin{lemma}[Soundess of Nondeterministic Lifting Rule]
The following inference rule is sound.
\[
\inferrule
{\vdash_{\bb{2}^{(-)}}\triple pCq}
{\triple pCq}
{\textsc{Nondeterministic Lift}}
\]
\end{lemma}
\begin{proof}By induction on the derivation $\vdash\triple pCq$.
Suppose that $S\vDash p$, so that means that $S\neq\emptyset$ and $\forall\sigma\in S.~\sigma\vDash p$. We know by induction that for any $\sigma\vDash p$, there is some $\tau$ such that $\dem{C}{\{\sigma\}} = \{\tau\}$ and $\tau\vDash q$. We also know that $\dem CS = \bigcup_{\sigma\in S}\dem{C}{\{\sigma\}}$ and since each for each $\sigma$, there is a $\tau$ such that $\dem{C}{\{\sigma\}}=\{\tau\}$, then $\forall\tau\in\dem CS$, $\tau\vDash q$ and so $\dem CS\vDash q$
\end{proof}

\begin{lemma}[Soundness of Error Propagation]
The following inference rule is sound:
\[
\inferrule
{\;}
{\vdash_M\triple{\textsf{er}:p}C{\textsf{er}:p}}
{\textsc{Error Propagation}}
\]
\end{lemma}
\begin{proof}
Suppose that $m\vDash (\er:p)$, and so there must be some $\sigma$ such that $m=\inj_L(\sigma)$ and $\inj_L(\sigma)\vDash(\er:p)$. Now, we have:
\begin{align*}
\dem{C}{\unit_M(\inj_L(\sigma))}
&= \bind_M\left(\unit_M(\inj_L(\sigma)), \lambda x.\left\{\begin{array}{ll}
\de{C}(y) & \text{if}~x=\inj_R(y) \\
\unit_M(x) & \text{if}~x=\inj_L(y)
\end{array}\right.\right) \\
&= \left(\lambda x.\left\{\begin{array}{ll}
\de{C}(y) & \text{if}~x=\inj_R(y) \\
\unit_M(x) & \text{if}~x=\inj_L(y)
\end{array}\right.\right)(\inj_L(\sigma)) \\ 
&= \unit_M(\inj_L(\sigma))
\end{align*}
And since we already know that $\inj_L(\sigma)\vDash(\er:p)$, we are done.
\end{proof}

\begin{lemma}[Soundness of Probabilistic Proof System]\label{lem:probrules}
The inference rules at the top of \Cref{fig:probproof} are sound.
\end{lemma}
\begin{proof}By induction on the derivation $\vdash\triple PCQ$
\begin{itemize}
\item\textsc{Lifting}. Suppose that $\mu\vDash (\PP[A] = p)$, so for every $\sigma\in\textsf{supp}(\mu)$, $\sigma\vDash A$ and $|\mu| = p$. We know by induction that for any $\sigma$ there is some $\tau_\sigma$ such that $\de{C}(\sigma) = \delta_{\tau_\sigma}$ and $\tau_\sigma\vDash B$. So, $\mu' = \dem{C}\mu = \sum_{\sigma\in\textsf{supp}(\mu)} \mu(\sigma)\cdot\de{C}(\sigma) = \sum_{\sigma\in\textsf{supp}(\mu)} \mu(\sigma)\cdot \delta_{\tau_\sigma}$. Therefore $|\mu'| = |\mu|$ and $\forall\tau\in\textsf{supp}(\mu')$, $\tau\vDash B$, so $\mu'\vDash (\PP[B] = p)$.

\item\textsc{Sample}. First, observe that:
\begin{align*}
\dem{x\samp\eta}{\mu} ~&= \bind(\mu,\lambda\sigma.\bind(\eta, \lambda v.\de{x:=v}(\sigma))) \\
&= \smashoperator{\sum_{\sigma\in\textsf{supp}(\mu)}}\mu(\sigma)\cdot\smashoperator{\sum_{v\in\textsf{supp}(\eta)}} \eta(v)\cdot \de{x:=v}(\sigma) \\
&= \smashoperator{\sum_{v\in\textsf{supp}(\eta)}} \eta(v)\cdot  \smashoperator{\sum_{\sigma\in\textsf{supp}(\mu)}}\mu(\sigma)\cdot \de{x:=v}(\sigma)\\
&= \smashoperator{\sum_{v\in\textsf{supp}(\eta)}} \eta(v)\cdot \dem{x:=v}{\mu}
\end{align*}
Now, by the same argument that we used in the lifting cases, since $\mu\vDash (\PP[A] = p)$ and $\triple{A}{x:=v}{B_v}$, then $\dem{x:=v}{\mu}\vDash (\PP[B_v]=p)$.
Therefore, we we can also weight the distribution to obtain $\eta(v)\cdot\dem{x:=v}{\mu}\vDash (\PP[B_v]=\eta(v)\cdot p)$. Now, the sum over $v\in\textsf{supp}(\eta)$ corresponds exactly to an outcome conjunction, so we have:
\[ \smashoperator{\sum_{v\in\textsf{supp}(\eta)}} \eta(v)\cdot \dem{x:=v}{\mu} \quad \vDash \quad \smashoperator{\bigoplus_{v\in\textsf{supp}(\eta)}} (\PP[B_v]=\eta(v)\cdot p) \]

\end{itemize}

\end{proof}

\begin{lemma}[Correctness of Expression Entailment]
If $m\vDash P\oplus Q$ and $P\vDash e$ and $Q\vDash\lnot e$, then $\dem{\assume{e}}{m}\vDash P$ for both the nondeterministic and probabilistic interpretations of expression entailment
\end{lemma}
\begin{proof}\;
\begin{itemize}
\item\textsc{Nondeterminism}. First note that $\dem{\assume{e}}{S} = \{\sigma \mid \sigma\in S, \de{e}_\textsf{Exp}(\sigma) =\tru\}$. In addition, $m\vDash P\oplus Q$ means that there are nonempty sets $S_1$ and $S_2$ such that $S_1\cup S_2= S$ and $S_1\vDash P$ and $S_2\vDash Q$. Depending on which atomic assertions we are using, $P$ is either some assertion $p$ or $(\ok:p)$, in either case, we know from $P\vDash e$ that $p\Rightarrow e$. We know that every state in $S_1$ satisfies $p$ (and therefore also $e$), so $\dem{\assume{e}}{S_1} = S_1$. By a similar argument, $\dem{\assume e}{S_2} = \emptyset$. Therefore $\dem{\assume e}{S} = S_1\cup\emptyset = S_1$ and we already know that $S_1\vDash P$.


\item\textsc{Probabilistic}. The semantics of assume are similar in this case; states are filtered from the support that do not agree with $e$ and the distribution is otherwise left unchanged. Now suppose that $\mu\vDash (\PP[A]=p)\oplus (\PP[B]=q)$ and therefore $\mu_1\vDash (\PP[A]=p)$ and $\mu_2\vDash (\PP[B]=q)$ such that $\mu_1+\mu_2=\mu$. All states in $\textsf{supp}(\mu_1)$ satisfy $A$ (and therefore also $e$ since $A\vDash e$), so $\dem{\assume e}{\mu_1} = \mu_1$. The opposite is true for $\mu_2$, so $\dem{\assume e}{\mu_2}=\varnothing$. Therefore $\dem{\assume e}{\mu} = \mu_1+\varnothing = \mu_1$ and we already know that $\mu_1\vDash (\PP[A]=p)$.

\end{itemize}
\end{proof}

\section{Additional Rules for Conditionals and Loops}
\label{sec:extra-rules}

As mentioned in \Cref{sec:rules}, fully generic looping rules for OL that work with all instances of the logic are not possible because different instances have different constraints when it comes to termination. However, it is possible to create an \emph{under-approximate} rule that unrolls a loop for a bounded number of iterations:
\[\small
\inferrule{
  \forall i.\ P_i\vDash e
  \\
  \forall i.\ Q_i\vDash \lnot e
  \\
  \forall i.\ \triple{P_i}{C}{P_{i+1} \oplus Q_{i+1}}
}
{\triple{P_0 \oplus Q_0}{\whl eC}{(\bigoplus_{i=0}^nQ_i) \oplus\top}}
{\textsc{Bounded Unrolling}}
\]
In this rule, $P_i$ is the outcome of running $C$ $i$ times with the guard remaining true and similarly $Q_i$ is the outcome of running $C$ $i$ times with the guard becoming false. The true components ($P_i$) are passed forward into the next iteration whereas the false components that cause the loop to exit ($Q_i$) are joined to the postcondition.

This rule avoids the termination question entirely by only looking at finite executions, the remainder of outcomes are covered by $\top$. Similar to the conditional rules seen in \Cref{fig:baserules}, it requires you to separate assertions into components that are ``true'' and ``false'' with respect to the loop guard $e$. This may not be possible in nondeterministic settings, as the loop body $C$ may only produce one outcome. It is, however, suitable for probabilistic applications so long as the probability of the loop guard is known (probabilistic assertions can always be split by probability mass).

For nondeterministic proof systems where we may not be able to split assertions into multiple outcomes, we can create specialized loops rules. Such a rule for the separation logic proof system is given below:
\[\small
\inferrule{
\forall i>0.~(p_i\Rightarrow e)\land \epsilon_i=\ok
\\
\forall i\in\mathbb{N}.~\triple{\epsilon_{i+1}:p_{i+1}}C{\epsilon_i:p_i}
\\
(p_0\Rightarrow \lnot e) \vee (\epsilon_0 = \er)
}
{\triple{\epsilon_n: p_n}{\whl eC}{\epsilon_0:p_0}}
{\textsf{While}}
\]
This rule is very similar to the rule for loops from Total Hoare Logic~\cite{variant} with the addition that the postcondition may not imply that $e$ is false if the program has crashed.
Similarly, we can formulate the familiar conditional rule that operates within a single outcome:
\[\small
\inferrule{
  \triple{\ok:p\land e}{C_1}{Q}
  \\
  \triple{\ok:p\land \lnot e}{C_2}{Q}
}
{\triple{\ok:p}{\iftf{e}{C_1}{C_2}}{Q}}
{\textsc{If (Single Outcome)}}
\]